\documentclass[conference]{IEEEtran}
\IEEEoverridecommandlockouts

\usepackage[english]{babel}

\usepackage{xspace}
\usepackage{xparse}
\usepackage{xcolor}
\definecolor{catcol}{HTML}{EAF7FB}
\definecolor{concol}{HTML}{FAF5DC}
\definecolor{abscol}{HTML}{E0CFE8}

\usepackage{algorithm,algorithmic}
\newcounter{algline}
\AtBeginEnvironment{algorithmic}{\refstepcounter{algline}}
\makeatletter
\@addtoreset{ALC@unique}{myalg}
\makeatother

\usepackage{graphicx}
\usepackage{quiver}
\usetikzlibrary{arrows.meta}
\usepackage{amsmath,amsthm,bm,thmtools}
\usepackage{amssymb,amsfonts}
\usepackage{stmaryrd}
\usepackage{mathtools}
\usepackage{cite}
\newcommand{\noadjust}{\hspace{1sp}} %
\usepackage{algorithmic}
\usepackage{graphicx}
\usepackage{textcomp}
\usepackage{subcaption}
\usepackage{wrapfig}

\usepackage[colorlinks=true, allcolors=blue]{hyperref}
\usepackage{cleveref} %

\usepackage{enumitem}

\newtheorem{theorem}{Theorem}
\newtheorem{proposition}[theorem]{Proposition}
\newtheorem{corollary}[theorem]{Corollary}
\newtheorem{lemma}[theorem]{Lemma}
\theoremstyle{definition}
\newtheorem{definition}[theorem]{Definition}
\newtheorem{example}[theorem]{Example}
\newtheorem{remark}[theorem]{Remark}
\newtheorem{assumption}[theorem]{Assumption}
\newtheorem{questions}[theorem]{Questions}
\newtheorem{construction}[theorem]{Construction}
\newtheorem{fact}[theorem]{Fact}

\Crefname{assumption}{Assumption}{Assumptions}
\crefname{assumption}{assumption}{assumptions}
\Crefname{construction}{Construction}{Constructions}
\crefname{construction}{construction}{constructions}
\Crefname{fact}{Fact}{Facts}
\crefname{fact}{fact}{facts}

\newcommand{\mahsa}[1]{\color{green!50!black}$\clubsuit$ Mahsa: #1 \color{black}\xspace}
\newcommand{\sam}[1]{\color{blue} Sam: #1 \color{black}\xspace}
\newcommand{\quentin}[1]{\color{teal} #1 (Quentin) \color{black}\xspace}
\newcommand{\daniela}[1]{\color{magenta} Daniela: #1 \color{black}\xspace}

\Crefname{ALC@unique}{Line}{Lines}
\newcounter{myalg}
\AtBeginEnvironment{algorithmic}{\refstepcounter{myalg}}
\makeatletter
\@addtoreset{ALC@unique}{myalg}
\makeatother

\renewcommand{\algorithmiccomment}[1]{\texttt{\small // \textcolor{blue}{#1}}}
\newcommand{\MATCH}[1]{\STATE \textbf{match} {#1} \textbf{with} \begin{ALC@g}}
\newcommand{\CASE}[2][]{\STATE \textbf{some} {#2}\textbf{:} \ifx\\#1\\\else

\hspace{1em}\algorithmiccomment{#1}\fi \begin{ALC@g}}
\newcommand{\ENDCASE}{\end{ALC@g}}
\newcommand{\ENDMATCH}{\end{ALC@g}}
\newcommand{\IFF}[2]{\STATE \algorithmicif{} {#1} \algorithmicthen{} {#2}}
\makeatletter
\newenvironment{problem}[1][htb]{%
  \renewcommand{\ALG@name}{Problem}%
  \begin{algorithm}[#1]%
  }{\end{algorithm}
}
\makeatother
\Crefname{problem}{Problem}{Problems}
\crefname{problem}{problem}{problems}

\newcommand{\NN}{\mathbb{N}}
\newcommand{\ZZ}{\mathbb{Z}}
\newcommand{\QQ}{\mathbb{Q}}
\newcommand{\KK}{K}
\newcommand{\OK}{\mathcal{O}_{\KK}}

\newcommand{\CC}{\mathbb{C}}

\newcommand{\adjoined}[1]{\ZZ\mathopen{}\left[#1\right]\mathclose{}}
\newcommand{\adjoinedff}[1]{\QQ\mathopen{}\left(#1\right)\mathclose{}}

\newcommand{\cat}[1]{\mathbf{#1}}
\newcommand{\C}{\cat{C}}
\newcommand{\D}{\cat{D}}
\newcommand{\I}{\cat{I}}
\renewcommand{\O}{\cat{O}}
\newcommand{\In}{\textsf{in}}
\newcommand{\St}{\textsf{st}}
\newcommand{\Out}{\textsf{out}}
\newcommand{\Auto}[1]{\cat{Auto}\mathopen{}\left({#1}\right)\mathclose{}}
\DeclareMathOperator{\Min}{Min}
\DeclareMathOperator{\Obs}{Obs}
\DeclareMathOperator{\Reach}{Reach}
\newcommand{\Set}{\cat{Set}}

\newcommand{\Top}{\cat{Top}}
\newcommand{\Meas}{\cat{Meas}}
\newcommand{\KVec}[1]{\cat{Vec}_{#1}}
\renewcommand{\Vec}[1]{\KVec{#1}}
\newcommand{\Mod}[1]{\cat{Mod}_{#1}}
\newcommand{\TFMod}[1]{\Mod{#1}^{\mathit{tf}}}
\newcommand{\FreeMod}[1]{\cat{FreeMod}_{#1}}
\newcommand{\klto}{\mathrel{\ooalign{\hfil$\mapstochar\mkern5mu$\hfil\cr$\to$\cr}}}

\newcommand{\func}[1]{\mathcal{#1}}
\newcommand{\A}{\func{A}}
\newcommand{\Ainit}{\A_{\mathit{init}}}
\newcommand{\Afinal}{\A_{\mathit{final}}}
\newcommand{\B}{\func{B}}
\newcommand{\F}{\func{F}}
\newcommand{\liftF}{\overline{\func{F}}}
\newcommand{\G}{\func{G}}
\newcommand{\Fflat}{\func{F}_{\flat}}

\renewcommand{\L}{\func{L}}
\newcommand{\sem}[1]{\left\llbracket #1 \right\rrbracket}  %

\newcommand{\flip}{\rotatebox[origin=c]{180}{!}}
\newcommand{\unit}{\eta}
\newcommand{\counit}{\varepsilon}

\DeclareMathOperator{\dom}{dom}
\DeclareMathOperator{\cod}{cod}
\DeclareMathOperator{\ob}{ob}
\DeclareMathOperator{\mor}{mor}
\newcommand{\id}{1}
\newcommand{\FreeCat}[1]{\mathsf{Free}(#1)}
\newcommand{\Fun}{\cat{Fun}}

\newcommand{\E}{\mathcal{E}}
\newcommand{\M}{\mathcal{M}}
\newcommand{\im}{\mathrm{im}}
\newcommand{\Surj}{\mathrm{Surj}}
\newcommand{\Inj}{\mathrm{Inj}}

\newcommand{\powser}[2]{#1\langle\!\langle #2 \rangle\!\rangle}
\newcommand{\powserw}[1]{\powser{#1}{\Sigma^*}}
\newcommand{\frMod}[2]{{#1}^{#2}}

\DeclareMathOperator{\sat}{sat}
\newcommand{\localize}[1][R]{{#1}^{-1}}
\newcommand{\norm}[1]{\mathcal{N}(#1)}
\newcommand{\inver}[1]{#1^{\times}}
\DeclareMathSymbol{:}{\mathpunct}{operators}{"3A} %
\newcommand{\isdef}{\stackrel{\mathrm{def}}{=}}
\newcommand{\oracle}[1]{\textsc{Oracle}_{#1}}
\newcommand{\wordImage}[2]{{#1}(\triangleright {#2})}

\usepackage{tcolorbox}

\NewDocumentCommand{\framecolorbox}{oommm}
 {%
  \IfValueTF{#1}
   {%
    \IfValueTF{#2}
     {\fcolorbox{#3}{#4}{\makebox[#1][#2]{#5}}}
     {\fcolorbox{#3}{#4}{\makebox[#1]{#5}}}%
   }
   {\fcolorbox{#3}{#4}{#5}}%
 }

 \def\mahsa#1{{\color{green!40!black}Mahsa: #1\color{black}}} %

 \newcommand{\yellowblock}[1]{
\noindent\fcolorbox{white}{concol}{
\begin{minipage}{0.48\textwidth}
#1
\end{minipage}}~~~~~
}
\newcommand{\blueblock}[1]{
\fcolorbox{white}{catcol}{
\begin{minipage}{0.48\textwidth}
#1
\end{minipage}}
}

\newcommand{\purpleblock}[1]{
\fcolorbox{white}{abscol}{
\begin{minipage}{0.48\textwidth}
#1
\end{minipage}}
}

\makeatletter
\tikzcdset{
  eq node/.style={
    commutative diagrams/math mode=false, anchor=center},
  eq/.style={
    phantom,
    /tikz/every to/.append style={
      edge node={node[commutative diagrams/eq node]
        {\@eqnswtrue\make@display@tag\ltx@label{#1}}}}}}
\makeatother

\usepackage{soul}

\tcbuselibrary{breakable}

\newcommand{\coloredblock}[2]{
\noindent\begin{tcolorbox}[colback=#1,%
                  colframe=white,%
                  width=0.5\textwidth,%
                  boxrule=0pt,
                  breakable,
                  oversize,
                  left=0mm,
                  right=0mm,
                  top=0mm,
                  bottom=0mm,
                  beforeafter skip balanced=0.2\baselineskip,
            ]
#2
\end{tcolorbox}%
}
\newtcbox{\mybox}[1]{on line,colback=#1,colframe=#1,boxsep=0pt,boxrule=0pt,left=2pt,right=2pt,top=2pt,bottom=1pt,}
\newcommand{\phantomspacing}{\vphantom{Aq}}
\newcommand{\inlineblock}[2]{\mybox{#1}{\phantomspacing{}#2}}

\renewcommand{\yellowblock}[1]{\coloredblock{concol}{#1}}
\renewcommand{\blueblock}[1]{\coloredblock{catcol}{#1}}
\renewcommand{\purpleblock}[1]{\coloredblock{abscol}{#1}}

\usepackage{orcidlink}
    
\begin{document}

\title{Learning Weighted Automata over Number Rings, Concretely and Categorically }%

\author{
\IEEEauthorblockN{Quentin Aristote\IEEEauthorrefmark{1}\orcidlink{0009-0001-4061-7553}, Sam van Gool\IEEEauthorrefmark{2}, Daniela Petri\c{s}an\IEEEauthorrefmark{1}\orcidlink{0000-0001-9712-930X} and Mahsa Shirmohammadi\IEEEauthorrefmark{2}\orcidlink{0000-0002-7779-2339} 
}
\IEEEauthorblockA{\IEEEauthorrefmark{1}Université Paris-Cité, CNRS, Inria, IRIF \qquad \IEEEauthorrefmark{2}Université Paris-Cité, CNRS, IRIF \\ F-75013, Paris, France}
}

\newif\ifhidecomments
\hidecommentstrue
\ifhidecomments
  \renewcommand{\sam}[1]{}
  \renewcommand{\mahsa}[1]{}
  \renewcommand{\daniela}[1]{}
  \renewcommand{\quentin}[1]{}
\fi

\maketitle

\begin{abstract}
We develop a generic reduction procedure for active learning problems. Our approach is inspired by a recent polynomial-time reduction  of the exact learning problem for weighted automata over integers to that for weighted automata over rationals (Buna-Marginean et al. 2024). Our procedure improves the efficiency of a category-theoretic automata learning algorithm, and poses new questions about the complexity of its implementation when instantiated to concrete categories. 

As our second main contribution, we address these complexity aspects in the concrete setting of learning weighted automata over number rings, that is, rings of integers in an algebraic number field. Assuming a full representation of a number ring $\OK$, we obtain an exact learning algorithm of $\OK$-weighted automata that runs in polynomial time in the size of the target automaton, the logarithm of the length of the longest counterexample, the degree of the number field, and the logarithm of its discriminant. Our algorithm produces an automaton that has at most one more state than the minimal one, and we prove that doing better requires solving the principal ideal problem, for which the best currently known algorithm is in quantum polynomial time.
\end{abstract}

\begin{IEEEkeywords}
active learning, weighted automata, number rings, category theory, functorial automata, Fatou extension, computational algebraic number theory.
\end{IEEEkeywords}

\section{Introduction}
\label{sec:intro}

A celebrated result in computational learning theory is Angluin's algorithm for  learning deterministic finite-state automata (DFAs) within the minimally adequate teacher (MAT) framework~\cite{Angluin87:Learning-regular-sets-from-queries-and-counterexamples}.
Angluin's algorithm, also known as the $\mathsf{L}^*$ algorithm,  learns DFAs for unknown target regular languages by iteratively interacting with an oracle using  \emph{membership} and \emph{equivalence} queries.
    A membership query asks whether a specific word belongs to the target language, whereas an equivalence query asks whether a hypothesis automaton recognizes the target language. In response to an equivalence query, the oracle either accepts the hypothesis as correct or provides a counterexample: a word on which the target language disagrees with the hypothesis language.    
In the MAT model,  a given learning algorithm is \emph{efficient} if its running
time is polynomial in the minimal representation of the target concept and the length of the longest
counterexample. The running time is an upper bound on the query complexity, that is, the total number of membership and equivalence queries. 
In contrast to the computational intractability of the problem of \emph{passively} learning DFAs~\cite{valiant1984},  $\mathsf{L}^*$  is an efficient procedure. %

Extensions of Angluin-style model learning have been applied for uncovering errors in software and hardware systems, including bank cards, network  protocols, and legacy software~\cite{vaandrager2017model,PeledVY02,higuera2010,BergGJLRS05,ChaluparPPR14,Leucker2006}. Such  
 applications are constrained by  limitations on the expressiveness of learning procedures. 
 This has motivated various extensions of the algorithm to more expressive models, such as tree, timed, register and B\"uchi automata,  weighted automata  over  fields, and sequential deterministic transducers~\cite{MarusicW15,ShahbazG09:Inferring-Mealy-Machines, BolligHKL09:Angluin-style-learning-of-NFA, HowarSJC:Inferring-canonical-register-automata,AartsFKV15:Learning-Register-Automata-with-Fresh-Value-Generation,
 Vilar96:Query-learning-of-subsequential-transducers,MoermanSSKS17:Learning-nominal-automata,AngluinF16:Learning-regular-omega-languages,BeimelBBKV00,buna-margineanLearningPolynomialRecursive2024}. 
Motivated by  structural similarities in these extensions, several category-theoretic frameworks for learning have been  introduced~\cite{UrbatS20:Automata-Learning:-An-Algebraic-Approach, HeerdtSS20:Learning-Automata-with-Side-Effects, BarloccoKR19:Coalgebra-Learning-via-Duality, colcombetLearningAutomataTransducers2021}.

This  paper  combines  computational and categorical aspects of automata  learning. 
Our main focus is the exact learning of $R$-valued rational functions $L\,:\, \Sigma^* \to R$, that is, those computable by finite automata weighted over some ring~$R$. This generalizes the classical efficient  algorithm of~\cite{BeimelBBKV00} for learning weighted automata over fields, itself one of the most well-known extensions of~$\mathsf{L}^*$.
In this setting, a membership query is adapted so that the oracle  provides the value $L(w)$ for a given word $w \in \Sigma^*$.

The work of~\cite{vanheerdtLearningWeightedAutomata2020} extends~\cite{BeimelBBKV00} by developing an active learning procedure for weighted automata over (semi)rings satisfying certain computability assumptions, including in particular all 
 principal ideal domains (PIDs). For PIDs, the termination of the proposed procedure relies on Noetherian properties, which does not yield a bound on its computational and query complexity.

A more recent work~\cite{buna-margineanLearningPolynomialRecursive2024} reexamines the problem
 of learning over PIDs and reduces it to learning over 
 their fields of fractions. The reduction places the 
 problem  in polynomial time for efficiently computable PIDs, %
 such as $\mathbb{Z}$ and $\mathbb{Q}[x]$. 
 In the setting of $\ZZ$, 
 the reduction acts as a \emph{translator} between a learning 
 procedure for~$\QQ$ and the oracle for $\ZZ$-learning, right before any equivalence query. It takes the hypothesis automaton from the learning procedure and, if possible, converts it into an equivalent minimal $\ZZ$-weighted automaton, as the %
 oracle only accepts such inputs. The translator then returns either the oracle's response if there is an equivalent $\ZZ$-weighted automaton or a counterexample showing the hypothesis is not $\ZZ$-valued, and thus not equivalent to the target.
This reduction also provides an 
efficient construction witnessing that $\ZZ$ is a strong Fatou ring. Here, a ring $R$ with field of fractions $\KK$ is \emph{weak Fatou} if any $R$-weighted  function computable by a finite $\KK$-weighted automaton is also computable by a finite $R$-weighted automaton;  
it is \emph{strong Fatou} if, additionally,  the canonical \emph{minimal} $R$-weighted automaton  has the same number of states as the equivalent minimal $\KK$-weighted automaton~\cite[Ch.~7]{berstelNoncommutativeRationalSeries2010}.

Questions around the Fatou property of variants of weighted automata are  discussed, for example, in the  monographs~\cite{salomaa2012automata,berstelNoncommutativeRationalSeries2010,martin2004grammars}, and in the seminal 1971 textbook of Paz~\cite{paz2014introduction}, who posed the question whether or not the minimal probabilistic automaton computing 
a $\QQ$-valued rational function can be defined over~$\QQ$. 
This question was answered  negatively by a  counterexample constructed in~\cite{ChistikovKMSW16,ChistikovKMSW17}, 
showing that such minimal probabilistic automata
require real algebraic numbers in transition entries.

A closely-related  question is the characterization of minimal weighted automata over number rings, where the transition entries are \emph{algebraic integers}, arguably the most   natural generalization of automata weighted over integers.
A first question is: \emph{are rings of algebraic integers always strong Fatou?} The following example shows that this is not the case; a detailed proof is in~\Cref{app:details-for-intro}.

 \begin{restatable}{example}{exNonFatouNumberRing}
  \label{ex:non-Fatou-number-ring}
  Consider the number ring $R = \adjoined{i\sqrt{5}}$ and its field of fractions $\KK = \adjoinedff{i\sqrt{5}}$.
  The $3$-state $R$-weighted automaton over the alphabet $\{a,b\}$ given in~\Cref{fig:example-wa} is state-minimal but has an equivalent $2$-state $\KK$-weighted automaton.
  \end{restatable}

\begin{figure}[h]
    \centering
    \vspace{-10pt}
    \begin{subfigure}{\linewidth}
        \centering
        \includegraphics[]{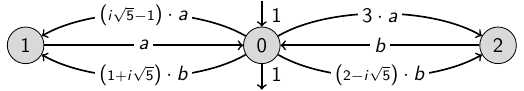}
        \caption{A minimal $\adjoined{i\sqrt{5}}$-weighted automaton.}
        \label{fig:example-wa:ok}
    \end{subfigure}
    
    \medskip
    
    \begin{subfigure}{\linewidth}
        \centering
        \includegraphics[]{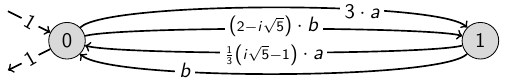}   
        \caption{A minimal $\adjoinedff{i\sqrt{5}}$-weighted automaton.}
        \label{fig:example-wa:k}
    \end{subfigure}
    \caption{Two equivalent $\adjoinedff{i\sqrt{5}}$-weighted automata.}
    \label{fig:example-wa}
\end{figure}

In this paper, we show  that the 
strong Fatou property 
\emph{almost} holds:  the minimal 
automaton over a number ring has at most one more state than  a minimal automaton over the corresponding  number  field. 
This follows from one of our main results: an efficient  exact learning and almost-minimization  algorithm for weighted automata over number rings.

The problem of designing learning algorithms for $R$-weighted automata beyond the case when $R$ is a field can be approached from a completely different angle, namely, by instantiating existing generic learning algorithms such as those provided in~\cite{UrbatS20:Automata-Learning:-An-Algebraic-Approach, colcombetLearningAutomataTransducers2021}. For example, \cite{colcombetLearningAutomataTransducers2021} builds on a view of automata as machines processing some input (such as words over a finite alphabet) and producing some effect (such as values in a semiring). 
This view can be formalized category-theoretically %
by seeing automata as functors that are valued in an output category, which is chosen to model a particular effect~\cite{colcombetpetrisan2017}. For instance, in the case of $R$-weighted automata, the output category is that of free $R$-modules. \cite{colcombetLearningAutomataTransducers2021} gives a generic learning algorithm $\mathsf{FunL}^*$ which is correct and terminates whenever the output category satisfies certain assumptions. These assumptions hold in particular for the category of vector spaces over a field $\KK$, and in this situation the  $\mathsf{FunL}^*$ algorithm instantiates to the existing learning algorithm of $\KK$-weighted automata~\cite{bergadanoLearningBehaviorsAutomata1994}. 

In the case of an arbitrary ring $R$, the category of free $R$-modules typically does not  satisfy the assumptions required by $\mathsf{FunL}^*$. This can be partially remedied by moving to a larger category: the category of all $R$-modules. This leads to the definition of a more general notion of automata, which we call \emph{$R$-modular automata} in this paper, to which $\mathsf{FunL}^*$ readily applies. Even if we obtain a formal minimal $R$-modular automaton in this way, a new difficulty arises: this automaton does not always correspond to an actual $R$-weighted automaton. Another drawback of the generic approach is its complexity. Even if we manage to move from an $R$-modular to an $R$-weighted automaton, $\mathsf{FunL}^*$ might run in exponential time. The concrete representations of the modules involved and the ensuing complexity analysis require additional care. 

In this paper we strive to reconcile the abstract and the concrete perspectives. We next summarize our main contributions.
 
\subsection{Overview of Main Results}
\label{sec:overview}

We briefly introduce the necessary notions required to present a high-level overview of our contributions. 
A complex number $\alpha$ is \emph{algebraic} if it is the root of a  non-zero polynomial in~$\QQ[x]$. 
The \emph{minimal polynomial} of~$\alpha$ %
is the monic polynomial of least degree in~$\QQ[x]$ that has $\alpha$ as a root. The number~$\alpha$ is called an \emph{algebraic integer}
if its minimal polynomial is in $\ZZ[x]$. 
A \emph{number field}~$\KK$  is a finite degree field extension of~$\QQ$, which, by the primitive element theorem, can be obtained by adjoining some algebraic number~$\alpha$ to $\QQ$. The \emph{degree} of $\KK$ is the degree of the minimal polynomial of $\alpha$.
A \emph{number ring} is the subring  of a number field comprised of all its algebraic integers. We denote by~$\OK$ the number ring of $\KK$. 

The main algorithmic question we address is:
\emph{Can a target $\OK$-automaton be learned in polynomial-time?}
 This question is much more challenging than the analogous one for $\ZZ$-weighted automata, due to the structural differences between $\ZZ$ and $\OK$. For instance, in contrast to $\ZZ$, $\OK$  is not in general a principal ideal domain, but only a Dedekind domain: while  a $\ZZ$-submodule of $\ZZ^n$ always has a basis of size at most $n$,  an $\OK$-submodule of $\OK^n$ may not have a basis. 
  In particular, there might be no basis for the usual forward or backward modules  explored while learning.  This  underlines one of the primary difficulties in developing learning procedures over~$\OK$. 
  We can nevertheless prove:

\begin{restatable}{theorem}{theoM}
\label{theo:OK}
Given a full representation of $\OK$, exact learning  of $\OK$-weighted automata is within polynomial time in the size of the target automaton, the logarithm of the length of the longest counterexample, the degree of $\KK$ and the logarithm of its discriminant.   
\end{restatable}

Our learning algorithm for $\OK$-weighted automata relies on a reduction procedure for $\KK$-weighted automata akin to~\cite[Algorithm~6]{buna-margineanLearningPolynomialRecursive2024}, given in \Cref{alg!make-automaton-integral}.
The $\OK$-automaton obtained by  our learning algorithm may not be minimal, but is \emph{almost-minimal}: it has at most one more state than the minimal one.
Furthermore, we show in \Cref{{prop:minimal-pip-hard}}  that deciding minimality of $\OK$-weighted automata allows us to decide the \emph{principal ideal problem}, which has, to the best of our knowledge, \emph{quantum polynomial time complexity}~\cite{biasseEfficientQuantumAlgorithms2015}.

\looseness=-1
The proof of correctness of \Cref{alg!make-automaton-integral} can be obtained at a concrete level, but it also relies on some more abstract principles. Another contribution of this paper is  the generic \Cref{alg:reduction} that generalizes the reduction procedure of $\QQ$-weighted automata to $\ZZ$-weighted automata~\cite{buna-margineanLearningPolynomialRecursive2024}.  Abstracting away from the concrete weighted automata, we can work at more general level with automata represented as functors, following~\cite{colcombetpetrisan2017}. We consider automata valued either in a category $\C$ (that will be instantiated to a category of $R$-modules), or in a category $\D$ (that will be instantiated to the category of $K$-vector spaces). In practice, we want to consider categories $\C$ and $\D$ such that computations are easier in $\D$. The aim of \Cref{alg:reduction} is to transform a $\D$-valued automaton $\A$ into a minimal $\C$-valued one whenever the language accepted by $\A$ factors through $\C$, or otherwise provides some word that witnesses this is not the case. Assuming $\C$ and $\D$ are related by a functor satisfying some reasonable assumptions, we prove in \Cref{thm:alg:reduction:correctness} that \Cref{alg:reduction} is correct and reduces the problem of learning $\C$-automata to that of learning $\D$-automata.

\subsection{Guideline for reading the paper}
After a commutative algebra primer in \Cref{sec:comm-alg-primer}, we explain the categorical view on minimal weighted automata in \Cref{sec:weighted-cat}.  \Cref{sec:learning-and-reduction,sec:categorical-reduction} apply this to obtain a generic algorithm for reducing learning problems. Finally, in~\Cref{sec:computational}, we obtain our main theorem on number rings.

\looseness=-1
We adopt the following colour convention to help readers with various backgrounds navigate through the paper. We describe  \inlineblock{concol}{in yellow} the concrete concepts on weighted automata, \inlineblock{catcol}{in blue} those pertaining to the relevant categories of modules, and \inlineblock{abscol}{in purple} concepts written in arbitrary categories. At times, the same concept is described in equivalent ways at different levels of abstraction, in order to ease the translation of terminology used across different communities.
Readers only interested in the complexity aspects of our results may first skip the more abstract  paragraphs in \Cref{sec:weighted-cat,sec:learning-and-reduction}, and \Cref{sec:categorical-reduction} altogether.

\section{Primer on commutative algebra}\label{sec:comm-alg-primer}
We recall a few definitions and facts about commutative algebra that we will need throughout this paper. Further details are given in \Cref{app:preliminaries:comm-algebra}; also see, e.g.,~\cite[Ch.~4, 6, 10]{broue2024}.

Throughout this paper, \emph{ring} means commutative ring with unit. The notion of \emph{module} generalizes the idea of a vector space, which has scalars in a field, to having scalars in an arbitrary ring.
\emph{Free} modules are essential for the algebraic analysis of $R$-weighted automata. For any (possibly infinite) set $Q$, a \emph{finite $R$-linear combination} is an expression $\sum_{i=1}^n r_i q_i$, with $r_i \in R$ and $q_i \in Q$ for each $1 \leq i \leq n$; said otherwise, it is a finitely supported function from $Q$ to $R$. The set of finite $R$-linear combinations, equipped with the expected addition and $R$-action, is a first example of a free $R$-module; we denote it by $\frMod{R}{Q}$. The elements $1_R \cdot q$, for $q \in Q$, constitute a \emph{canonical basis} for $\frMod{R}{Q}$; we identify each $q$ with its corresponding basis element. The \emph{rank} of $\frMod{R}{Q}$ is the cardinality of $Q$. 

For any $R$-module $M$ and $(x_q)_{q \in Q}$ a $Q$-indexed set of elements of $M$, there exists a unique morphism $\phi \colon \frMod{R}{Q} \to M$ that sends 
$q$ to $x_q$. 
The family $(x_q)_{q \in Q}$ is \emph{free} in $M$ if $\phi$ is injective, \emph{generating} for $M$ if $\phi$ is surjective, and a \emph{basis} for $M$ if $\phi$ is bijective. The module $M$ is called \emph{finitely generated} if it has a finite generating family and \emph{free} if it has a basis. An element $m \in M$ has \emph{torsion} if there exists $r \neq 0$ such that $rm = 0$. An $R$-module $M$ is \emph{torsion-free} if there are no non-zero torsion elements. Free $R$-modules are torsion-free, but the converse is not true; for example, the ideal $(2, 1 + i\sqrt{5})$ in the ring $\adjoined{i\sqrt{5}}$ is not free. Note that a submodule of a torsion-free module is again torsion-free. 

\looseness=-1
Let $R$ be an integral domain.
The \emph{field of fractions} of $R$ is obtained by adding formal inverses for all non-zero elements of $R$: its elements are 
pairs $r/s$, where $r \in R$ and $s \in \inver{R}$, and $r/s$ is identified with $r'/s'$ when $rs' = r's$; note that it is in particular an $R$-module. Let $\KK$ be the field of fractions of $R$. 
Extending this definition to modules, for any $R$-module $M$, there is a $\KK$-vector space  $\localize M$, called the \emph{localization} or \emph{extension of scalars} of $M$. Its elements are formal fractions $m/r$ of an element $m \in M$ and an element $r \in \inver{R}$, where we identify $m/r$ and $m'/r'$ if there is $s
\in \inver{R}$ such that $sr \cdot m' = sr' \cdot m$. For any $R$-linear map $f: M \to N$, we obtain a $\KK$-linear map $\localize f: \localize M \to \localize N$, which maps $m/r$ to $f(m)/r$.
The \emph{rank} of an $R$-module is the dimension of its localization. 

A \emph{fractional ideal} of a ring $R$ is a non-zero $R$-submodule $\mathfrak{a}$ of $\KK$ such that, for some $r \in R \setminus \{0\}$, the submodule $r \mathfrak{a} \coloneq \{ra \ | \ a \in \mathfrak{a}\}$ is contained in $R$. The \emph{dual} of a fractional ideal $\mathfrak{a}$ is the fractional ideal $\mathfrak{a}^{-1} \coloneq \{x \in \KK \ | \ x \mathfrak{a} \subseteq R\}$, and $\mathfrak{a}$ is \emph{invertible} if $\mathfrak{a}^{-1} \mathfrak{a} = R$. An integral domain $R$ is a \emph{Dedekind domain} if, and only if, every fractional ideal is invertible. It is a \emph{principal ideal domain (PID)} when all its ideals are principal.
Any PID is a Dedekind domain, and any Dedekind domain $R$ is \emph{Noetherian}, meaning that any ideal of $R$ is finitely generated. Moreover, if $R$ is a Dedekind domain with field of fractions $\KK$, then every finitely generated torsion-free $R$-module $M$ is a direct sum of rank $1$ submodules, $M = \bigoplus_{i=1}^n E_i$, where each $E_i$ is of the form $\mathfrak{a}_ie_i$, for some fractional ideal $\mathfrak{a}_i$ of $\KK$ and $e_i \in M$, such that the system $(e_i)_{i=1}^n$ is free, see, e.g., \cite[Cor.~10.2.3]{broue2024}. Such a set of rank $1$ submodules $E_1, \dots, E_n$ is called a \emph{pseudo-basis} for $M$.

\section{Minimal  weighted and modular automata}
\label{sec:weighted-cat}

In this section, we recall the fundamental definitions of weighted automata over a ring. We phrase each definition in two ways: using the more traditional \inlineblock{concol}{automata-theoretic} language, and using a \inlineblock{catcol}{category-theoretic} formulation \cite{colcombetpetrisan2017}. More details on the categorical notions are in \Cref{app:cat}.

\subsection{Basic definitions}

Let $R$ be a ring and $\Sigma$ a finite alphabet.
\yellowblock{%
\phantomspacing{}%
Let $Q$, $Q'$ be  sets. Any linear transformation $f\colon \frMod{R}{Q}\to \frMod{R}{Q'}$ is uniquely determined by the values $(f(q))_{q\in Q}$, and can thus be represented by a matrix of size $|Q|\times|Q'|$, whose entry on row $i$ and column $j$ is the coefficient of $q'_j$ in $f(q_i)$.
\phantomspacing{}}

\blueblock{
\phantomspacing{}Let $\FreeMod{R}$ denote the category whose objects are sets $Q, Q',\ldots$ and whose morphisms from $Q$ to $Q'$, denoted $f\colon Q\klto Q'$, are functions from $Q$ to the free $R$-module $\frMod{R}{Q'}$ on $Q'$. The category $\FreeMod{R}$ is equivalent to the category of free $R$-modules.  It is also known as the \emph{Kleisli category} of the monad $\frMod{R}{-}$ of formal $R$-linear combinations.
\phantomspacing{}}
A row vector $\bm{\alpha}\in \frMod{R}{Q}$ uniquely determines a linear map $R\to \frMod{R}{Q}$, and similarly, a column vector $\bm{\beta}\in \frMod{R}{Q}$ uniquely determines a linear map $\frMod{R}{Q}\to R$. In view of this, and to make the correspondence with the categorical definitions more transparent, we slightly rephrase the usual definition of an $R$-weighted automaton as follows:
\yellowblock{\phantomspacing{}
An \emph{$R$-weighted automaton ($R$-WA)} $\A$ over an alphabet $\Sigma$ is a tuple $(Q,\A(\triangleright),(\A(\sigma))_{\sigma\in\Sigma},\A(\triangleleft))$, where
 \begin{itemize}
    \item $Q$ is a set of \emph{states},
     \item $\A(\triangleright)\colon R \to \frMod{R}{Q}$ is a linear map of \emph{initial weights},
 \item $\A(\sigma)\colon \frMod{R}{Q}
\to \frMod{R}{Q}$ is a linear \emph{transition map} for each letter $\sigma \in \Sigma$, and %
 \item $\A(\triangleleft)\colon \frMod{R}{Q} \to R$ is a linear map of \emph{final weights}.\phantomspacing{} %
 \end{itemize}
}

\blueblock{\phantomspacing{}
Consider the category $\I$ free over the multigraph
\[
\begin{tikzcd}[ampersand replacement = \&]
\In \arrow[r, "\triangleright"] \& \St \arrow[r, "\triangleleft"] \arrow["\sigma (\sigma\in\Sigma)", loop, distance=2em, in=55, out=125] \& \Out\,.
\end{tikzcd}
\] An \emph{$R$-weighted automaton} is a functor $\A\colon\I\to \FreeMod{R}$,  such that $\A(\In)= R$ and $\A(\Out)= R$.
\phantomspacing{}}

The two definitions of $R$-weighted automaton are equivalent.
Indeed, a functor $\A\colon\I\to\FreeMod{R}$ uniquely corresponds to an $R$-weighted automaton, in the sense of the first definition, on the set of states $Q \coloneq \A(\St)$. Since $\I$ is a free category on a multigraph, the functor $\A$ is entirely specified by its values on the morphisms $\triangleright,\triangleleft$ and $\sigma$ for $\sigma\in\Sigma$. The $\FreeMod{R}$-morphism $\A(\triangleright)\colon 1\klto Q$ corresponds to a function $1\to \frMod{R}{Q}$, and therefore to a linear map $R\to \frMod{R}{Q}$ of initial weights, which by a small abuse of notation we still denote $\A(\triangleright)$. The same remark applies to $\A(\sigma)$ and $\A(\triangleleft)$.
\yellowblock{\phantomspacing{}
  Note that the function $\sigma \mapsto \A(\sigma)$ from letters to linear maps
  extends uniquely to a homomorphism from the monoid $\Sigma^*$ of words over
  $\Sigma$ to the monoid of linear maps $\frMod{R}{Q} \to\frMod{R}{Q}$.
  Concretely, for any word $w = \sigma_1 \dots \sigma_n$, we define
\[\A(w) \isdef \A(\sigma_1) \cdots \A(\sigma_n),\]
that is the composition of $n$ linear transformations, starting with
$\A(\sigma_1)$ and ending with $\A(\sigma_n)$. Define, for any $w \in
\Sigma^*$,
\begin{align*}
  &\A(\triangleright w) \isdef \A(w)\circ \A(\triangleright)  \colon R \to \frMod{R}{Q}, \\
  &\A(w \triangleleft) \isdef  \A(\triangleleft) \circ \A(w) \colon \frMod{R}{Q} \to R, \quad \text{ and } \\
  &\A(\triangleright w \triangleleft) \isdef \A(\triangleleft) \circ\A(w) \circ\A(\triangleright)\colon R\to R.
\end{align*}

  For any $w \in \Sigma^*$, we write $\sem{\A}(w) \coloneq \A(\triangleright w
  \triangleleft)(1)$, i.e., the scalar 
  corresponding to the linear map $\A(\triangleright w \triangleleft)\colon R\to
  R$. The function $\sem{\A}\colon\Sigma^*\to R$ is the \emph{$R$-weighted language
    computed by} the automaton $\A$. 
\phantomspacing{}}

\blueblock{\phantomspacing{}
  As part of a functor $\A\colon\I\to\FreeMod{R}$,
  we have, for every $w\in\Sigma^*$, the $\FreeMod{R}$-morphisms
\begin{align*}
  & \A(w) \colon Q\klto Q, \quad  & \A(w\triangleleft) \colon Q\klto 1, & \text{ and } \\
  & \A(\triangleright w) \colon 1\klto Q,  & \A(\triangleright w \triangleleft) \colon 1 \klto 1, & 
\end{align*}
which correspond to the linear maps defined above.

  Consider the full subcategory $\O$ of $\I$ on the objects $\In$ and $\Out$
  and denote $\iota\colon\O\to\I$ the inclusion functor. The \emph{weighted
    language computed by} $\A\colon\I\to\FreeMod{R}$ is the functor
  $\A\circ\iota\colon \O\to\FreeMod{R}$.
Note that the morphisms in $\O$ from $\In$ to $\Out$ are $\triangleright w
\triangleleft$, for $w \in \Sigma^*$, and that $\A\circ\iota(\triangleright
w\triangleleft)\colon 1\klto 1$ corresponds to $\A(\triangleright
w\triangleleft)$ as defined above.
\phantomspacing{}}

We recall the notion of weighted automata morphism. 
\yellowblock{\phantomspacing{}
 Let $\A$ and $\A'$ be $R$-weighted automata with sets of states $Q$ and $Q'$, respectively. A \emph{morphism} from $\A$ to $\A'$ is a linear map $\phi\colon \frMod{R}{Q}\to \frMod{R}{Q'}$ such that  $\A'(\sigma)\phi = \phi \A(\sigma)$ for every $\sigma \in
\Sigma$, $\A'(\triangleright) = \phi\A(\triangleright)$ and
$\A'(\triangleleft)\phi = \A(\triangleleft)$.
\phantomspacing{}}

\blueblock{\phantomspacing{}
A \emph{morphism} between $R$-weighted automata $\A$ and $\A'$ is a natural transformation $\phi\colon \A\to\A'$ such that the $\In$- and $\Out$-components of $\phi$ are identities. 
We also denote by $\phi$ the component of this natural transformation at the object $\St$.
\phantomspacing{}}

The two definitions of morphism coincide:  the preservation of initial and final weights and of the transition maps  amounts to the commutativity of the following diagrams, for $\sigma\in\Sigma$:
\[\begin{tikzcd}[cramped, row sep={.1em}]
	& {\,\A(\St)} & [-11pt]{\A(\St)} & {\A(\St)} & [-11pt]{\,\A(\St)} \\
	R &&&&& R \\
	& {\A'(\St)} & [-11pt]{\A'(\St)} & {\A'(\St)} & [-11pt]{\A'(\St)}
	\arrow["\phi", from=1-2, to=3-2]
	\arrow["{\A(\sigma)}", from=1-3, to=1-4]
	\arrow["\phi"', from=1-3, to=3-3]
	\arrow["\phi", from=1-4, to=3-4]
	\arrow["{\A(\triangleleft)}", pos=0.2, from=1-5, to=2-6]
	\arrow["\phi"', from=1-5, to=3-5]
	\arrow["{\A(\triangleright)}", pos=0.8, from=2-1, to=1-2]
	\arrow["{\A'(\triangleright)}"', pos=0.8, from=2-1, to=3-2]
	\arrow["{\A(\sigma)}"', from=3-3, to=3-4]
	\arrow["{\A(\triangleleft)}"', pos=0.2, from=3-5, to=2-6]
\end{tikzcd}\]
Whenever there is a morphism of automata $\A \to \A'$, $\A$ and $\A'$ are said to be \emph{conjugates of one another}. \emph{Conjugating} $\A$ with states $Q$ by a $Q \times Q'$- (resp. $Q' \times Q$) matrix $M$ means constructing an automaton $\A'$ with states $Q'$ such that $M$ is a morphism of automata $\A \to \A'$ (resp. $\A' \to \A$). 
How to construct such a conjugate, when one exists, is standard (it amounts to solving a system of $R$-linear equations).

{}
\subsection{Minimization of \texorpdfstring{$R$}{R}-modular and \texorpdfstring{$R$}{R}-weighted automata}
\label{sec:mincat}

From the categorical point of view, we think of $R$-weighted automata as functors valued in the category $\FreeMod{R}$ of free $R$-modules. In order to investigate minimization and learning algorithms for such automata categorically, we need at times to expand our view on them  and regard them as functors valued in the category $\Mod{R}$ of \emph{all} $R$-modules. This is because, for an arbitrary commutative ring $R$, the category $\FreeMod{R}$ may not have sufficient structure: products and quotients of free modules are not free in general. We now recall the generic definition of word automata valued in an arbitrary category $\C$ from~\cite{colcombetpetrisan2017}. We will then  instantiate it to the case when $\C$ is $\Mod{R}$ to obtain a generalization of $R$-weighted automata that we call \emph{$R$-modular automata}.
\purpleblock{\phantomspacing{}
Let $\C$ be a category and $I,O$ two objects of $\C$. A \emph{$(\C,I,O)$-automaton} is a functor $\A\colon\I\to\C$ such that $\A(\In)=I$ and $\A(\Out)=O$. A \emph{morphism between  $(\C,I,O)$-automata} $\A$ and $\A'$ is a natural transformation $\alpha\colon\A\to\A'$ such that the components $\alpha_{\In}$ and $\alpha_{\Out}$ are identities.
$(\C,I,O)$-automata and their morphisms form a category denoted by $\Auto{\C,I,O}$.\phantomspacing{}
}
We instantiate this definition for $R$-modules:
\blueblock{
\begin{definition} 
\phantomspacing{}An \emph{$R$-modular automaton} over the alphabet $\Sigma$ is a $(\Mod{R}, R,R)$-automaton. \phantomspacing{}
\end{definition}
}
\yellowblock{
\begin{definition}
    \label{def:conc-R-modular-automata}
\phantomspacing{}An \emph{$R$-modular automaton} over the alphabet $\Sigma$ is a tuple $(M,\A(\triangleright),(\A(\sigma))_{\sigma\in\Sigma},\A(\triangleleft))$ consisting of an arbitrary module $M$ and linear maps $\A(\triangleright)\colon R\to M$, $\A(\sigma)\colon M\to M$, for $\sigma \in\Sigma$ and $\A(\triangleleft)\colon M\to R$.\phantomspacing{}
\end{definition}
}

\begin{remark}
For a commutative ring $R$, the $(\FreeMod{R},R,R)$-automata are the $R$-weighted automata in the usual sense.
$R$-modular automata are similar, the only difference being that the state object is  not necessarily a free $R$-module, and thus cannot be represented via its basis elements.
Notice that when $R=\KK$ is a field, the categories $\FreeMod{\KK}$ and $\Mod{\KK}$ are equivalent, since any vector space is isomorphic to $\frMod{\KK}{X}$, where $X$ is any chosen basis. For this reason, minimization and learning of $\KK$-weighted automata are much simpler than for an arbitrary ring.
\end{remark}

\purpleblock{\phantomspacing{}
  A \emph{$(\C,I,O)$-language} is a functor $\L\colon\O\to\C$ such that
  $\L(\In)=I$ and $\L(\Out)=O$. The \emph{language recognized by a
    $(\C,I,O)$-automaton} $\A$ is the composite $\A\circ\iota$.

For a {$(\C,I,O)$-language} $\L\colon\O\to\C$ we denote by $\Auto{\L}$ the category of automata recognizing $\L$. Notice that,
if $\phi\colon\A\to\A'$ is a morphism of $(\C,I,O)$-automata, then the languages recognized by these automata coincide. Hence the categories $\Auto{{\L}}$ are the connected components of the category $\Auto{\C,I,O}$.
\phantomspacing{}}

The crucial idea behind the functorial view of automata is to view transformations between different kinds of automata as 
\emph{liftings of functors} between the respective output categories.
\purpleblock{
\begin{fact}
\label{lem:lifting-functors-to-automata}
\phantomspacing{}Let $\F: \C \to \D$ be a functor and $I$ and $O$ be two objects of $\C$.
If a $(\C,I,O)$-automaton $\A$ recognizes a $(\C,I,O)$-language $\L$, then $\F \circ \A$ is a $(\D,\F I,\F O)$-automaton recognizing the $(\D,\F I,\F O)$-language $\F \circ \L$. 

This defines a functor $\liftF\colon \Auto{\L}\to\Auto{\F\circ\L}$, the \emph{lifting} of $\F$.\phantomspacing{}
\end{fact}
}

\blueblock{\phantomspacing{}
 An $R$-weighted language $L \colon \Sigma^* \to R$ can be seen as a $(\FreeMod{R},1,1)$-language  or as a $(\Mod{R},R,R)$-language.
This is witnessed by a canonical functor $\FreeMod{R} \to \Mod{R}$ which sends an object $Q$ to the module $R^Q$ freely generated by $Q$. Its lifting embeds the category of $R$-weighted automata accepting $L$ as a subcategory of the category of $R$-modular automata accepting $L$.
\phantomspacing{}}

We now recall the construction of forward and backward modules for $R$-weighted automata, and show how this is an instance of the functorial approach to minimization of~\cite{colcombetpetrisan2017}.
In the rest of this section, let $\A$ be an $R$-WA with states $Q$. %
\yellowblock{\phantomspacing{}
A \emph{configuration} of $\A$ is a row vector $r \in \frMod{R}{Q}$; it is \emph{reachable} if
there is a linear combination $\sum_{i = 1}^k \lambda_i w_i$ of words such
that $r = \sum_{i = 1}^k \lambda_i \A(\triangleright w_i)$. The \emph{forward module of $\A$} is the submodule $\Reach{\A}$ of $\frMod{R}{Q}$ consisting of the reachable configurations in $\A$. %
\phantomspacing{}}

When $R$ is a field, $\Reach{\A}$ is the state space of the \emph{forward conjugate} $\overrightarrow{\A}$ of $\A$, see e.g.~\cite{kieferNotesEquivalenceMinimization2020}. However, in general, the forward module is not necessarily free. This justifies the move to $R$-modular automata. 
\blueblock{\phantomspacing{}
Let $\L$ be the $(\Mod{R},R,R)$-language recognized by $\A$.
The category $\Auto{\L}$ of $R$-modular automata recognizing $\L$ has an
initial object $\Ainit$~\cite[Lemma~3.1]{colcombetpetrisan2017}. We recall an
explicit construction: $\Ainit(\St)$ is  the free
$R$-module on the set $\Sigma^*$, the linear map $\Ainit(\triangleright)$ sends
$1$ to $\varepsilon$, the transition map $\Ainit(a)$ sends the generator $w$ to
$wa$ and the final weights linear map $\Ainit(\triangleleft)$ sends $w$ to
$L(w)$. The unique morphism from $\Ainit$ to $\A$ is given by the linear map
$!_{\A}\colon\Ainit(\St)\to\A(\St)$ that sends a generator $w \in \Sigma^*$ to the
configuration $\A(\triangleright w)\in\A(\St)$.
The \emph{forward module} $\Reach{\A}$ is the 
image of $!_{\A}$, and carries the structure of an $R$-modular automaton, since the concept of image lifts to the category $\Auto{\L}$, see \Cref{app:cat:fac-systems}. 
\phantomspacing{}}

\yellowblock{\phantomspacing{}
Let $\powserw{R}$ denote the module of $R$-weighted languages.
Given a configuration $r \in \frMod{R}{Q}$, the $R$-weighted language $\L(\A)_r$ \emph{observed by} $\A$ \emph{starting from} $r$ is defined
as $\L(\A)_r(v) \coloneq \A(v \triangleleft)(r)$, for $v \in \Sigma^*$. 
The \emph{backward module} of $\A$ is the submodule $\Obs{\A}$ of $\powserw{R}$ consisting of the languages observed by $\A$, and is computed by merging the configurations of $\A$ that yield the same observed language.
\phantomspacing{}}
When $R$ is a field, $\Obs{\A}$ is the state space of the \emph{backward conjugate} $\overleftarrow{\A}$ of $\A$ in the sense of~\cite{kieferNotesEquivalenceMinimization2020}.
In general, $\Obs{\A}$ also carries the structure of an R-modular automaton, and again, this is an instance of~\cite[Section~2.2]{petrisanAutomataMinimizationFunctorial2020}.
\blueblock{\phantomspacing{}
The category $\Auto{\L}$ has a final object $\Afinal$, whose state space is the $R$-module $\powserw{R}$. The linear map $\Afinal(\triangleright)$ sends $1$ to $L$, the transition map $\Afinal(a)$ sends $L' \in \powserw{R}$ to the \emph{derivative} language $a^{-1}L'$ -- sending $u\in\Sigma^*$ to $L'(au)$ -- and the final weights linear map $\Afinal(\triangleleft)$ sends $L'\colon\Sigma^*\to R$ to $L'(\varepsilon)$.

The unique morphism from $\A$ to $\Afinal$ is the linear map $\flip_{\A}\colon\A(\St)\to\powserw{R}$ which sends a configuration $r\in\A(\St)$ to $\L(\A)_r$. The module $\Obs{\A}$ is the image of $\flip_{\A}$. %
\phantomspacing{}}

\looseness=-1
When $R$ is a field, the minimization of $\A$ can be described as the forward conjugate of the backward conjugate of $\A$, see e.g.~\cite[Prop.~3.5]{kieferNotesEquivalenceMinimization2020}. 
The minimization can also be described more abstractly and more generally in the categorical framework, provided that the output category $\C$ carries additional structure:
\purpleblock{
\begin{definition}
\phantomspacing{}A category $\C$ is called \emph{good-for-minimization} when $\C$ has countable powers and countable copowers, and is moreover equipped with a factorization system (see \Cref{app:cat:fac-systems,app:cat:co-limits} for these standard definitions).
\end{definition}
A factorization system generalizes factorization of functions: 
it consists of two classes $\E$ and $\M$ of $\C$-morphisms such that in particular every $f: X \to Y$ factors, uniquely up to unique isomorphism, as  
$\begin{tikzcd}[ampersand replacement=\&, column sep=small]
    X \arrow[r,two heads, "e"] \& Z \arrow[r,tail, "m"] \& Y
\end{tikzcd}$ with $e \in \E$ and $m \in \M$. 
We call both $Z$ and $m: Z \rightarrowtail Y$ the \emph{$(\E,\M)$-image} of $f$.

\begin{construction}\label{con:minimal}
If $\C$ is good-for-minimization and $\L$ is a $(\C,I,O)$-language, then the category $\Auto{\L}$ inherits a factorization system from $\C$, and has initial and final objects $\Ainit$ and $\Afinal$. A minimal automaton $\Min{\L}$ is then obtained as the factorization of the unique morphism  $!_{\Afinal} \colon \Ainit\to\Afinal$, see~\cite[Lem.~3.2]{petrisanAutomataMinimizationFunctorial2020}. Furthermore, for any automaton $\A$ in $\Auto{\L}$ we define $\Reach{\A}$ as the  factorization of the unique morphism $!_{\A} \colon \Ainit\to\A$ and $\Obs{\A}$ as the  factorization of the unique morphism $\flip_{\A} \colon \A\to\Afinal$. Then 
$\Min{\L}$ is isomorphic to $\Reach(\Obs{\A})$ and to $\Obs(\Reach{\A})$~\cite[Lem.~2.3]{petrisanAutomataMinimizationFunctorial2020}.\phantomspacing{}
\end{construction}
}

\blueblock{
\begin{fact}
\label{ex:good-for-minimization-categories}
    \phantomspacing{}The following categories are good-for-minimization: the category $\Mod{R}$ of
    $R$-modules and linear maps, its full subcategory $\TFMod{R}$ on torsion-free modules, and in particular the category $\Vec{\KK}$ of vector spaces over a field $\KK$. For all of these categories, we use the factorization system that factors any linear map as a surjective linear map followed by an injective linear map.  By contrast, on $\FreeMod{R}$ for a general ring $R$, there is no obvious factorization system that  yields a meaningful notion of minimization. \phantomspacing{}
\end{fact}
}

Instantiating \Cref{con:minimal} to $\Mod{R}$-automata we  obtain a  \emph{minimal $R$-modular automaton} $\Min{\L}$ accepting $\L$ as $\Reach(\Obs(\A))$. However, an important caveat is that the state space of $\Min{\L}$ is an $R$-module which is not necessarily free without further assumptions on the ring $R$.
As a consequence, we cannot hope to always obtain a reasonable finite presentation of the minimal $R$-modular automaton.
In the next subsection, we examine additional assumptions on $R$ that do allow for this, and are satisfied for all number rings.

{}

\subsection{Minimal automata over Dedekind domains}
\label{sec:minimal-automata-structure}
Let $R$ be an integral domain, and let $\KK$ be its field of fractions. A first way to represent an $R$-modular automaton $\A$ is as a $\KK$-weighted automaton with state-space $\localize{\A(\St)}$. This includes in particular the classical view of $R$-weighted automata as $\KK$-weighted automata with weights in $R$. 
We define the \emph{rank} of an $R$-modular automaton $\A$ as the rank of the module $\A(\St)$ and the \emph{rank} of a recognizable $R$-weighted language $L$ as the rank of the module $\Min L$.

The $R$-module $\powserw{R}$ is not necessarily free (even when $R = \mathbb{Z}$, see, e.g., \cite{Sch2008}), but $\powserw{R}$ is torsion-free, and thus any submodule of $\powserw{R}$ is again torsion-free. In particular, for any $R$-weighted automaton $\A$, the  backward module $\Obs{\A}$ is torsion-free, and so is its forward submodule $\Reach(\Obs{\A})$. Thus, if $R$ is an integral domain, then the state module of a minimal $R$-modular automaton is always torsion-free. For a rational $R$-weighted language $L$, $\Min L$ is also always finitely generated \cite[Theorem 5.21]{drosteWeightedAutomata2021}.

In particular, if $R$ is a PID, then all finitely generated torsion-free modules are free, so that it is possible to choose a basis for $\Min L$ and the minimal $R$-modular automaton is in fact $R$-weighted. More generally, if $R$ is a \emph{Dedekind domain}, then the finitely generated torsion-free modules, in particular $\Min L$, do not necessarily have a basis, but they have a pseudo-basis. This is key in proving the following proposition, which motivates this work; we give a proof in \Cref{app:proofs-for-weighted-cat}.

\begin{restatable}{proposition}{almostMinDedekindWA}
\label{prop:almost-minimal-dedekind-wa}
    Let $R$ be a Dedekind domain, and let $L$ be an $R$-weighted language of rank $n$ computed by a finite $R$-WA. There exists an $R$-WA computing $L$ with at most $n+1$ states.
\end{restatable}
We call an automaton as in the conclusion of \Cref{prop:almost-minimal-dedekind-wa} \emph{almost minimal}.
An illustration of \Cref{prop:almost-minimal-dedekind-wa} is given in \Cref{ex:non-Fatou-number-ring}. The goal of the present work is to provide an algorithm computing almost-minimal automata when $R$ is a number ring, and thus an  effective Dedekind domain.

\section{Learning problems and reduction procedures}
\label{sec:learning-and-reduction}

We now tackle the problem of computing the minimal automaton recognizing a language, focusing on  \emph{active learning}. %
\yellowblock{\phantomspacing{}
For a fixed ring $R$,
the goal of the active learning algorithm is to compute the minimal $R$-weighted automaton recognizing a certain
$R$-weighted language $L$ with the only help of an oracle able to answer two
kind of queries:
\begin{itemize}%
\item \emph{value queries}: given an input word $w \in \Sigma^*$, the oracle
  returns the value $L(w)$;
\item \emph{equivalence queries}: given an $R$-weighted automaton $\A$, the
  oracle decides whether $\sem{\A} = L$, and, if this is not the case, outputs a
  counterexample input word $w \in \Sigma^*$ such that $\sem{\A}(w) \neq L(w)$.\phantomspacing{}
\end{itemize}
}
\purpleblock{\phantomspacing{}
More generally, the problem of active learning can be stated for
$(\C,I,O)$-automata, as in \Cref{problem:active-learning} below. 
In this setting, the oracle can answer:
\begin{itemize}
    \item \emph{value queries}: given an input word $w \in \Sigma^*$, the oracle
  returns a morphism $\L(\triangleright w
\triangleleft): I \to O$;
    \item \emph{equivalence queries}: the oracle decides whether a
$(\C,I,O)$-automaton $\A$ recognizes the target $(\C,I,O)$-language $\L$, and, if this is not the case outputs a counterexample $w\in\Sigma^*$ such that $\A\circ\iota(\triangleright w \triangleleft)\neq\L(\triangleright w
\triangleleft)$.\phantomspacing{}
\end{itemize}
}

\begin{problem}
    \caption{Active learning of $(\C,I,O)$-automata}
    \label[problem]{problem:active-learning}
    \begin{algorithmic}
        \REQUIRE an oracle able to answer value and equivalence queries for a
        $(\C,I,O)$-language $\L$
        \ENSURE $\Min \L$
    \end{algorithmic}
\end{problem}

 A generic category-theoretic algorithm for solving \Cref{problem:active-learning} was given in~\cite{colcombetLearningAutomataTransducers2021}. It encompasses, among other things,
the active learning algorithm for learning $\KK$-weighted automata~\cite{bergadanoLearningBehaviorsAutomata1994}, which runs in polynomial time.
However, the generality of the algorithm in~\cite{colcombetLearningAutomataTransducers2021} 
comes at a cost. First, it takes for granted
basic operations of $R$ and its modules 
which could be difficult to compute, or even undecidable. 
Second, it does not take
into account possible optimizations for specific choices of $\C$,
$I$ and $O$. 
This is the case, in particular, for $\ZZ$-weighted automata.
On the one hand, by instantiating the generic algorithm of~\cite{colcombetLearningAutomataTransducers2021} to the case of 
$\ZZ$-weighted automata, one obtains an algorithm that runs in
exponential time. On the other hand, \cite{buna-margineanLearningPolynomialRecursive2024}~gives a  refined polynomial time algorithm, tailor-made for $\ZZ$-weighted automata, obtained 
by a clever reduction to the  learning problem for $\QQ$-weighted automata. %

More generally, we consider a Noetherian integral
domain $R$, and we first tackle the active learning problem for $R$-modular automata, which reduces to that of active learning for automata weighted in $R$'s field of fractions, $\KK$.

The key ingredient of this reduction is a procedure \textsc{Transform}
which, given as input a $\KK$-weighted automaton $\A$
recognizing $L \in \powserw{\KK}$, either (a) outputs a
minimal $R$-modular automaton equivalent to $\A$, 
or, otherwise, (b) outputs a counterexample word
$w \in \Sigma^*$ such that $L(w) \notin R$.
Note that, in case (a), $L$ is actually in $\powserw{R}$, while in case (b), no $R$-modular automaton equivalent to $\A$ can exist.
The procedure \textsc{Transform} will be an instance of the more general \Cref{alg:reduction}. Let us describe how it can be used for implementing the reduction between  the  learning problems.

Let $L_R \colon \Sigma^* \to R$ be an $R$-weighted language, and write $L_\KK \colon \Sigma^* \to \KK$ for the composite of $L_R$ with the  inclusion map of $R$ into its field of fractions. Assume that $\oracle{R}$ is an oracle for the active learning of $L_R$. We implement an oracle $\oracle{\KK}$ for the active learning of $L_{\KK}$ as follows.

\begin{itemize}
   \item  \emph{Value queries}
to $\oracle{\KK}$ are transmitted to $\oracle{R}$.
\item Consider a $\KK$-weighted automaton $\A$ submitted for an \emph{equivalence query} to $\oracle{\KK}$. We run 
\textsc{Transform} with input $\A$. If the output is an $R$-modular automaton equivalent to $\A$, this automaton is in turn 
fed to $\oracle{R}$. Otherwise, if %
\textsc{Transform} outputs a word $w$ such that $\sem{\A}(w)
\notin R$, then $\oracle{\KK}$ returns $w$ as a  counterexample. This is correct, since $L_\KK(w) = L_R(w)\in R$, but $\sem{\A}(w) \not\in R$, so $\sem{\A}(w) \neq L_{\KK}(w)$.
\end{itemize}
Once $\oracle{\KK}$ is implemented, we first learn a minimal $\KK$-weighted automaton recognizing $L_\KK$. Since $L_\KK$ factors through $R$, %
we can use %
\textsc{Transform} once more to transform the minimal $\KK$-weighted automaton into an equivalent $R$-modular automaton, which turns out to be minimal.

\section{Generic reduction of learning problems}
\label{sec:categorical-reduction}

In this section, and as our first main contribution, we reconcile the generic and the tailor-made approaches to learning: we describe additional
conditions on the categorical framework that allow for optimizations similar to those 
described in~\cite{buna-margineanLearningPolynomialRecursive2024} for
$\ZZ$-weighted automata.  
Moreover, this allows us to go beyond $\ZZ$-weighted automata: in \Cref{sec:computational} we instantiate the generic \Cref{alg:reduction} to show that active learning of automata
weighted over number rings can be solved in polynomial time.
\blueblock{\phantomspacing{}
To summarize the discussion in \Cref{sec:learning-and-reduction}, the goal is to learn $R$-modular automata. However, switching to $\KK$-weighted automata makes computations easier. The fact that $R$-modular automata can be seen as $\KK$-weighted automata is witnessed at a category-theoretic level, as an instance of \Cref{lem:lifting-functors-to-automata} via  the lifting of the  localization functor $
\localize -: \Mod{R} \to
\Vec{\KK}
$. The action of the latter functor on $R$-modules and $R$-linear maps was described in~\Cref{sec:comm-alg-primer}. %
\phantomspacing{}}

We abstract this even further as follows. 
\purpleblock{
\begin{assumption}
  \label{assumption:functor}
  \phantomspacing{}$\C$ and $\D$ are two  good-for-minimization categories, 
   with respective factorization systems
   $(\E_\C, \M_\C)$ and $(\E_\D, \M_\D)$, $I$ and $O$ two objects of $\C$, 
  and $\F: \C \to \D$ is a functor.\phantomspacing{}
\end{assumption}
}

The intuition behind \Cref{assumption:functor} is that $\C$ is a category where we want to carry out some
computation, but $\D$ is a category where it is easier to compute. Using the lifting of $\F$ (\Cref{lem:lifting-functors-to-automata}), we aim to compute as much as possible in $\D$ before pulling the result back along
$\F$ and finishing the computation in $\C$. 
\Cref{assumption:functor} can be instantiated to our context by taking $\C \coloneq \Mod{R}$, $\D \coloneq \Vec{\KK}$, and $\F \coloneq \localize -$. 
We now ask: when does the problem of learning
minimal $(\C,I,O)$-automata reduce to the problem of learning minimal $(\D,\F I,
\F O)$-automata? Given sufficient assumptions, we will provide such a reduction
by describing a procedure for solving the following problem. %

\begin{problem}
  \caption{%
  Transforming a $(\D,\F I,\F O)$-automaton into a $(\C,I,O)$-automaton}
  \label[problem]{problem:reduction}
  \begin{algorithmic}
    \REQUIRE a $(\D,\F I,\F O)$-automaton recognizing a language $\L$
    \ENSURE \emph{either} $\Min \L'$, where $\L'$ is a $(\C,I,O)$-language such that $\L = \F \circ \L'$, \emph{or} a word $w \in \Sigma^*$ for which there does not exist $f \colon I \to O$ such that $\F f = \L(\triangleright w
    \triangleleft)$.%
  \end{algorithmic}
\end{problem}

In \Cref{assumption:adjunction} below, we will state additional conditions that we put on the functor $\F$ to be
able to solve \Cref{problem:reduction}. Being able to pull back computations from $\D$ to $\C$ intuitively requires  some form of inverse of the functor $\F$. We require the existence a functor $\G$ which is \emph{right adjoint} to $\F$.
Briefly (see \Cref{app:cat:adjunction} for more details), this means that, for any objects $X$ in $\C$ and $Y$ in $\D$, we have a bijection between the sets of morphisms $\D(\F X,Y)$ and $\C(X,\G Y)$, which is moreover natural in both coordinates.
For any object $X$ of $\C$, the \emph{unit} of the adjunction at $X$ is the morphism $\unit_X \colon X \to \G\F X$  corresponding to the identity morphism $\F X \to \F X$. For any object $Y$ of $\D$, the $\emph{counit}$ of the adjunction at $Y$ is the morphism $\counit_Y \colon \F\G Y \to Y$ corresponding to the identity morphism $\G Y \to \G Y$.
\purpleblock{
\begin{restatable}{assumption}{assumptionAdjunction}
  \label{assumption:adjunction}
  \phantomspacing{}Under \Cref{assumption:functor}, assume moreover that
  \begin{enumerate}
  \item \label{assumption:adjunction:factorization-system-preservation} $\F [\E_\C] \subseteq \E_\D$ and $\F^{-1}[\M_\D] = \M_\C$; %
  \item \label{assumption:adjunction:right-adjoint} $\F$ has a right-adjoint $\G: \D \to \C$;
  \item \label{assumption:adjunction:unit} for every object $X$ of $\C$, $\unit_X \colon X \rightarrow \G \F X$ is in $\M_\C$;%
  \item \label{assumption:adjunction:counit} for every object $Y$ of $\D$, $\counit_X \colon \G \F X \rightarrow X$ is in $\M_\D$.\phantomspacing{}
  \end{enumerate}
\end{restatable}
}

In order to satisfy \Cref{assumption:adjunction}, the localization functor  needs to be restricted to the category $\TFMod{R}$ of \emph{torsion-free} $R$-modules and
linear maps between them, which, just as $\Vec{\KK}$, is a good-for-minimization category, see \Cref{ex:good-for-minimization-categories}.
\blueblock{
\begin{restatable}{lemma}{lemExtScalars}
  \label{lemma:extension-of-scalars-left-adjoint}
  \phantomspacing{}The functor $\localize -: \TFMod{R} \to \Vec{\KK}$ satisfies
  \Cref{assumption:adjunction}.\phantomspacing{}
\end{restatable}
}

\begin{remark}
 The conditions in \Cref{assumption:adjunction} are rather mild. For readers aware of fibrations: these conditions are automatically satisfied by functors $\F:
  \C \to \D$ that are opfibrations and for which each fiber along $\F$ has a terminal object. 
  The functor $\localize -$ is an example of one, and such opfibrations also come a dime a dozen in category theory: topological functors
  \cite[\textsection 21]{adamekAbstractConcreteCategories2009}, for
  instance the forgetful functors $\Top \to \Set$ or $\Meas \to \Set$, are a
  prominent example.
  We leave for further work the study of these other concrete examples.
\end{remark}

Before giving the algorithm for solving \Cref{problem:reduction}, let us start
by relating the minimal automata in $\C$ and $\D$. \emph{A priori}, in the
setting of \Cref{assumption:functor}, they could be entirely unrelated, as we
do not impose any relationship between the factorization systems. However, the additional \Cref{assumption:adjunction} ensures that post-composing by $\F$ preserves minimality:
\purpleblock{
\begin{restatable}{proposition}{propMinimalityPreservation}
  \label{prop:minimality-preservation}
  \phantomspacing{}Under \Cref{assumption:adjunction}, let $\L$ be a
  $(\C,I,O)$-language. Then $\F \circ \Min \L \cong \Min( \F \circ \L)$.\phantomspacing{}
\end{restatable}
}

\yellowblock{\phantomspacing{}
What this says in the $R$-weighted setting is that the minimal
$R$-modular automaton recognizing an $R$-weighted language $L_R\colon\Sigma^*\to R$ is also minimal
among all $\KK$-weighted automaton recognizing $L_{\KK}\colon\Sigma^*\to R\hookrightarrow \KK$. This, combined with \Cref{prop:almost-minimal-dedekind-wa} and the fact that N\oe{}therian integral domains are weak Fatou rings, shows that Dedekind domains are what we call \emph{almost-strong} Fatou rings:
\begin{restatable}{corollary}{dedekindAlmostStrongFatou}
    \label{corollary:dedekind-almost-strong-fatou}
    Let $R$ be a Dedekind domain, $\KK$ its field of fractions, and $L$ an $R$-weighted language. If $L$ is computed by a $\KK$-weighted automaton with $n$ states, then it is also computed by an $R$-weighted automaton with $n+1$ states.
\end{restatable}
This suggests a natural question: Are other rings, beyond Dedekind domains, still almost-strong Fatou rings, possibly with a different number of extra states than $1$? We leave this for future work.
\phantomspacing{}}

Notice how any algorithm that solves \Cref{problem:reduction} implements in
particular minimization of $(\C,I,O)$-automata: given a $(\C,I,O)$-automaton
$\A$, starting with $\F \circ \A$ as input $(\D,\F I,\F O)$-automaton, the
output will be the minimal automaton equivalent to $\A$ (because $\F \circ \A$
recognizes a language that factors through $\F$). It is already known that
minimization procedures can be understood categorically~\cite{colcombetpetrisan2017,Adamek-coalg,Rot2016}. In the functorial framework, to minimize a
$(\C,I,O)$-automaton $\A$, one should compute $\A' = \Obs \A$ (intuitively,
merging equivalent states) and then $\Reach(\A')$ (intuitively, restricting to
reachable states).
{ }
\purpleblock{
Recall from \Cref{con:minimal} that, when working in a good-for-minimization category $\C$ equipped with a factorization system $(\E,\M)$, the automaton $\Reach{\A'}$ is defined as the $(\E,\M)$-image of the unique morphism $!_\A'\colon\Ainit\to\A'$. Its state object is -- in view of~\cite[Lemma~3.2]{petrisanAutomataMinimizationFunctorial2020}  -- the image of a unique morphism $[\A'(\triangleright w)]_{w \in \Sigma^*}\colon\coprod_{\Sigma^*} I \to \A'(\St)$ obtained via the universal property of the coproduct of $\Sigma^*$-many copies of $I$. A generic algorithm for computing $\Reach \A'$ %
is explained in
more detail in \cite[\textsection 3]{aristoteFunctorialApproachMinimizing2023b}:
one aims to find a \emph{finite $I$-generating family of words of $\A'(\St)$}, as we define now. 
\begin{definition}\label{dfn:I-generating}
Given $W\subseteq\Sigma^*$ write  $\wordImage{\A'}{W}$  for the $(\E,\M)$-image of the morphism $[\A'(\triangleright w)]_{w \in W}: \coprod_W I\to \A'(\St)$.  A \emph{finite $I$-generating family of words of $\A'(\St)$} is a finite set $W \subseteq \Sigma^*$ such that $(\Reach
\A')(\St) \cong \wordImage{\A'}{W}$.
\end{definition}

The idea of the generic algorithm for computing $\Reach(\A')$ is to construct an increasing sequence of $\M$-subobjects $\A'(\triangleright W_0)\rightarrowtail\A'(\triangleright W_1) \rightarrowtail \ldots \rightarrowtail\Reach(\A')(\St)$ -- starting  with $W_0 = \{ \varepsilon \}$, and, while
there is some $w \in W_i$ and $\sigma \in \Sigma$ that makes this image
strictly increase, one adds $w\sigma$ to $W_{i+1}$. Under the right assumptions, this sequence stabilizes in finitely many steps and yields $\Reach(\A')$.
}

\blueblock{\phantomspacing{}
For $R$-weighted automata, finding an $I$-generating family for the forward module of an automaton $\A'$ means finding a finite set $W$ of
words such that any reachable configuration of $\A'$ is 
a linear combination of configurations reached by following some word in
$W$. This can be achieved  by starting with $W = \{ \varepsilon \}$ and adding
words to $W$ as long as this strictly increases the module of $W$-reachable
configurations in $\A'(\St)$.
\phantomspacing{}}

\yellowblock{\phantomspacing{}
If $R$ is a field, then the algorithm just described is precisely the usual
minimization algorithm for $R$-weighted automata, and it has polynomial-time
complexity in the dimension of the minimized automaton because any strictly
increasing chain of vector spaces $V_0 \subset \cdots \subset V_n = (\Reach
\A')(\St)$ must have length at most the dimension of $(\Reach \A')(\St)$. But if
$R$ is not a field, say if $R = \ZZ$ and $\Reach(\A')(\St) = \ZZ$, then
there exist chains of strictly increasing submodules of $\ZZ$ of arbitrary lengths, despite $\ZZ$ having rank 1. For example, for all $n>0$, we have the chain $0 \subsetneq \langle 2^n \rangle \subsetneq \langle  2^{n-1}
  \rangle \subsetneq \cdots \subsetneq \langle 1 \rangle = \ZZ$. 

To avoid this complexity pitfall,~\cite{buna-margineanLearningPolynomialRecursive2024} develops the following strategy for $\ZZ$-weighted automata: they first
compute a $\QQ$-generating family of words $W$ for the corresponding
$\QQ$-weighted automaton, so as to first fill-up the rank of the module of
$Q$-reachable configurations, and only afterwards, the set $W$ is completed to an 
$R$-generating family.
\phantomspacing{}}

 In the remainder of this section we generalize this idea beyond $\ZZ$-automata to the  abstract  setting of \Cref{assumption:adjunction}.

The functor $\F$ already allows us to transform a $(\C,I,O)$-automaton into a $(\D,\F I,\F O)$-automaton. 
Furthermore, following~\cite[Lemma~3.4]{petrisanAutomataMinimizationFunctorial2020}, the adjoint functors $\F$ and $\G$ can be lifted to adjoint functors $\Fflat$ and $\G^\sharp$ between the categories $\Auto{\C, I, \G\F O}$ and $\Auto{\D,\F I,\F O}$, for which we recall the definitions in \Cref{app:proofs-categorical-reductions}. To understand the gist of \Cref{alg:reduction}, we just recall that given a $(\D,\F I,\F O)$-automaton $\A$, the $(\C,I,\G\F O)$-automaton $\G^\sharp(\A)$ with state object $\G(\A(\St))$ is given below, where $\A(\triangleright)^\sharp=\G\A(\triangleright)\circ\unit_I$ corresponds to $\A(\triangleright)\colon\F I\to \A(\St)$ via the adjunction:\\
\centerline{$
\begin{tikzcd}[ampersand replacement=\&, column sep=large]
  I \& {\G\A(\St)} \& {\G\F O\,.}
    \arrow["{\A(\triangleright)^\sharp}", from=1-1, to=1-2]
    \arrow["{\G\A(a)}"', from=1-2, to=1-2, loop, in=125, out=55, distance=10mm]
    \arrow["{\G\A(\triangleleft)}", from=1-2, to=1-3]
  \end{tikzcd}
$}

The crucial property of the functor $\G^\sharp$ is that it preserves observable automata, that is, automata $\A$ such that $\A\cong\Obs{\A}$, see \Cref{app:proofs-categorical-reductions} for full proofs. For this reason,  $\Reach(\G^\sharp(\Obs{A}))$ is a minimal automaton, which happens to come from a minimal $(\C,I,O)$-automaton whenever the language accepted by $\A$ factors through $\F$.  We can now state the generic \Cref{alg:reduction}, which aims to compute a generating family of words for $\Reach(\G^\sharp(\Obs{A}))$, starting  in $\D$ and completing it in $\C$. The algorithm only
describes the steps to compute the state-space of the minimal automaton. Its transitions are easily inferred  from the generating family of words, as explained, in the categorical setting, in
\cite[Algorithm 2]{aristoteFunctorialApproachMinimizing2023b}.

\begin{algorithm}
  \caption{(within \Cref{assumption:adjunction}) Transforming a
    $(\D,\F I, \F O)$-automaton into a $(\C,I,O)$-one}
  \label{alg:reduction}
  \begin{algorithmic}[1]
    \REQUIRE a $(\D,\F I,\F O)$-automaton $\A$ recognizing $\L$ %
    \ENSURE if $\L = \F \circ \L'$ for some $(\C,I,O)$-language $\L'$:\\ \hspace{28pt} $(\Min
    \L')(\St)$; \\ \hspace{22pt} otherwise: $w \in \Sigma^*$ such that $\L(\triangleright w
    \triangleleft) \notin \F[\C(I,O)]$

    \COMMENT{Merge equivalent states to obtain $\A'$}
    \STATE compute $\A' = \Obs_{(\E_\D,\M_\D)} \A$ \label{alg:reduction:line:obs}

    \COMMENT{Compute an $\F I$-generating family of words for $\A'(\St)$ in $\D$}
    \STATE $W \coloneq \{ \varepsilon \}$ \label{alg:reduction:line:begin_1st_while}

    \WHILE{there is some $(w, \sigma) \in W \times \Sigma$ s.t. the $(\E_\D,\M_\D)$-images 
      $\A'(\triangleright W)$ and $\A'(\triangleright (W
        \cup \{w\sigma\}))$ are not isomorphic}

    \IFF{$\A'(\triangleright w\sigma \triangleleft) \notin \F[\C(I,O)]$}{\algorithmicreturn{} $w \sigma$}

    \STATE $W \coloneq W \cup \{ w\sigma \}$

    \ENDWHILE \label{alg:reduction:line:end_1st_while}

    \COMMENT{Complete $W$ into an $I$-generating family of words for $\G \A'(\St)$ in $\C$}

    \WHILE{there is some $(w, \sigma) \in W \times \Sigma$ s.t.
      the $(\E_\C,\M_\C)$-images 
      $\G^\sharp(\A')(\triangleright W)$ and $\G^\sharp(\A')(\triangleright (W\cup\{w\sigma\}))$ are not isomorphic}
      \label{alg:reduction:line:begin_2nd_while}

    \IFF{$\A'(\triangleright w\sigma \triangleleft) \notin
      \F[\C(I,O)]$}{\algorithmicreturn{} $w\sigma$}

    \STATE $W \coloneq W \cup \{ w\sigma \}$
    \ENDWHILE \label{alg:reduction:line:end_2nd_while}

    \COMMENT{$(\Min \L')(\St)$ is the space of $W$-reachable states of
      $\G\A'(\St)$ in $\C$}
    \RETURN the $(\E_\C,\M_\C)$-image 
    $\G^\sharp(\A')(\triangleright W)$ 
    \label{alg:reduction:line:extract_basis}
  \end{algorithmic}
\end{algorithm}

\begin{restatable}{theorem}{reductionCorrect}  \label{thm:alg:reduction:correctness}
  \Cref{alg:reduction} is correct\footnotemark and reduces the problem of learning minimal
  $(\C,I,O)$-automata to that of learning minimal $(\D,\F I, \F O)$-automata.
\end{restatable}

We briefly describe \Cref{alg:reduction} in the case of $R$-weighted setting, i.e. for $\C = \TFMod{R}$ 
and $\D
= \Mod{\KK}$ both equipped with the (surjections, injections) factorization
system, $\F = \localize -$ and $I = O = R$ so that $\F I = \F O = \KK$. We
focus in particular on the case $R = \ZZ$ that was developed in
\cite{buna-margineanLearningPolynomialRecursive2024}. %

The algorithm starts with a $\KK$-weighted automaton $\A$ (equivalently a
$(\Vec{\KK}, \KK, \KK)$-automaton) recognizing a $\KK$-weighted language $\L$.
On \cref{alg:reduction:line:obs}, its equivalent states are first merged, thus computing its backward conjugate $\A'$ -- a standard procedure whose complexity is  linear in the dimension of $\A'(\St)$, as described, for example, in~\cite{kieferNotesEquivalenceMinimization2020}. 

On
\crefrange{alg:reduction:line:begin_1st_while}{alg:reduction:line:end_1st_while},
we start with $W = \{ \varepsilon \}$ and add $w\sigma$ (with $w \in W$ and
$\sigma \in \Sigma$) to $W$ whenever $\A'(\triangleright w\sigma)$ is not in the
$\KK$-span $\langle \A'(\triangleright w') \mid w' \in W \rangle_\KK$. The number
of words added to $W$ then is bounded by the dimension of the vector space
$(\Min \L)(\St)$. If at any point some $w\sigma$ is such that $\A'(\triangleright
w\sigma \triangleleft) \notin R$, the algorithm stops and outputs $w\sigma$ as a
counterexample as $\L(\triangleright w\sigma \triangleleft) \notin R$.

\footnotetext{We do not mention termination here but it could also be encompassed by the categorical framework, in a similar fashion to what is done in \cite{colcombetLearningAutomataTransducers2021}.}

On
\crefrange{alg:reduction:line:begin_2nd_while}{alg:reduction:line:end_2nd_while},
we continue expanding the set $W$ obtained above by adding $w \sigma$ to $W$ whenever
$\A'(\triangleright w\sigma)$ is not in the \emph{$R$-span} $\langle \A'(\triangleright
w') \mid w' \in W \rangle_R$. It is  not obvious how this can actually
be checked: for $R = \ZZ$, the authors of
\cite{buna-margineanLearningPolynomialRecursive2024} make use of the so-called \emph{Smith
  Normal Form} of a matrix, which can be computed in polynomial time. Again, any
counterexample to $\L$ actually being $R$-weighted interrupts the algorithm.
The number of words added to $W$ in the second while loop is now bounded by the
maximum length of certain strict chains of submodules $M_0 \subset \cdots
\subset M_n$, where $M_0$ and $M_n$ are fixed and of rank the dimension of
$(\Min \L)(\St)$. For $R = \ZZ$, this maximum length is bounded polynomially 
in this rank and the bit-size of the encoding of $\A$
\cite{buna-margineanLearningPolynomialRecursive2024}. Requiring more generally that $R$ be Noetherian -- this is the case in particular of PIDs and Dedekind domains -- ensures that such chains will at least always be finite when $\L'$ is computed by a $\KK$-weighted automaton, and thus that the algorithm will terminate.

Finally  \cref{alg:reduction:line:extract_basis} is reached if and only if
$\L = \localize \L'$ for some $R$-weighted language $\L'$. In this case, a
representation of $(\Min \L')(\St)$ is extracted from the generating family $W$.
What this representation will actually be varies depending on $R$: for $R =
\ZZ$ and when $\L'$ is computed by a $\QQ$-WA, $(\Min \L')(\St)$ has a basis which can be computed using the
\emph{Smith Normal Form} \cite{buna-margineanLearningPolynomialRecursive2024}.
But for other rings $R$ such a basis may not exist. 
For $R$ a Dedekind domain we have mentioned in \Cref{sec:minimal-automata-structure} that when $(\Min \L')(\St)$ is finitely-generated with rank $n$, it has a pseudo-basis of size $n$ which gives rise by \Cref{prop:almost-minimal-dedekind-wa} to an $R$-WA with $n+1$ states.

The main optimization of \Cref{alg:reduction} with respect to the usual learning
algorithm lies in the result stated in \Cref{lemma:alg:reduction:second_while_loop_invariant}, which generalizes the following key observation from the $R$-weighted setting:
the $R$-module of $W$-reachable configurations in $\A'$ may increase during the second
\textbf{while} loop, but the corresponding $\KK$-vector space of $W$-reachable configurations does not, therefore the module's rank is also fixed.

\begin{restatable}{lemma}{lemAlgReductionSecondWhileLoopInvariant}
  \label{lemma:alg:reduction:second_while_loop_invariant}
  Let $W_i \subseteq W_{i+1} \subseteq
  \Sigma^*$ be any two consecutive values of $W$ during \Cref{alg:reduction}'s
  second \textbf{while} loop
  (\crefrange{alg:reduction:line:begin_2nd_while}{alg:reduction:line:end_2nd_while}),
  and write $m_i: \G^\sharp(\A')(\triangleright W_i) \rightarrowtail \G^\sharp(\A')(\triangleright
  W_{i+1})$ for the $\M_\C$-morphism between the corresponding
  $(\E_\C,\M_\C)$-images.   Then, $\F m_i$
  is an isomorphism in $\D$. 
\end{restatable}

The genericity of
\Cref{alg:reduction} leaves open three questions that should be addressed to get 
an actual implementation. The answers depend on the category $\C$ in which we work.

\begin{questions}\label{black-boxes}
\begin{enumerate}[leftmargin=*]
\item \label{q:q1} What are the properties of $(\Min \L)(\St)$ (when $\L$ is recognizable)
  and how can these be used to represent it in memory? For $\C = \Mod{\ZZ}$,
  this was the existence of a basis for $(\Min \L)(\St)$.
\item \label{q:q2} How one should compute factorizations: how should it be checked that two $(\E_\C,\M_C)$ images, as on \cref{alg:reduction:line:begin_2nd_while}, coincide 
  and how should the representation of $(\Min
  \L)(\St)$ chosen above be deduced from the generating family of words $W$? In
  $\C = \Mod{\ZZ}$, we mentioned the answer relied on the \emph{Smith Normal Form}. %
\item \label{q:q3} Knowing \Cref{lemma:alg:reduction:second_while_loop_invariant}, can the
  number of words added to $W$ in the second \textbf{while} loop be bounded, so
  that the overall complexity of the reduction can be bounded as well?
\end{enumerate}
\end{questions}

In the next section we answer these questions in the case $\C = \TFMod{\OK}$,
where $\OK$ is a number ring.%

\section{Polynomial-time  Algorithm for the Exact Learning over Number Rings}\label{sec:computational}

We give a concrete implementation of the general algorithm developed in~\Cref{sec:categorical-reduction} in the case of automata weighted over number rings, thereby proving~\Cref{theo:OK}.

Throughout this section, $\KK$ is a number field of degree $d$, and $\OK$ is its ring of integers. See \Cref{app:preliminaries,app:algNF} for extended preliminaries and examples. 
Performing operations  in  $\OK$ typically  requires a suitable representation, such as a compactly represented primitive element. Following~\cite[p.~13]{biasseComputationHNFModule2017}, we  require what we call a \emph{full representation of}~$\OK$, described in more detail in~\Cref{sec:complexity} below.

\theoM*

\algsetup{indent=0.4em}
\begin{algorithm}
    \caption{Computing an $\OK$-WA from a $\KK$-WA}
    \label{alg!make-automaton-integral}
    \begin{algorithmic}[1]
        \REQUIRE a $\KK$-WA $\A=(Q,\A(\triangleright),(\A(\sigma))_{\sigma\in\Sigma},\A(\triangleleft))$  
        
        \algorithmiccomment{{\small\textcolor{blue}{Find a basis B of the backward space}}}
        \STATE \label{alg!make-automaton-integral!line!first}  $W_B = \{ \varepsilon \}$
        \STATE \algorithmicwhile\ there is a $(\sigma, w) \in \Sigma \times W_B$ such that $\A(\sigma w\triangleleft) \notin \langle \A(u\triangleleft ) \mid u \in W_B \rangle_{\KK}$ \algorithmicdo \ $W_B = W_B \cup \{ \sigma w \}$ \algorithmicend
        
        \STATE $B = \left[ \begin{matrix}
            \A(w_1\triangleleft ) & \cdots & \A(w_m\triangleleft )
        \end{matrix}\right]$ if $W_B = \{ w_1, \ldots, w_m \}$
        
        \COMMENT{Conjugate $\A$ with  $B$ to obtain $\A'$  }
        \STATE \label{alg!make-automaton-integral!line!defineforward} Define $\A'=(Q',\A'(\triangleright),(\A'(\sigma))_{\sigma\in\Sigma},\A'(\triangleleft))$ such that, for all $\sigma \in \Sigma$,\[\A'(\triangleright)=\A(\triangleright)B \, , B \A'(\sigma) = \A(\sigma) B \, , B \A'(\triangleleft) = \A(\triangleleft)\]\vspace{-.5cm}

        \COMMENT{Compute a generating set of the forward module of $\A'$ via \Cref{alg:compute-OK-generators}}
        
        \MATCH{the output of \Cref{alg:compute-OK-generators} ran on $\A'$}

        \STATE \textbf{some} $w \in \Sigma^*$\textbf{:} \algorithmicreturn\ $w$
        \CASE{$W \subset \Sigma^*$}
        
        \STATE \label{alg!make-automaton-integral!line!pseudo-basis} Extract a pseudo-basis $\{(\bm{v_i},\mathfrak{a}_i)| 1\leq i\leq \ell\}$ for the forward $\OK$-module $\Reach{\A'}$, using pseudo-HNF.
        
        \COMMENT{Compute an (almost) minimal genera-\\ting set for the forward module of~$\A'$}
        \STATE Call \Cref{alg-psudobasis-generator} on the pseudo-basis, obtaining a generating set $\{\bm{y_i}| 1\leq i\leq \ell+1\}$.
         
         Let $F\coloneq\begin{pmatrix}
            \bm{y_1} &
            \cdots &
            \bm{y_{\ell+1}}
        \end{pmatrix}^t$
        
        \COMMENT{Conjugate $\A'$ with $F$ to obtain $\A''$ }
        \STATE Find $\A''=(Q,\A''(\triangleright),(\A''(\sigma))_{\sigma\in\Sigma},\A''(\triangleleft))$ such that, for all $\sigma \in \Sigma$,  
      \[ \hspace{-3mm} \A''(\triangleright)F=\A'(\triangleright) \, ,  \A''(\sigma)F = F \A'(\sigma) \, ,  F\A''(\triangleleft) = \A(\triangleleft)\]\vspace{-.5cm}
        \RETURN $\A''$
        \ENDCASE
        
        \ENDMATCH
        
    \end{algorithmic}
\end{algorithm}

\Cref{alg!make-automaton-integral} implements a procedure, analogous to \cite[Algorithm 6]{buna-margineanLearningPolynomialRecursive2024}, which computes, given a $\KK$-WA, an equivalent $\OK$-WA if possible, and returns a counterexample otherwise. 
\blueblock{\phantomspacing{}
\Crefrange{alg!make-automaton-integral!line!first}{alg!make-automaton-integral!line!pseudo-basis} of \Cref{alg!make-automaton-integral} implement \Cref{alg:reduction} in the specific case of $(\TFMod{\OK},\OK,\OK)$-automata, by answering Questions~\ref{black-boxes} in this setting. While an instantiation of \Cref{alg:reduction} only computes a minimal $\OK$-modular automaton, in practice we are interested in $\OK$-weighted automata. 
This is why \Cref{alg!make-automaton-integral} performs an additional step and calls \Cref{alg-psudobasis-generator}, in order to transform an $\OK$-modular automaton into an $\OK$-weighted one.
\phantomspacing{}}

We now first give an overview of \Cref{alg!make-automaton-integral}, then give more details in \Cref{sec:number-rings:ideals,sec:number-rings:modules}, and prove its polynomial-time complexity in \Cref{sec:pseudo-basesFM,sec:complexity}.

 Let~$\A=(Q,\A(\triangleright),(\A(\sigma))_{\sigma\in\Sigma},\A(\triangleleft))$ be a $\KK$-weighted automaton, and write $n = |Q|$. Recall from~\Cref{sec:mincat} that
 the forward and backward spaces of $\A$ are the subspaces 
of $\KK^n$ consisting of all reachable and observable vectors, respectively. That is, the forward space is the $\KK$-span of~$\{\A(\triangleright w) \, | \, w\in \Sigma^*\}$ and the backward space is that of $\{\A(w \triangleleft) \, | \, w\in \Sigma^*\}$.  Similarly,  the 
  forward $\OK$-module of $\A$ is the $\OK$-span of~$\{\A(\triangleright w) \, | \, w\in \Sigma^*\}$. 
  
We now give an overview of \Cref{alg!make-automaton-integral} and its ingredients. We refer to the corresponding steps of \Cref{alg:reduction} \inlineblock{catcol}{in blue}.

\inlineblock{catcol}{\textit{\Cref{alg:reduction:line:obs} of \Cref{alg:reduction}.}}
\Cref{alg!make-automaton-integral} first conjugates~$\A$ with a 
matrix $B$, 
whose $m$ columns form a basis of the backward space of~$\A$. Denote by $\A'$ the resulting automaton, defined in \cref{alg!make-automaton-integral!line!defineforward} of \Cref{alg!make-automaton-integral}.  
Assuming that $\A$ has an equivalent $\OK$-WA,
the forward $\OK$-module of $\A'$ is a submodule of $\OK^m$ (this matters when it comes to the complexity of the algorithm). Indeed, for all words~$w\in \Sigma^*$, by definition of~$\A'$, 
$\A'(\triangleright w) = \A(\triangleright w ) B
    = \left[\begin{matrix} \A(\triangleright w w_1 \triangleleft) & \cdots & \A(\triangleright w w_m\triangleleft )\end{matrix}\right]$,
and this vector only contains values from $\OK$, by assumption.

\inlineblock{catcol}{\textit{\Crefrange{alg:reduction:line:begin_1st_while}{alg:reduction:line:end_2nd_while} of \Cref{alg:reduction}.}}
Next, \Cref{alg!make-automaton-integral} calls
\Cref{alg:compute-OK-generators}, which, on an input WA with $m$ states,  
either 
    returns a generating set for the forward $\OK$-module, if this forward module lies within $\OK^m$, 
    or returns a counterexample, otherwise. Note that the generating set may be very large.

\inlineblock{catcol}{\textit{\Cref{alg:reduction:line:extract_basis} of \Cref{alg:reduction}.}}
If it were possible to extract a basis from the generating set for the forward $\OK$-module, then we could simply conjugate $\A'$ with the matrix whose rows form that basis, directly producing an $\OK$-automaton.
This is precisely the approach taken in the case of $\ZZ$ in \cite[Algorithm 6]{buna-margineanLearningPolynomialRecursive2024}.
However, the essential difference, when moving from $\ZZ$ to $\OK$, is that $\OK$-submodules of~$\OK^m$ may fail to have a basis, precisely because $\OK$ is not a PID in general. However, $\OK$ is still a Dedekind domain, and we can use the notion of \emph{pseudo-basis} for $\OK$-modules, recalled in \Cref{sec:comm-alg-primer}.
\blueblock{\phantomspacing{}Pseudo-bases answer Question~\ref{black-boxes}.\ref{q:q1} for number rings: how to represent the minimal $\OK$-modular automaton?\phantomspacing{}}

The \emph{pseudo-Hermite Normal Form} (pseudo-HNF)  is a means to compute a pseudo-basis of an $\OK$-submodule of $\OK^m$, starting from a generating set~\cite[Sec.~1.4]{cohen2012advanced}, also see~\cite[Thm.~2.4]{bosmapohst1991}. This is  analogous to using the HNF for computing a $\ZZ$-basis for a $\ZZ$-submodule of $\ZZ^m$.
We provide an introduction to HNF and $\ZZ$-modules in \Cref{app:HNF}, and we give details about pseudo-HNF and $\OK$-modules in~\Cref{sec:pseudo-basesFM}.
\Cref{alg!make-automaton-integral} uses the pseudo-HNF in \Cref{alg!make-automaton-integral!line!pseudo-basis} to compute a pseudo-basis from the generating set output by \Cref{alg:compute-OK-generators}. %
\blueblock{\phantomspacing{}The pseudo-HNF answers Question~\ref{black-boxes}.\ref{q:q2} for number rings: how to represent images of linear maps computationally?\phantomspacing{}} 

The pseudo-basis obtained in this way gives a minimal representation of the forward module of $\A'$. However, if we were to conjugate $\A'$ with the pseudo-basis directly, then we would not be sure to obtain an $\OK$-weighted automaton, but only a $\KK$-weighted automaton. This is because the fractional ideals in the pseudo-basis may contain elements outside~$\OK$.
A main added ingredient in our algorithm, compared to~\cite[Algorithm 6]{buna-margineanLearningPolynomialRecursive2024},
is that we compute, starting from a
size-$\ell$ pseudo-basis for a submodule $M$ of $\OK^\ell$, an $\OK$-generating set for $M$ of size at most $\ell + 1$ (\Cref{alg-psudobasis-generator}).  
Conjugating $\A'$ 
with a matrix whose rows are formed by this generating set produces an $\OK$-weighted automaton 
whose number of states is at most one more than
the number of states of the canonical minimal $\KK$-weighted automaton equivalent to~$\A$. %
\blueblock{\phantomspacing{}This last added ingredient is not encompassed by \Cref{alg:reduction}: it is an additional step taken to transform the resulting minimal $\OK$-modular automaton into an $\OK$-weighted automaton, and can also be understood as a constructive version of the generic \Cref{prop:almost-minimal-dedekind-wa}, which holds for all Dedekind domains.\phantomspacing{}}

\subsection{Ideals in Number Rings} \label{sec:number-rings:ideals}
The ring~$\OK$ of algebraic integers in $\KK$
is a $\ZZ$-module  of rank~$[\KK:\QQ]$.
An integral basis of $\OK$ is 
a $\ZZ$-basis $\{\omega_1,\cdots, \omega_d\}$ of $\OK$ as a $\ZZ$-module, that is, 
\begin{align}
\label{eq-defintegral}
  \OK=\ZZ \,\omega_1 \oplus \cdots \oplus \ZZ \, \omega_d\,.  
\end{align}

Since $\ZZ$ is a PID, each $\ZZ$-submodule of $\OK$ is finitely generated; by this, each ideal~$\mathfrak{a}\subseteq \OK$ can  be viewed as a finitely generated $\ZZ$-module.  Every ideal~$\mathfrak{a}\subseteq \OK$ is even a \emph{full-rank} $\ZZ$-submodule of~$\OK$; its
\emph{norm},  denoted by $\norm{\mathfrak{a}}$, is defined as $|\OK/\mathfrak{a}|$, see~\cite[Prop.~4.6.3]{cohen2013course}.
The \emph{fractional ideals} of $\OK$ are the finitely generated full-rank  $\ZZ$-submodules $\mathfrak{b}$ of $\KK$.  
For clarity, we sometimes refer to 
 ideals~$\mathfrak{a}\subseteq \OK$ as  \emph{integral ideals}.
 For any fractional ideal~$\mathfrak{b}$, there exists $r\in \OK$ such that $r\mathfrak{b}$ is integral.

The number rings $\OK$ are not in general PIDs, see \Cref{example-cyclotomic,example-quadratic} in \Cref{app:algNF}. 
Still,
 the ring of integers $\OK$ of a number field~$\KK$ is a Dedekind domain, see for instance~\cite[Prop. 1.2.3]{cohen2012advanced}.
Moreover, each fractional ideal~$\mathfrak{b}$ can    be written uniquely as a product of powers of prime ideals.
We provide a detailed discussion about representation and complexity of manipulating  algebraic numbers and ideals in \Cref{app:algoALGF}. 
 Since $\OK$ is a Dedekind domain, the set of fractional ideals of $\KK$ forms a commutative group under ideal multiplication with identity element $\OK$, see for example \cite[Thm.~4.6.14]{cohen2013course}.
We note in passing that ideals in $\OK$ when seen  
as $\OK$-modules are finitely generated, but might not have a basis.

\subsection{Modules Over Number Rings} \label{sec:number-rings:modules}
\label{sec:modules-numberring}
After recalling the notion of pseudo-basis from~\cite{cohenHermiteSmithNormal1996}, we  state several facts we use about  modules over  number rings. 
 Let $M \subseteq \OK^n$ be an $\OK$-module.
Let $\bm{v}_1, \cdots, \bm{v}_k \in \KK^n$ and  
$\mathfrak{a}_1, \cdots, \mathfrak{a}_k$ be fractional ideals, we say that 
$\{(\bm{v}_i,\mathfrak{a}_i) \mid 1\leq i\leq k\}$ is a pseudo-generating set of $M$ if
$M= \mathfrak{a}_1 \bm{v}_1 + \cdots +  \mathfrak{a}_k \bm{v}_k$,
and we say that it is a pseudo-basis of~$M$ if
$M= \mathfrak{a}_1 \bm{v}_1 \oplus \cdots \oplus \mathfrak{a}_k \bm{v}_k$.
By~\cite[Cor.~1.2.25]{cohen2012advanced}, every $\OK$-module $M \subseteq \OK^n$ has a pseudo-basis.

\begin{restatable}[{\noadjust\cite[Lemma 1.2.20]{cohen2012advanced}}]{lemma}{lemisoJJp}
\label{lem-isoJJp}
 If $\mathfrak{a}$ and $\mathfrak{b}$ are fractional ideals of $\OK^n$, there is an isomorphism of $\OK$-modules such that 
 $\mathfrak{a} \oplus \mathfrak{b} \simeq \OK \oplus \mathfrak{a}\mathfrak{b}$.
\end{restatable}

We outline the key elements required to prove~\Cref{lem-isoJJp} in \Cref{app:modules-numberring}. We employ similar arguments to those used in the outline in~\Cref{alg-psudobasis-generator}.
The following proposition follows from iteratively applying~\Cref{lem-isoJJp}. 

\begin{proposition}[{\noadjust\cite[Prop. 1.2.19]{cohen2012advanced}}]
\label{prop-ngen}
    Let $M$ be a finitely generated, torsion-free $\OK$-module of rank~$n$. Then there exists a fractional ideal~$\mathfrak{b}$ of $\OK$ such that $M \simeq \OK^{n-1} \oplus \mathfrak{b}$.
\end{proposition}
As each fractional ideal $\mathfrak{b}$ of $\OK$ can be generated by two elements, $M$ in the above proposition has a generating set of cardinality~$n+1$. 

To compute a pseudobasis for an $\OK$-module, we recall a definition of pseudo-HNF~\cite[Sec.~1.4]{cohen2012advanced}. 
For simplicity we only give the definition for  full-rank modules; it can be extended to lower-rank modules. %
Let $M$ be a full-rank $\OK$-module of $\OK^n$ and 
$\{(\bm{v_i},\mathfrak{a}_i) \mid 1\leq i\leq k\}$ be a pseudo-generating set for $M$; in particular, $k \geq n$. 
Let $A = [a_{i,j}]_{1 \le i \le n, 1 \le j \le k}$ be the $n\times k$ matrix whose $i$-th column is $\bm{v_i}$.
By~\cite[Thm. 1.4.6]{cohen2012advanced}, 
there exist fractional ideals $\mathfrak{b}_1, \ldots, \mathfrak{b}_k$ and a $k\times k$ matrix $U$ such that 
    (1) $AU=[0|H]$ where $H$ is upper triangular with $1$ on the diagonal, and 
    (2) $\{(\bm{h}_i,\mathfrak{c}_i) \mid 1\leq i\leq n\}$ is a pseudo-basis for $M$, where $\bm{h}_i$ is the $i$-th column of $H$ and $\mathfrak{c}_i\coloneq\mathfrak{b}_{k-n+i}$ for $1\leq i\leq n$.

By~\cite[Thm. 1.4.9]{cohen2012advanced}, the ideals $\mathfrak{c}_i$ in this computation  are unique and only depend on~$M$ and not on the basis element~$\bm{h}_i$. The ideal
$\prod_{i=1}^n \mathfrak{c}_i$ corresponds to the ideal $\mathfrak{b}$ in \Cref{prop-ngen}. It is integral, and can be computed from the \emph{minor ideals} of $A$ and the fractional ideals $\mathfrak{a}_i$.   
The  $r \times r$ minor ideals are defined, for $I \subseteq \{1,\ldots,n\}$ and $J \subseteq \{1,\ldots,k\}$ of size $r$, as 
 \begin{equation}
 \label{eq-primediv}
    \mathfrak{d}_{I,J}\coloneq\det([a_{i,j}]_{i\in I, j\in J}) \prod_{j\in J} \mathfrak{a}_j \, .
 \end{equation}
Then $\prod_{i=1}^n \mathfrak{c}_i$ is the sum of the $\mathfrak{a}_i$'s and all $n\times n$ minor ideals of $A$~\cite[Def. 1.4.8 and Thm. 1.4.9]{cohen2012advanced}.
Given  a full representation of~$\OK$, the main Theorem of~\cite{biasseComputationHNFModule2017} shows that 
the pseudo-HNF for full-rank modules is computable within polynomial-time.
A crucial lemma in our learning algorithm is the computation of a reasonably small generating set of $\OK$-submodules of $\OK^n$ from a pseudo-basis:  

\begin{restatable}{lemma}{lempseudobasisgen}
\label{lem:pseudobasisgen}
    Let $\{(\mathfrak{a}_i,\bm{v_i})| 1\leq i\leq n\}$ be a pseudo-basis for an
    $\OK$-module 
    $M \subseteq \OK^n$. Given 
 a full representation of~$\OK$, 
    a generating set of cardinality at most $n+1$ is computable within polynomial time in the size of the input pseudo-basis.
\end{restatable}

\looseness=-1
\Cref{alg-psudobasis-generator} provides a high-level description of the steps required in~\Cref{lem:pseudobasisgen} (the detailed proof can be found in  \Cref{app:modules-numberring}). 
The derived complexity bound relies on  the notion of ideal factor refinement~\cite[Alg. 5.6 and Prop. 5.7]{ge1994recognizing} and an adaptation of~\cite[Lemma 5.2.2]{stein2012algebraic} and~\cite[Props. 1.3.10 and 1.3.12]{cohen2012advanced}.

\subsection{Fast Computation of an Almost Minimal Generating Set}
\label{sec:pseudo-basesFM}

As described in the overview,  \Cref{alg!make-automaton-integral} calls 
\Cref{alg:compute-OK-generators} to compute a generating set  for the forward $\OK$-module of $\A'$, if the forward module lies within $\OK^m$.

In~\Cref{lem:lengthdchain} we establish an upper bound on the length of strictly increasing chains of  $\OK$-submodules of $\OK^m$, which is a key result in the complexity analysis of  \Cref{alg:compute-OK-generators}. 
A similar bound for~$\ZZ$-submodules of $\ZZ^m$ 
is derived in~\cite[Prop. 3.1]{buna-margineanLearningPolynomialRecursive2024} using the elementary divisor theorem.  
There, the proof relies on the fact that if $M\subset N \subseteq \ZZ^m$ are  $\ZZ$-modules  of equal rank~$r \leq m$, there is a basis $\bm{v_1},\cdots,\bm{v_r}$ of $N$ and nonzero scalars $d_1,\cdots,d_r\in \ZZ$
such that 
    $N =\ZZ \bm{v_1} \oplus \cdots \oplus  \ZZ \bm{v_r}$ and 
    $M =  \ZZ  \, d_1  \bm{v_1}\oplus \cdots \oplus  \ZZ \, d_r  \bm{v_r}.$
The scalars~$d_i$ are  the \emph{primary divisors} of the  quotient module $N/M \cong \ZZ/(d_1 \ZZ) \oplus \ldots \oplus \ZZ/(d_r \ZZ)$.
Since $N/M$ is a finite group, using Lagrange’s Theorem, its cardinality gives an upper bound on the length of strictly increasing chains of modules  between $M$ and $N$.  
In~\cite[Prop.~3.1]{buna-margineanLearningPolynomialRecursive2024}, the cardinality of $N/M$
is derived by computation of 
primary divisors through the Smith Normal Form (a combination of an HNF and an HNF of the transpose~\cite[Sec. 2.4.4]{cohen2013course}).

We generalize this to the number ring setting as follows.

\begin{restatable}{lemma}{lemlengthdchain}
\label{lem:lengthdchain}
Let $\{(\mathfrak{a}_i,\bm{v_i})| 1\leq i\leq m\}$ be a pseudo-generating set for a full-rank
    $\OK$-module 
    $M \subseteq \OK^n$.
Let $A$
be the $n\times m$ matrix whose $i$-th column is $\bm{v_i}$.
Let $\mathfrak{d}$ be  the sum of all $n\times n$ minor ideals of $A$ and of the  $\mathfrak{a}_i$'s.
    Then all strictly increasing chains of $\OK$-modules $M=M_1\subset M_2 \subset \cdots \subset M_{k-1} \subset M_k \subseteq \OK^n$
     have length at most $\log(\norm{ \mathfrak{d}})$.
\end{restatable}
\blueblock{\phantomspacing{}This lemma is the answer, in the number ring setting, to Question~\ref{black-boxes}.\ref{q:q3}: how to bound the length of
strictly increasing chains of $\OK$-submodules?\phantomspacing{}}

We briefly argue that the HNF and pseudo-HNF suffice to establish  complexity bounds for~\cite[Algorithm 5]{buna-margineanLearningPolynomialRecursive2024} and for our analogous \Cref{alg:compute-OK-generators} over number fields.  
Both procedures employ a two-pass search strategy: (1) The first pass identifies elements that increase the rank of the sought-after module.
(2) The second pass augments the  module  until the complete forward module is constructed.

This two-pass approach is noted to be crucial for ensuring polynomial time already in the $\ZZ$ setting. The first pass induces  a strictly  increasing  sequence of vector spaces, and  has length at most $m$. 
Given the $\ZZ$-module $M$ at the end of the first phase, 
the second phase may iterate over as many times as there are strictly increasing $\ZZ$-modules of rank~$r$ equal to that of $M$ and lying between $M$ and $\ZZ^m$.
As described above  computing an upper bound on such chains  was proved by the elementary divisor theorem and the SNF in~\cite[Prop. 3.1]{buna-margineanLearningPolynomialRecursive2024}. 
We argue that $M$ can be assumed to be full-rank in our algorithms, so that $|\ZZ^m/M|$ is computable through a full-rank HNF (see~\Cref{lem:HNF} in the appendix). Because the minimization over both $\QQ$ and $\KK$ is in polynomial time, we can indeed assume without loss of generality that the input automaton is minimal so that its forward module is full-rank. In fact in the reduction of the learning over a number ring to the learning over its field of fractions, this property is even automatically satisfied.

\subsection{Overall Complexity of \texorpdfstring{\Cref{alg!make-automaton-integral,alg:compute-OK-generators}}{Algorithms 4 and 5}}
\label{sec:complexity}

Before stating the complexity of \Cref{alg!make-automaton-integral,alg:compute-OK-generators}, we first outline  the elements needed for the full representation of~$\OK$. 
Let  $C_{\KK}\coloneq d^{4} (\log d+ \log \Delta)$ such that $d=[K:\QQ]$ and $\Delta$ is the discriminant of $\KK$; %
see \Cref{app:algoALGF} for more details on the discriminant of~$\KK$ and the symbolic and regular representation of algebraic numbers.
A \emph{full representation of~$\OK$} consists of 
 an integral basis~$\Omega$ of~$\OK$,  as  in~\eqref{eq-defintegral}, where $1\in \Omega$, together with  
 a primitive element~$\theta$, whose minimal polynomial $m_{\theta}=x^d+\sum_{i=0}^{d-1} a_i x^i$ is such that
 $\log\mathopen{}\left( \prod_{1\leq i\leq j\leq d} |\sigma_i(\theta)-\sigma_j(\theta)|^2 \right)\mathclose{} \leq C_K$,
 where the $\sigma_i(\theta)$ are the $d$ distinct zeros  of $m_\theta$, and
the bit length of $a_i$ is bounded by $C_{\KK}$.
It is shown in \cite[pp. 590--591]{biasseComputationHNFModule2017} that there exists such a primitive element $\theta$, which is a sum of a subset of~$\Omega$.
 The elements of $\Omega$ and $\theta$ are given both in symbolic and regular representations.
We note that while~\cite[pp. 590--591]{biasseComputationHNFModule2017} require more extensive precomputed data, we omit these prerequisites, as they are computable in polynomial time from~$\Omega$ and $\theta$. 

The detailed complexity analysis    of  \Cref{alg!make-automaton-integral,alg:compute-OK-generators}, stated below, can be found in \Cref{app:modules-numberring}. 

\begin{restatable}{lemma}{complexAlg}
\label{complexAlg12}
Given a full representation of $\OK$, \Cref{alg!make-automaton-integral,alg:compute-OK-generators} run within polynomial time in
 the size of the input automaton, the degree of $\KK$ and the logarithm of its discriminant.   
\end{restatable}

Recall that the \emph{principal ideal problem (PIP)} is to decide if an integral ideal $\mathfrak{a} \subseteq \OK$, given by its basis as a $\ZZ$-module, is principal. The PIP can be solved in quantum polynomial time \cite{biasseEfficientQuantumAlgorithms2015}.
To conclude this paper, we relate this problem to computing minimal $\OK$-WA; the proof is in \Cref{app:modules-numberring}.
\begin{restatable}{proposition}{decidingMinimalityPIPHard}
    \label{prop:minimal-pip-hard}
    Deciding whether an $\OK$-WA is state-minimal is PIP-hard.
\end{restatable}

Number rings play a central role in cryptography, and Proposition 26 can be seen in this light: a number of cryptographic schemes rely on the hardness of finding the generator of a principal ideal \cite[\textsection 6]{biasseEfficientQuantumAlgorithms2015}.

\begingroup
  \renewcommand\thefootnote{}\footnote{This research has received financial support from the Agence Nationale de la Recherche (ANR), grants ANR-23-CE48-0012-01 and ANR-22- CE48-0005, from the European Research Council (ERC) under the European Union’s Horizon research and innovation program (grant agreement No. 101167526), and from United Kingdom Research and Innovation (UKRI), grant EP/X033813/.}%
  \addtocounter{footnote}{-1}%
\endgroup

\newpage

 \bibliographystyle{IEEEtran} 
 \bibliography{bibtex}

\begin{thebibliography}{10}
\providecommand{\url}[1]{#1}
\csname url@samestyle\endcsname
\providecommand{\newblock}{\relax}
\providecommand{\bibinfo}[2]{#2}
\providecommand{\BIBentrySTDinterwordspacing}{\spaceskip=0pt\relax}
\providecommand{\BIBentryALTinterwordstretchfactor}{4}
\providecommand{\BIBentryALTinterwordspacing}{\spaceskip=\fontdimen2\font plus
\BIBentryALTinterwordstretchfactor\fontdimen3\font minus
  \fontdimen4\font\relax}
\providecommand{\BIBforeignlanguage}[2]{{%
\expandafter\ifx\csname l@#1\endcsname\relax
\typeout{** WARNING: IEEEtran.bst: No hyphenation pattern has been}%
\typeout{** loaded for the language `#1'. Using the pattern for}%
\typeout{** the default language instead.}%
\else
\language=\csname l@#1\endcsname
\fi
#2}}
\providecommand{\BIBdecl}{\relax}
\BIBdecl

\bibitem{Angluin87:Learning-regular-sets-from-queries-and-counterexamples}
\BIBentryALTinterwordspacing
D.~Angluin, ``Learning regular sets from queries and counterexamples,''
  \emph{Information and Computation}, vol.~75, no.~2, pp. 87--106, 1987.
  [Online]. Available:
  \url{https://www.sciencedirect.com/science/article/pii/0890540187900526}
\BIBentrySTDinterwordspacing

\bibitem{valiant1984}
L.~G. Valiant, ``A theory of the learnable,'' \emph{Communications of the ACM},
  vol.~27, no.~11, pp. 1134--1142, 1984.

\bibitem{vaandrager2017model}
F.~Vaandrager, ``Model learning,'' \emph{Communications of the ACM}, vol.~60,
  no.~2, pp. 86--95, 2017.

\bibitem{PeledVY02}
D.~Peled, M.~Vardi, and M.~Yannakakis, ``Black box checking,'' \emph{J. Autom.
  Lang. Comb.}, vol.~7, no.~2, pp. 225--246, 2002.

\bibitem{higuera2010}
C.~Higuera, \emph{Grammatical inference : learning automata and
  grammars}.\hskip 1em plus 0.5em minus 0.4em\relax Cambridge New York:
  Cambridge University Press, 2010.

\bibitem{BergGJLRS05}
T.~Berg, O.~Grinchtein, B.~Jonsson, M.~Leucker, H.~Raffelt, and B.~Steffen,
  ``On the correspondence between conformance testing and regular inference,''
  in \emph{{FASE} 2005}, ser. LNCS, vol. 3442, 2005, pp. 175--189.

\bibitem{ChaluparPPR14}
G.~Chalupar, S.~Peherstorfer, E.~Poll, and J.~de~Ruiter, ``Automated reverse
  engineering using lego{\textregistered},'' in \emph{{WOOT} '14, San Diego,
  CA, USA, August 19, 2014}.\hskip 1em plus 0.5em minus 0.4em\relax {USENIX}
  Association, 2014.

\bibitem{Leucker2006}
M.~Leucker, ``Learning meets verification,'' in \emph{International Symposium
  on Formal Methods for Components and Objects}.\hskip 1em plus 0.5em minus
  0.4em\relax Springer, 2006, pp. 127--151.

\bibitem{MarusicW15}
\BIBentryALTinterwordspacing
I.~Marusic and J.~Worrell, ``Complexity of equivalence and learning for
  multiplicity tree automata,'' \emph{J. Mach. Learn. Res.}, vol.~16, pp.
  2465--2500, 2015. [Online]. Available:
  \url{https://dl.acm.org/doi/10.5555/2789272.2912078}
\BIBentrySTDinterwordspacing

\bibitem{ShahbazG09:Inferring-Mealy-Machines}
M.~Shahbaz and R.~Groz, ``Inferring mealy machines,'' in \emph{FM 2009: Formal
  Methods}, A.~Cavalcanti and D.~R. Dams, Eds.\hskip 1em plus 0.5em minus
  0.4em\relax Berlin, Heidelberg: Springer Berlin Heidelberg, 2009, pp.
  207--222.

\bibitem{BolligHKL09:Angluin-style-learning-of-NFA}
B.~Bollig, P.~Habermehl, C.~Kern, and M.~Leucker, ``Angluin-style learning of
  nfa,'' in \emph{Proceedings of the 21st International Joint Conference on
  Artificial Intelligence}, ser. IJCAI'09.\hskip 1em plus 0.5em minus
  0.4em\relax San Francisco, CA, USA: Morgan Kaufmann Publishers Inc., 2009, p.
  1004–1009.

\bibitem{HowarSJC:Inferring-canonical-register-automata}
F.~Howar, B.~Steffen, B.~Jonsson, and S.~Cassel, ``Inferring canonical register
  automata,'' in \emph{Verification, Model Checking, and Abstract
  Interpretation}, V.~Kuncak and A.~Rybalchenko, Eds.\hskip 1em plus 0.5em
  minus 0.4em\relax Berlin, Heidelberg: Springer Berlin Heidelberg, 2012, pp.
  251--266.

\bibitem{AartsFKV15:Learning-Register-Automata-with-Fresh-Value-Generation}
F.~Aarts, P.~Fiterau-Brostean, H.~Kuppens, and F.~Vaandrager, ``Learning
  register automata with fresh value generation,'' in \emph{Theoretical Aspects
  of Computing - ICTAC 2015}, M.~Leucker, C.~Rueda, and F.~D. Valencia,
  Eds.\hskip 1em plus 0.5em minus 0.4em\relax Cham: Springer International
  Publishing, 2015, pp. 165--183.

\bibitem{Vilar96:Query-learning-of-subsequential-transducers}
\BIBentryALTinterwordspacing
J.~M. Vilar, ``Query learning of subsequential transducers,'' in
  \emph{Grammatical Inference: Learning Syntax from Sentences, 3rd
  International Colloquium, ICGI-96, Montpellier, France, September 25-27,
  1996, Proceedings}, ser. Lecture Notes in Computer Science, L.~Miclet and
  C.~de~la Higuera, Eds., vol. 1147.\hskip 1em plus 0.5em minus 0.4em\relax
  Springer, 1996, pp. 72--83. [Online]. Available:
  \url{https://doi.org/10.1007/BFb0033343}
\BIBentrySTDinterwordspacing

\bibitem{MoermanSSKS17:Learning-nominal-automata}
\BIBentryALTinterwordspacing
J.~Moerman, M.~Sammartino, A.~Silva, B.~Klin, and M.~Szynwelski, ``Learning
  nominal automata,'' in \emph{Proceedings of the 44th ACM SIGPLAN Symposium on
  Principles of Programming Languages}, ser. POPL '17.\hskip 1em plus 0.5em
  minus 0.4em\relax New York, NY, USA: Association for Computing Machinery,
  2017, p. 613–625. [Online]. Available:
  \url{https://doi.org/10.1145/3009837.3009879}
\BIBentrySTDinterwordspacing

\bibitem{AngluinF16:Learning-regular-omega-languages}
\BIBentryALTinterwordspacing
D.~Angluin and D.~Fisman, ``Learning regular omega languages,''
  \emph{Theoretical Computer Science}, vol. 650, pp. 57--72, 2016, algorithmic
  Learning Theory. [Online]. Available:
  \url{https://www.sciencedirect.com/science/article/pii/S0304397516303760}
\BIBentrySTDinterwordspacing

\bibitem{BeimelBBKV00}
A.~Beimel, F.~Bergadano, N.~H. Bshouty, E.~Kushilevitz, and S.~Varricchio,
  ``Learning functions represented as multiplicity automata,'' \emph{J. {ACM}},
  vol.~47, no.~3, pp. 506--530, 2000.

\bibitem{buna-margineanLearningPolynomialRecursive2024}
\BIBentryALTinterwordspacing
A.~{Buna-Marginean}, V.~Cheval, M.~Shirmohammadi, and J.~Worrell, ``On
  {{Learning Polynomial Recursive Programs}},'' \emph{Proceedings of the ACM on
  Programming Languages}, vol.~8, no. POPL, pp. 34:1001--34:1027, Jan. 2024.
  [Online]. Available: \url{https://dl.acm.org/doi/10.1145/3632876}
\BIBentrySTDinterwordspacing

\bibitem{UrbatS20:Automata-Learning:-An-Algebraic-Approach}
\BIBentryALTinterwordspacing
H.~Urbat and L.~Schr\"{o}der, ``Automata learning: An algebraic approach,'' in
  \emph{Proceedings of the 35th Annual ACM/IEEE Symposium on Logic in Computer
  Science}, ser. LICS '20.\hskip 1em plus 0.5em minus 0.4em\relax New York, NY,
  USA: Association for Computing Machinery, 2020, p. 900–914. [Online].
  Available: \url{https://doi.org/10.1145/3373718.3394775}
\BIBentrySTDinterwordspacing

\bibitem{HeerdtSS20:Learning-Automata-with-Side-Effects}
\BIBentryALTinterwordspacing
G.~v. Heerdt, M.~Sammartino, and A.~Silva, ``Learning automata with
  side-effects,'' in \emph{Coalgebraic Methods in Computer Science: 15th IFIP
  WG 1.3 International Workshop, CMCS 2020, Colocated with ETAPS 2020, Dublin,
  Ireland, April 25–26, 2020, Proceedings}.\hskip 1em plus 0.5em minus
  0.4em\relax Berlin, Heidelberg: Springer-Verlag, 2020, p. 68–89. [Online].
  Available: \url{https://doi.org/10.1007/978-3-030-57201-3_5}
\BIBentrySTDinterwordspacing

\bibitem{BarloccoKR19:Coalgebra-Learning-via-Duality}
S.~Barlocco, C.~Kupke, and J.~Rot, ``Coalgebra learning via duality,'' in
  \emph{Foundations of Software Science and Computation Structures},
  M.~Boja{\'{n}}czyk and A.~Simpson, Eds.\hskip 1em plus 0.5em minus
  0.4em\relax Cham: Springer International Publishing, 2019, pp. 62--79.

\bibitem{colcombetLearningAutomataTransducers2021}
\BIBentryALTinterwordspacing
T.~Colcombet, D.~Petri{\c s}an, and R.~Stabile, ``Learning {{Automata}} and
  {{Transducers}}: {{A Categorical Approach}},'' in \emph{29th {{EACSL Annual
  Conference}} on {{Computer Science Logic}} ({{CSL}} 2021)}, ser. Leibniz
  {{International Proceedings}} in {{Informatics}} ({{LIPIcs}}), C.~Baier and
  J.~{Goubault-Larrecq}, Eds., vol. 183.\hskip 1em plus 0.5em minus 0.4em\relax
  Dagstuhl, Germany: Schloss Dagstuhl--Leibniz-Zentrum f{\"u}r Informatik,
  2021, pp. 15:1--15:17. [Online]. Available:
  \url{https://drops.dagstuhl.de/opus/volltexte/2021/13449}
\BIBentrySTDinterwordspacing

\bibitem{vanheerdtLearningWeightedAutomata2020}
G.~v. Heerdt, C.~Kupke, J.~Rot, and A.~Silva, ``Learning {{Weighted Automata}}
  over {{Principal Ideal Domains}},'' in \emph{Foundations of {{Software
  Science}} and {{Computation Structures}}}, J.~{Goubault-Larrecq} and
  B.~K{\"o}nig, Eds.\hskip 1em plus 0.5em minus 0.4em\relax Cham: Springer
  International Publishing, 2020, pp. 602--621.

\bibitem{berstelNoncommutativeRationalSeries2010}
\BIBentryALTinterwordspacing
J.~Berstel and C.~Reutenauer, \emph{Noncommutative {{Rational Series}} with
  {{Applications}}}, ser. Encyclopedia of {{Mathematics}} and Its
  {{Applications}}.\hskip 1em plus 0.5em minus 0.4em\relax Cambridge: Cambridge
  University Press, 2010. [Online]. Available:
  \url{https://www.cambridge.org/core/books/noncommutative-rational-series-with-applications/CADC75E4ADDA69E99BB8B0D8FE9AD119}
\BIBentrySTDinterwordspacing

\bibitem{salomaa2012automata}
A.~Salomaa and M.~Soittola, \emph{Automata-theoretic aspects of formal power
  series}.\hskip 1em plus 0.5em minus 0.4em\relax Springer Science \& Business
  Media, 2012.

\bibitem{martin2004grammars}
C.~Mart{\'\i}n-Vide and V.~Mitrana, \emph{Grammars and Automata for String
  Processing: From Mathematics and Computer Science to Biology, and
  Back}.\hskip 1em plus 0.5em minus 0.4em\relax CRC Press, 2004, vol.~9.

\bibitem{paz2014introduction}
A.~Paz, \emph{Introduction to probabilistic automata}.\hskip 1em plus 0.5em
  minus 0.4em\relax Academic Press, 1971.

\bibitem{ChistikovKMSW16}
\BIBentryALTinterwordspacing
D.~Chistikov, S.~Kiefer, I.~Marusic, M.~Shirmohammadi, and J.~Worrell, ``On
  restricted nonnegative matrix factorization,'' in \emph{43rd International
  Colloquium on Automata, Languages, and Programming, {ICALP} 2016, July 11-15,
  2016, Rome, Italy}, ser. LIPIcs, I.~Chatzigiannakis, M.~Mitzenmacher,
  Y.~Rabani, and D.~Sangiorgi, Eds., vol.~55.\hskip 1em plus 0.5em minus
  0.4em\relax Schloss Dagstuhl - Leibniz-Zentrum f{\"{u}}r Informatik, 2016,
  pp. 103:1--103:14. [Online]. Available:
  \url{https://doi.org/10.4230/LIPIcs.ICALP.2016.103}
\BIBentrySTDinterwordspacing

\bibitem{ChistikovKMSW17}
\BIBentryALTinterwordspacing
------, ``On rationality of nonnegative matrix factorization,'' in
  \emph{Proceedings of the Twenty-Eighth Annual {ACM-SIAM} Symposium on
  Discrete Algorithms, {SODA} 2017, Barcelona, Spain, Hotel Porta Fira, January
  16-19}, P.~N. Klein, Ed.\hskip 1em plus 0.5em minus 0.4em\relax {SIAM}, 2017,
  pp. 1290--1305. [Online]. Available:
  \url{https://doi.org/10.1137/1.9781611974782.84}
\BIBentrySTDinterwordspacing

\bibitem{colcombetpetrisan2017}
\BIBentryALTinterwordspacing
T.~Colcombet and D.~Petrisan, ``Automata minimization: a functorial approach,''
  in \emph{7th Conference on Algebra and Coalgebra in Computer Science, {CALCO}
  2017, June 12-16, 2017, Ljubljana, Slovenia}, ser. LIPIcs, F.~Bonchi and
  B.~K{\"{o}}nig, Eds., vol.~72.\hskip 1em plus 0.5em minus 0.4em\relax Schloss
  Dagstuhl - Leibniz-Zentrum f{\"{u}}r Informatik, 2017, pp. 8:1--8:16.
  [Online]. Available: \url{https://doi.org/10.4230/LIPIcs.CALCO.2017.8}
\BIBentrySTDinterwordspacing

\bibitem{bergadanoLearningBehaviorsAutomata1994}
F.~Bergadano and S.~Varricchio, ``Learning behaviors of automata from
  multiplicity and equivalence queries,'' in \emph{Algorithms and
  {{Complexity}}}, ser. Lecture {{Notes}} in {{Computer Science}},
  M.~Bonuccelli, P.~Crescenzi, and R.~Petreschi, Eds.\hskip 1em plus 0.5em
  minus 0.4em\relax Berlin, Heidelberg: Springer, 1994, pp. 54--62.

\bibitem{biasseEfficientQuantumAlgorithms2015}
\BIBentryALTinterwordspacing
J.-F. Biasse and F.~Song, ``Efficient quantum algorithms for computing class
  groups and solving the principal ideal problem in arbitrary degree number
  fields,'' in \emph{Proceedings of the 2016 {{Annual ACM-SIAM Symposium}} on
  {{Discrete Algorithms}} ({{SODA}})}, ser. Proceedings.\hskip 1em plus 0.5em
  minus 0.4em\relax {Society for Industrial and Applied Mathematics}, Dec.
  2015, pp. 893--902. [Online]. Available:
  \url{https://epubs.siam.org/doi/10.1137/1.9781611974331.ch64}
\BIBentrySTDinterwordspacing

\bibitem{broue2024}
M.~Broué, \emph{From Rings and Modules to Hopf Algebras}.\hskip 1em plus 0.5em
  minus 0.4em\relax Springer, 2024.

\bibitem{kieferNotesEquivalenceMinimization2020}
\BIBentryALTinterwordspacing
S.~Kiefer, ``Notes on {{Equivalence}} and {{Minimization}} of {{Weighted
  Automata}},'' Sep. 2020. [Online]. Available:
  \url{http://arxiv.org/abs/2009.01217}
\BIBentrySTDinterwordspacing

\bibitem{petrisanAutomataMinimizationFunctorial2020}
\BIBentryALTinterwordspacing
T.~Colcombet and D.~Petri{\c s}an, ``Automata {{Minimization}}: A {{Functorial
  Approach}},'' \emph{Logical Methods in Computer Science}, vol. Volume 16,
  Issue 1, Mar. 2020. [Online]. Available:
  \url{https://lmcs.episciences.org/6213/pdf}
\BIBentrySTDinterwordspacing

\bibitem{Sch2008}
S.~Schröer, ``Baer's result: The infinite product of the integers has no
  basis,'' \emph{Amer. Math. Monthly}, vol. 115, pp. 660--663, 2008.

\bibitem{drosteWeightedAutomata2021}
\BIBentryALTinterwordspacing
M.~Droste and D.~Kuske, ``Weighted automata,'' in \emph{Handbook of {{Automata
  Theory}}}, J.-{\'E}. Pin, Ed.\hskip 1em plus 0.5em minus 0.4em\relax European
  Mathematical Society Publishing House, Z{\"u}rich, Switzerland, 2021, pp.
  113--150. [Online]. Available: \url{https://doi.org/10.4171/Automata-1/4}
\BIBentrySTDinterwordspacing

\bibitem{adamekAbstractConcreteCategories2009}
J.~Ad{\'a}mek, H.~Herrlich, and G.~E. Strecker, \emph{Abstract and {{Concrete
  Categories}}: {{The Joy}} of {{Cats}}}, ser. Dover Books on Advanced
  Mathematics.\hskip 1em plus 0.5em minus 0.4em\relax Dover Publications, Aug.
  2009.

\bibitem{Adamek-coalg}
J.~Ad{\'a}mek, F.~Bonchi, M.~H{\"u}lsbusch, B.~K{\"o}nig, S.~Milius, and
  A.~Silva, ``A coalgebraic perspective on minimization and determinization,''
  in \emph{Foundations of Software Science and Computational Structures},
  L.~Birkedal, Ed.\hskip 1em plus 0.5em minus 0.4em\relax Berlin, Heidelberg:
  Springer Berlin Heidelberg, 2012, pp. 58--73.

\bibitem{Rot2016}
\BIBentryALTinterwordspacing
J.~Rot, ``Coalgebraic minimization of automata by initiality and finality,''
  \emph{Electronic Notes in Theoretical Computer Science}, vol. 325, pp.
  253--276, 2016, the Thirty-second Conference on the Mathematical Foundations
  of Programming Semantics (MFPS XXXII). [Online]. Available:
  \url{https://www.sciencedirect.com/science/article/pii/S1571066116300937}
\BIBentrySTDinterwordspacing

\bibitem{aristoteFunctorialApproachMinimizing2023b}
\BIBentryALTinterwordspacing
Q.~Aristote, ``Functorial approach to minimizing and learning deterministic
  transducers with outputs in arbitrary monoids,'' Nov. 2023. [Online].
  Available: \url{https://ens.hal.science/hal-04172251v3}
\BIBentrySTDinterwordspacing

\bibitem{biasseComputationHNFModule2017}
\BIBentryALTinterwordspacing
J.-F. Biasse, C.~Fieker, and T.~Hofmann, ``On the computation of the {{HNF}} of
  a module over the ring of integers of a number field,'' \emph{Journal of
  Symbolic Computation}, vol.~80, pp. 581--615, May 2017. [Online]. Available:
  \url{https://www.sciencedirect.com/science/article/pii/S0747717116300736}
\BIBentrySTDinterwordspacing

\bibitem{cohen2012advanced}
H.~Cohen, \emph{Advanced topics in computational number theory}.\hskip 1em plus
  0.5em minus 0.4em\relax Springer, 2000.

\bibitem{bosmapohst1991}
\BIBentryALTinterwordspacing
W.~Bosma and M.~Pohst, ``Computations with finitely generated modules over
  dedekind rings,'' in \emph{Proceedings of the 1991 International Symposium on
  Symbolic and Algebraic Computation}, ser. ISSAC '91.\hskip 1em plus 0.5em
  minus 0.4em\relax New York, NY, USA: Association for Computing Machinery,
  1991, p. 151–156. [Online]. Available:
  \url{https://doi.org/10.1145/120694.120714}
\BIBentrySTDinterwordspacing

\bibitem{cohen2013course}
H.~Cohen, \emph{A course in computational algebraic number theory}.\hskip 1em
  plus 0.5em minus 0.4em\relax Springer, 1993.

\bibitem{cohenHermiteSmithNormal1996}
\BIBentryALTinterwordspacing
------, ``Hermite and {{Smith}} normal form algorithms over {{Dedekind}}
  domains,'' \emph{Mathematics of Computation}, vol.~65, no. 216, pp.
  1681--1699, 1996. [Online]. Available:
  \url{https://www.ams.org/mcom/1996-65-216/S0025-5718-96-00766-1/}
\BIBentrySTDinterwordspacing

\bibitem{ge1994recognizing}
G.~Ge, ``Recognizing units in number fields,'' \emph{Mathematics of
  computation}, vol.~63, no. 207, pp. 377--387, 1994.

\bibitem{stein2012algebraic}
W.~Stein, ``Algebraic number theory, a computational approach,'' \emph{Harvard,
  Massachusetts}, 2012.

\bibitem{maclane1971}
S.~Mac~Lane, \emph{Categories for the working mathematician}.\hskip 1em plus
  0.5em minus 0.4em\relax Springer-Verlag New York Heidelberg Berlin, 1971.

\bibitem{domichHermiteNormalForm1987}
\BIBentryALTinterwordspacing
P.~D. Domich, R.~Kannan, and L.~E. Trotter, ``Hermite {{Normal Form Computation
  Using Modulo Determinant Arithmetic}},'' \emph{Mathematics of Operations
  Research}, vol.~12, no.~1, pp. 50--59, Feb. 1987. [Online]. Available:
  \url{https://pubsonline.informs.org/doi/10.1287/moor.12.1.50}
\BIBentrySTDinterwordspacing

\bibitem{sonn1967theorem}
J.~Sonn and H.~Zassenhaus, ``On the theorem on the primitive element,''
  \emph{The American Mathematical Monthly}, vol.~74, no.~4, pp. 407--410, 1967.

\bibitem{bach1993factor}
E.~Bach, J.~Driscoll, and J.~Shallit, ``Factor refinement,'' \emph{Journal of
  Algorithms}, vol.~15, no.~2, pp. 199--222, 1993.

\bibitem{blomer1998probabilistic}
J.~Bl{\"o}mer, ``A probabilistic zero-test for expressions involving roots of
  rational numbers,'' in \emph{Algorithms—ESA’98: 6th Annual European
  Symposium Venice, Italy, August 24--26, 1998 Proceedings 6}.\hskip 1em plus
  0.5em minus 0.4em\relax Springer, 1998, pp. 151--162.

\bibitem{buchmann1999factor}
J.~Buchmann and F.~Eisenbrand, ``On factor refinement in number fields,''
  \emph{Mathematics of computation}, vol.~68, no. 225, pp. 345--350, 1999.

\end{thebibliography}

\newpage
\appendices

\crefalias{section}{appendix}
\crefalias{subsection}{subappendix}

\onecolumn

\section{Extended Algebraic Preliminaries}
\label{app:preliminaries}
\subsection{Commutative Algebra}
\label{app:preliminaries:comm-algebra}

\subsubsection{Rings}

Throughout the paper, `ring' means `commutative ring with unit'. That is, a ring $R$ is a set $R$ equipped with binary operations, addition 
and multiplication, satisfying the following conditions: $R$ is a commutative group under addition, $R$ is a commutative monoid under multiplication, and multiplication is distributive with respect to addition. An element $r$ of $R$ is \emph{invertible} if there exists $s \in R$ such that $rs = 1$. We write $\inver{R}$ for the set of invertible elements of $R$.

\medskip
\subsubsection{Modules}

Let $R$ be a ring. An \emph{$R$-module} is an additive group $M$ equipped with an action $R \times M \to M$, called \emph{scalar multiplication}, which is unital, associative, and bilinear: for any $r, s \in R$ and $m, n \in M$: $1_R m = m$, $r(sm) = (rs)m$, and both $(r+s)m = rm + sm$ and $r(m+n) = rm + rn$. A \emph{submodule} of a module $M$ is a subgroup $N$ of $M$ that is invariant under  scalar multiplication: for any $r \in R$ and $n \in N$, we have $rn \in N$. 

For any subset $S$ of $M$, the \emph{submodule generated by} $S$ is the smallest submodule of $M$ that contains $S$. The \emph{sum} of two submodules $M_1, M_2$ of $M$ is the submodule generated by the union $M_1 \cup M_2$, and is denoted $M_1 + M_2$. If moreover $M_1 \cap M_2 = \{0_M\}$, then $M_1 + M_2$ is a \emph{direct sum}, and it is then denoted $M_1 \oplus M_2$. Note that $R$ itself can be seen as an $R$-module. Its submodules are exactly the \emph{ideals} of $R$, and the ideals generated by a singleton are called \emph{principal}. A \emph{morphism} or \emph{$R$-linear map} between $R$-modules $M$ and $N$ is a group homomorphism $\phi \colon M \to N$ that respects the actions, i.e., $\phi(rm) = r\phi(m)$ for all $r \in R$ and $m \in M$.

\smallskip

\emph{Free} modules are essential for the algebraic analysis of $R$-weighted automata. For any (possibly infinite) set $Q$, a \emph{finite $R$-linear combination} is an expression $\sum_{i=1}^n r_i q_i$, with $r_i \in R$ and $q_i \in Q$ for each $1 \leq i \leq n$; said otherwise, it is a finitely supported function from $Q$ to $R$. The set of finite $R$-linear combinations, equipped with the expected addition and $R$-action, is a first example of a free $R$-module; we denote it by $\frMod{R}{Q}$. The elements $1_R \cdot q$, for $q \in Q$, constitute a \emph{canonical basis} for $\frMod{R}{Q}$; we identify each $q$ with its corresponding basis element. The \emph{rank} of $\frMod{R}{Q}$ is the cardinality of $Q$. 

For any $R$-module $M$ and $(x_q)_{q \in Q}$ a $Q$-indexed set of elements of $M$, there exists a unique morphism $\phi \colon \frMod{R}{Q} \to M$ that sends 
$q$ to $x_q$. 
The family $(x_q)_{q \in Q}$ is \emph{free} in $M$ if $\phi$ is injective, \emph{generating} for $M$ if $\phi$ is surjective, and a \emph{basis} for $M$ if $\phi$ is bijective. The module $M$ is called \emph{finitely generated} if it has a finite generating family and \emph{free} if it has a basis. An element $m \in M$ has \emph{torsion} if there exists $r \neq 0$ such that $rm = 0$. An $R$-module $M$ is \emph{torsion-free} if there are no non-zero torsion elements. Free $R$-modules are torsion-free, but the converse is not true; for example, the ideal $(2, 1 + i\sqrt{5})$ in the ring $\adjoined{i\sqrt{5}}$ is not free. Note that a submodule of a torsion-free module is again torsion-free. 

\medskip
\subsubsection{Field of fractions}

A ring $R$ is an \emph{integral domain} if it has no zero divisors, i.e., for any $r, s \in R$, if $rs = 0$, then $r = 0$ or $s = 0$. Let $R$ be an integral domain. 
The \emph{field of fractions} of $R$ is obtained by adding formal inverses for all non-zero elements of $R$: its elements are 
pairs $r/s$, where $r \in R$ and $s \in \inver{R}$, and $r/s$ is identified with $r'/s'$ when $rs' = r's$; note that it is in particular an $R$-module. Let $\KK$ be the field of fractions of $R$. 
Extending this definition to modules, for any $R$-module $M$, there is a $\KK$-vector space  $\localize M$, called the \emph{localization} or \emph{extension of scalars} of $M$. Its elements are formal fractions $m/r$ of an element $m \in M$ and an element $r \in \inver{R}$, where we identify $m/r$ and $m'/r'$ if there is $s
\in \inver{R}$ such that $sr \cdot m' = sr' \cdot m$. For any $R$-linear map $f: M \to N$, we obtain a $\KK$-linear map $\localize f: \localize M \to \localize N$, which maps $m/r$ to $f(m)/r$.
The \emph{rank} of an $R$-module is the dimension of its localization.  

\medskip
\subsubsection{Ideals}

Given a ring~$R$, a subset $I\subseteq R$ is said to be an \emph{ideal} if $I$ is a subgroup of $R$ under addition, and is closed under multiplication by elements of~$R$, i.e., for all $r\in R$ and $a\in I$, we have $ra\in I$.
An ideal $I$ of $R$ is \emph{proper} if it is not equal to $R$. A proper ideal~$I$ is called a \emph{prime} ideal if for any $a,b\in R$ such that $ab\in I$, at least one of $a$ or $b$ is in $I$. The ring $R$ is called a \emph{principal ideal domain (PID)} if every ideal $I\in R$ is generated by a single element, that is, there exists $a \in R$ such that $I=\{ra \,| \, r\in R\}=(a)$.  

A \emph{fractional ideal} of a ring $R$ is a non-zero $R$-submodule $\mathfrak{a}$ of $\KK$ such that, for some $r \in R \setminus \{0\}$, the submodule $r \mathfrak{a} \coloneq \{ra \ | \ a \in \mathfrak{a}\}$ is contained in $R$. The \emph{dual} of a fractional ideal $\mathfrak{a}$ is the fractional ideal $\mathfrak{a}^{-1} \coloneq \{x \in \KK \ | \ x \mathfrak{a} \subseteq R\}$, and $\mathfrak{a}$ is \emph{invertible} if $\mathfrak{a}^{-1} \mathfrak{a} = R$. An integral domain $R$ is a \emph{Dedekind domain} if, and only if, every fractional ideal is invertible.

\medskip
\subsubsection{On PIDs and Dedekind domains}

Any PID is a Dedekind domain, and any Dedekind domain $R$ is \emph{Noetherian}, meaning that any ideal of $R$ is finitely generated. An $R$-module $M$ is \emph{Noetherian} if every submodule of $M$ is finitely generated. Equivalently, a ring $R$ is Noetherian if every ascending chain of ideals in $R$ stablizes, and an $R$-module $M$ is Noetherian if every ascending chain of submodules of $M$ stabilizes. In a Noetherian ring, all finitely generated modules are Noetherian.

In $R$ is a PID, then every finitely generated torsion-free $R$-module is free, and thus has a basis. If $R$ is a Dedekind domain with field of fractions $\KK$, then every finitely generated torsion-free $R$-module $M$ is a direct sum of rank $1$ submodules, $M = \bigoplus_{i=1}^n E_i$, where each $E_i$ is of the form $\mathfrak{a}_ie_i$, for some fractional ideal $\mathfrak{a}_i$ of $\KK$ and $e_i \in M$, such that the system $(e_i)_{i=1}^n$ is free, see, e.g., \cite[Cor.~10.2.3]{broue2024}. Such a set of rank $1$ submodules $E_1, \dots, E_n$ is called a \emph{pseudo-basis} for $M$.

\section{Algebraic Number Fields}
\label{app:algNF}
Denote by $\QQ[x]$ the ring of univariate polynomials with rational coefficients. 
A  non-zero polynomial in $\QQ[x]$ is \emph{monic}  when its leading coefficient is equal to 1. 

\medskip 

Recall that a complex number $\alpha\in \CC$ is algebraic if it is the root of a non-zero polynomial in~$\QQ[x]$. The minimal polynomial of~$\alpha$ (over~$\QQ$) is the  monic polynomial of least degree in~$\QQ[x]$ that has $\alpha$ as a root. The degree of $\alpha$ is defined to be the degree of its minimal polynomial. The roots of the
minimal polynomial of $\alpha$  are called the Galois conjugates of~$\alpha$. 
The (absolute) norm of~$\alpha$,
denoted $\norm{\alpha}$, is the product of the Galois conjugates of~$\alpha$.

A number field~$\KK$ is a subfield of~$\CC$ having finite degree over~$\QQ$, meaning it can be viewed as a finite dimensional vector space over~$\QQ$. By the primitive element theorem,  every number field has the form $\QQ[\theta]$ for some algebraic number~$\theta$, where
\[\QQ[\theta]\coloneq\left\{\sum_{i=1}^d a_i \theta^i ~\middle|~ a_i \in \QQ\right\}\]
and $d$ is the degree of~$\theta$. We say that~$\theta$ is a primitive element for~$\KK$, and define the degree of~$\KK$ over~$\QQ$, denoted by~$[\KK : \QQ]$, to be~$d$.

An algebraic integer~$\alpha$ is \emph{integral} (over~$\ZZ$) if its minimal polynomial is in $\ZZ[x]$. 
The set of all algebraic integers $\mathbb{A} \subset \CC$ is closed under addition, subtraction and multiplication and therefore is a commutative ring.

Given any number field~$\KK$, its ring of integers~$\OK\coloneq\KK \cap \mathbb{A}$ can be characterized as a $\ZZ$-module  of rank~$[\KK:\QQ]$.
Since $\ZZ$ is a PID, each $\ZZ$-submodule of $\OK$ is finitely generated; by this, each ideal~$\mathfrak{a}\subseteq \OK$ can  be viewed as a finitely generated $\ZZ$-module. Indeed, 
by~\cite[Proposition 4.6.3]{cohen2013course}, each   ideal~$\mathfrak{a}\subseteq \OK$ is even a full-rank $\ZZ$-submodules of~$\OK$. 
The finitely generated full-rank $\ZZ$-modules $\mathfrak{b}$ of $\KK$ are called \emph{fractional ideals}.  
For clarity, we sometimes refer to 
 ideals~$\mathfrak{a}\subseteq \OK$ as  \emph{integral ideals}.
 The set of fractional ideals of $\KK$ forms a commutative group~\cite[Theorem 4.6.14]{cohen2013course} under ideal multiplication with identity element $\OK$. Hence,
  the inverse of $\mathfrak{a}$ is defined as $\mathfrak{a}^{-1}\coloneq\{\alpha \in \KK \, |\, \alpha  \mathfrak{a} \subseteq \OK\}$. Moreover, for each fractional ideal~$\mathfrak{b}$, there exists $r\in \ZZ$ such that $r\mathfrak{b}$ is integral.

The number rings $\OK$ are not in general unique factorization domains, where 
every (non-zero) element can be uniquely written as a product of irreducible elements; see \Cref{example-cyclotomic} and \Cref{example-quadratic} for instance. Ideals in $\OK$ when seen  
as $\OK$-modules are finitely generated, but might not have a basis. In this manner number rings differ from PIDs which, on the contrary, form a strict subset of unique factorization domains, and have all their ideals have a basis by definition.

In fact, by~\cite[Proposition 1.2.3]{cohen2012advanced}, number rings are Dedekind domains. In particular, each fractional ideal is equal to a unique product of powers of prime ideals and can be generated by at most two elements.

\begin{example}[Cyclotomic fields]
\label{example-cyclotomic}
An $n$-th root of unity $\alpha$ is a root of the polynomial~$x^n-1$. It is called \emph{primitive} if it is not an $m$-th root of unity for any $m<n$.

Let $\zeta_n$ be a primitive  $n$-th root
of unity.
The minimal polynomial~$\Phi_n$ of $\zeta_n$ is known as the $n$-th \emph{cyclotomic} polynomial. When $n$ is a (rational) prime, $\Phi_n$ takes the form $\Phi_n(x)=1+x+x^2+ \ldots+x^{n-1}$.

The $n$-th cyclotomic field is $\adjoinedff{\zeta_n}$, whose ring of integers is $\adjoined{\zeta_n}$.
An example is the Gaussian integers~$\adjoined{i}$ which is the ring of integers of $\adjoinedff{i}$. 
The  ideal $(1 + 2i)\adjoined{i}$ of $\adjoined{i}$ can be viewed as
\begin{itemize}
    \item  a $\adjoined{i}$-module generated by $(1+2i)$, 
    \item  a $\ZZ$-module generated by $1+2i$  and $-2+i$. 
\end{itemize}

The ring $\adjoined{i}$ is a PID, but in general the cyclotomic rings $\mathbb {Z} [\zeta_{n}]$ are not PIDs. For example, $\mathbb {Z} [\zeta_{23}]$ 
does not satisfy unique factorization of elements, meaning it is not a UFD and therefore cannot be a PID. 
All number rings are Dedekind domains, where the unique factorization of elements is replaced by the unique factorization of ideals.
\end{example}

\begin{example}[Quadratic fields]
\label{example-quadratic}
A number field~$\KK$ is quadratic if there exists a square-free
$\beta$ such that $K=\QQ[\sqrt{\beta}]$. 
The ring of integers~$\OK$ of $\KK$ 
is  of the form~$\adjoined{\theta}$,
where 
\[\theta=\begin{cases}
    \sqrt{\beta} & \text{if } \beta \not \equiv 1 \pmod{4}\,, \\
    \frac{\sqrt{\beta}-1}{2} & \text{if } \beta  \equiv 1 \pmod{4} \,.
\end{cases}\]
For example, the ring of integers $\OK$ of~$\KK=\adjoinedff{\sqrt{5}}$ 
is~$\adjoined{\frac{1+\sqrt{5}}{2}}$. 
The subring $\adjoined{\sqrt{5}}$ is called an order (as it has finite index in $\OK$). 
The order
$\adjoined{\sqrt{5}}$ is not a PID, as shown by the factorizations 
\[4 \;=\; 2 \cdot 2 \;= \; (3+\sqrt{5})(3-\sqrt{5})\,.\]
\end{example}

An order is a subring of $\KK$ whose additive group is free of rank~$n$. The maximal order in $\KK$ is simply the ring $\OK$ of integers. While the number rings $\OK$ are Dedekind Domains, their proper orders are not necessarily.
This is because Dedekind domains require all nonzero ideals to factor uniquely into prime ideals, whereas proper orders may fail this property due to the presence of non-invertible ideals\footnote{Due to this fact, our learning algorithm for WAs over $\OK$, in the current form, cannot be used for learning WAs over orders of $\OK$.}.

\section{Proofs for \texorpdfstring{\Cref{sec:intro}}{Section 1}}
\label{app:details-for-intro}

\exNonFatouNumberRing*

\begin{proof}
It is easy to see that these two automata are equivalent once we notice that they compute the same values for words of length 2, using that $\left(2 - i\sqrt{5}\right)\left(i\sqrt{5} - 1\right) = 3\left(1 + i\sqrt{5}\right)$. The automaton of \Cref{fig:example-wa:ok} is the one constructed in the proof of \Cref{prop:minimal-pip-hard} for the ideal $\mathfrak{a} = \left(3, 2 - i\sqrt{5}\right)$. It is indeed easy to check that \[ \mathfrak{a}^{-1} = \tfrac{1}{3} \mathopen{}\left(1 - i\sqrt{5}, 2 + i\sqrt{5}\right)\mathclose{} \]
and there is an isomorphism
\[ \phi: (3, 2-i\sqrt{5}) \oplus \tfrac{1}{3}\mathopen{}\left(1 - i\sqrt{5}, 2 + i\sqrt{5}\right)\mathclose{} \cong \adjoined{i\sqrt{5}}^2 \]
given by the formula 
\[ \begin{pmatrix} a \\ b \end{pmatrix} \mapsto a \begin{pmatrix}\frac{i\sqrt{5} - 1}{3} \\ 1 \end{pmatrix} + b \begin{pmatrix} i\sqrt{5} - 2 \\ 3 \end{pmatrix} \]
(notice that $\phi \begin{pmatrix} 3 \\ -1 \end{pmatrix} = \begin{pmatrix} 1 \\ 0 \end{pmatrix}$ and $\phi \begin{pmatrix} 2 - i\sqrt{5} \\ \frac{i\sqrt{5} - 1}{3}\end{pmatrix} = \begin{pmatrix} 0 \\ 1 \end{pmatrix}$ so that the image of $\phi$ is indeed precisely $\adjoined{i\sqrt{5}}^2$).

We now show that $\mathfrak{a}$ is not principal and thus, by the proof of \Cref{prop:minimal-pip-hard}, that this $\adjoined{i\sqrt{5}}$-weighted automaton is indeed minimal. An integral basis for $\adjoined{i\sqrt{5}}$ is given by the two elements $1$ and $i\sqrt{5}$. In this integral basis, the two generators of $\mathfrak{a}$ are the vectors $\begin{pmatrix}
    3 \\ 0
\end{pmatrix}$ and $\begin{pmatrix}
    2 \\ -1
\end{pmatrix}$. By \cite[\textsection 4.6.1]{cohen2013course}, putting these two vectors in a matrix, the norm $\norm{\mathfrak{a}}$ of $\mathfrak{a}$ is the absolute value of the determinant of this matrix, namely $3$. And by \emph{loc. cit.}, if $\mathfrak{a}$ were principal and generated by some $a + bi\sqrt{5}$, we would then have $3 = \norm{\mathfrak{a}} = a^2 + 5b^2$ with $a, b \in \ZZ$: this is clearly impossible.
\end{proof}

\section{Proofs for \texorpdfstring{\Cref{sec:weighted-cat}}{Section 3}}
\label{app:proofs-for-weighted-cat}
\almostMinDedekindWA*

\begin{proof}
    We have already argued that $\Min L$ has a pseudo-basis. In fact, by \cite[Theorem 1.2.19]{cohen2012advanced} we may even choose this pseudo-basis so that all except one of the fractional ideals are $R$: we write $\Min L \cong R^{n-1} \oplus \mathfrak{a}$ for some fractional ideal $\mathfrak{a}$. Because this is an isomorphism, the $R$-modular automaton structure on $\Min L$ can be conjugated to equip $R^{n-1} \oplus \mathfrak{a}$ with an $R$-modular automaton structure as well. Write $\A$ for this $R$-modular automaton.

    By \cite[Lemma 1.2.20]{cohen2012advanced}, there is an isomorphism $\mathfrak{a} \oplus \mathfrak{a}^{-1} \cong R^2$. Extend this isomorphism to get an injection $m: R^{n-1} \oplus \mathfrak{a} \rightarrowtail R^{n+1}$ and a surjection $e: R^{n+1} \twoheadrightarrow R^{n-1} \oplus \mathfrak{a}$ such that $e \circ m$ is the identity. We can now conjugate $\A$ into an $R$-modular automaton $\A'$ such that $\A'(\St) = R^{n+1}$ by setting $\A'(\triangleright) = m \circ \A(\triangleright)$, $\A'(\sigma) = m \circ \A(\sigma) \circ e$ for all $\sigma \in \Sigma$ and $\A'(\triangleleft) = \A(\triangleleft) \circ e$. 

    Because $\A'(\St)$ is a free module of rank $n+1$ and because $e$ is a morphism of automata $\A' \to \A$, $\A'$ is an $R$-weighted automaton with $n+1$ states that computes $L$.
\end{proof}

\section{Category theory preliminaries}
\label{app:cat}

In this section, we recall the notions from category theory used in this paper. More details can be found in any standard text on category theory, such as~\cite{maclane1971}~or~\cite{adamekAbstractConcreteCategories2009}. As a running example we recall in particular the functorial framework for automata minimization of \cite{colcombetpetrisan2017}.

\subsection{Categories, functors, natural transformations, and examples}
A \emph{category} $\C$ is a pair of classes $(\ob \C, \mor \C)$, called the \emph{objects} and \emph{morphisms} of $\C$, respectively, equipped with the following structure:
\begin{itemize}
    \item for any morphism $f$ of $\C$, a \emph{domain} object $\dom(f)$ and a \emph{codomain} object $\cod(f)$; the notation $f \colon A \to B$ means that $\dom(f) = A$ and $\cod(f) = B$;
    \item for any object $C$ of $\C$, a distinguished morphism $\id_C \colon C \to C$, the \emph{identity} on $C$;
    \item for any morphisms $f \colon A \to B$ and $g \colon B \to C$ of $\C$, a morphism $g \circ f \colon A \to C$, the \emph{composition} of $f$ and $g$; 
\end{itemize}
satisfying the following axioms:
\begin{itemize}
    \item composition is associative: for any $f \colon A \to B$, $g \colon B \to C$, $h \colon C \to D$, we have $h \circ (g \circ f) = (h \circ g) \circ f$;
    \item the identities are neutral elements: for any $f \colon A \to B$, $f \circ \id_A = f = \id_B \circ f$.
\end{itemize}
For objects $A$ and $B$ in a category $\C$, we write $\C(A,B)$ for the class of morphisms from $A$ to $B$. A category is called \emph{locally small} if $\C(A,B)$ is a set for any objects $A, B$, and \emph{small} if the class $\mor \C$ is a set.

A \emph{subcategory} $\D$ of a category $\C$ is a pair $(\ob \D, \mor \D)$, where $\ob \D$ is a subclass of $\ob \C$ and $\mor \D$ is a subclass of $\mor \C$, such that for any composable pair of morphisms $f, g \in \mor \D$, the composition $g \circ f$ is also in $\mor \D$. The subcategory $\D$ is called \emph{full} if, for any objects $A, B \in \ob \D$, we have $\D(A,B) = \C(A,B)$.

A morphism $f \colon A \to B$ is a \emph{monomorphism}, or just \emph{mono}, if it is left-cancellable, i.e., for any $g_1,g_2 \colon A' \to A$, if $f \circ g_1 = f \circ g_2$ then $g_1 = g_2$. Analogously, a morphism is an \emph{epimorphism}, or just \emph{epi}, if it is right-cancellable. A morphism $f \colon A \to B$ is an \emph{isomorphism}, or just \emph{iso}, if it admits a two-sided inverse, i.e., there exists $g \colon B \to A$ such that $g \circ f = \id_A$ and $f \circ g = \id_B$; in this case, the objects $A$ and $B$ are \emph{isomorphic}, and we write $A \cong B$. In any category, isomorphisms are both mono and epi, but there may in general be morphisms that are both mono and epi without being iso.

\begin{example}\label{example:categories}
The following are examples of locally small categories that are relevant to this paper:
    \begin{enumerate}
        \item The category $\Set$ of sets, i.e., $\ob \Set$ is the class of sets and, for any sets $A,B$, $\Set(A,B)$ is the set of morphisms from $A$ to $B$.
        \item For any ring $R$, the category $\Mod{R}$ of $R$-modules, i.e., $\ob \Mod{R}$ is the class of $R$-modules and, for any $R$-modules $M$ and $N$, $\Mod{R}(M,N)$ is the set of $R$-linear maps from $M$ to $N$.
        \item For any field $\KK$, the category $\Vec{\KK}$ of $\KK$-vector spaces, which is just the category $\Mod{R}$ in the special case $R = \KK$.
        \item For any ring $R$, the category $\FreeMod{R}$ of free $R$-modules, i.e., the full subcategory of $\Mod{R}$ whose objects are the free $R$-modules.
        \item For any ring $R$, the category $\TFMod{R}$ of torsion-free $R$-modules, i.e., the full subcategory of $\Mod{R}$ whose objects are the torsion-free $R$-modules.
        \item A \emph{directed multigraph} is a structure $G = (V,E)$, where $V$ is a set of nodes and $E$ is a multiset of pairs of nodes. For any directed multigraph $G$, there is a \emph{free category} $\FreeCat{G}$ on $G$, which can be realized concretely as follows: the objects of $\FreeCat{G}$ are the nodes of $G$, and for any nodes $v,w$ of $G$, the set of morphisms  $\FreeCat{G}(v,w)$ consists of the finite paths from $v$ to $w$ in $G$. 
        For any node $v$ of $G$, the identity $\id_v$ on $v$ is the empty path. The composition of paths $e \colon v \to w$ and $f \colon w \to x$ is defined by concatenation.
        \item In particular, for fixed finite alphabet $\Sigma$, we denote by $\I$ the category free over the multigraph
\[
\begin{tikzcd}[ampersand replacement = \&]
\In \arrow[r, "\triangleright"] \& \St \arrow[r, "\triangleleft"] \arrow["\sigma (\sigma\in\Sigma)", loop, distance=2em, in=55, out=125] \& \Out\, .
\end{tikzcd}
\]
Concretely, the morphisms of $\I$ can be described as follows: in addition to the three identity morphisms $\id_{\In}$, $\id_{\St}$ and $\id_{\Out}$,  for any word $w \in \Sigma^+$, we have four morphisms: $w \colon \St \to \St$,  $\triangleright w \colon \In \to \St$, $w\triangleleft \colon \St \to \Out$, and $\triangleright w \triangleleft \colon \In \to \Out$. 

We denote by $\O$ the full subcategory of $\I$ on the objects $\In$ and $\Out$: it consists in these two objects, and, identities excluded, a morphism $\triangleright w \triangleleft: \In \to \Out$ for every $w \in \Sigma^*$.
    \end{enumerate}
\end{example}

Let $\C$ and $\D$ be categories. A \emph{functor} $\F \colon \C \to \D$ is a pair of mappings, one from $\ob \C$ to $\ob \D$ and one from $\mor \C$ to $\mor \D$, such that the following properties hold:
\begin{itemize}
    \item for any morphism $f \colon A \to B$ in $\C$, we have $\F f \colon \F A \to \F B$, i.e., $\dom(\F f) = FA$ and $\cod(\F f) = \F B$;
    \item for any object $A$ of $\C$, $\F(\id_A) = \id_{\F A}$;
    \item for any morphisms $f \colon A \to B$, $g \colon B \to C$ in $\C$, we have $\F(g \circ f) = \F g \circ \F f$.
\end{itemize}
If $\F \colon \C \to \D$ and $\G \colon \D \to \E$ are functors, then their \emph{composition} is the functor $\G \circ \F \colon \C \to \E$ which sends an object $A$ of $\C$ to the object $\G\F A$ and a morphism $f \colon A \to B$ of $\C$ to the morphism $\G\F f$. The \emph{identity functor} on a category $\C$ is the functor $\C \to \C$ which is the identity on both objects and morphisms. 

\begin{example}\label{exa:functors}
    \begin{enumerate}
        \item \label{it:forgetful} There is a functor $\mathcal{U} \colon \Mod{R} \to \Set$ which sends an $R$-module $M$ to its underlying set, and an $R$-linear map $M \to N$ to its underlying function. This is called the \emph{forgetful functor} from $\Mod{R}$ to $\Set$. There are many other categories whose objects are `sets with additional structure' and whose morphisms are `functions preserving the additional structure', and they always admit an analogous forgetful functor to $\Set$. %
        \item Let $R$ be a ring with field of fractions $\KK$. The construction which associates to a torsion-free $R$-module $M$ the localization $\localize M$ and to an $R$-linear map $\phi \colon M \to N$ the vector space map $\localize \phi \colon \localize M \to \localize N$ (see \Cref{sec:comm-alg-primer}) is the \emph{localization functor} $\localize - \colon \TFMod{R} \to \Vec{\KK}$.
        
        In the other direction, we have a functor $-_R: \Vec{\KK} \to \TFMod{R}$, called \emph{restriction of scalars}, which views a $\KK$-vector space as only being an $R$-module and a $\KK$-linear    transformation as only being $R$-linear.
        \item Let $\C$ be a category and $I, O$ two objects of $\C$. Recall that a $(\C, I, O)$-automaton is, by definition, a functor $\A \colon \I \to \C$ such that $\A(\In) = I$ and $\A(\Out) = O$.
        \item A subcategory $\C$ of a category $\D$ comes with a functor that sends every object of $\C$ to the corresponding object of $\D$, and every morphism of $\C$ to the corresponding morphism of $\D$. 
        
        Write $\iota: \O \to \I$ for this functor in the case where $\C = \O$ and $\D = \I$. For any $(\C,I,O)$-automaton given as a functor $\A: \I \to \C$, the composite functor $\L = \A \circ \iota: \O \to \C$ is the $(\C,I,O)$-\emph{language recognized} or \emph{computed} by $\A$.
    \end{enumerate}
\end{example}

    Let $\C$ and $\D$ be categories and $\F, \G \colon \C \to \D$ functors. A \emph{natural transformation} $\alpha$ from $\F$ to $\G$, written $\alpha \colon \F \Rightarrow \G$, is an $\ob \C$-indexed family of $\D$-morphisms, $(\alpha_C \colon \F C \to \G C)_{C \in \ob \C}$, such that, for any morphism $f \colon A \to B$ in $\C$, we have $(\G f) \circ \alpha_A = \alpha_B \circ (\F f)$, i.e., the following diagram commutes:
\[
\begin{tikzcd}
\F A \arrow[r, "\F f"] \arrow[d, "\alpha_A"] & \F B  \arrow[d, "\alpha_B"] \\
\G A \arrow[r, "\G f"] & \G B
\end{tikzcd}
\]
For any categories $\C$ and $\D$, the \emph{functor category} $\Fun(\C,\D)$ has as its objects the functors from $\C$ to $\D$, and, given two such functors $\F$ and $\G$, the class of morphisms $\Fun(\C,\D)(\F,\G)$ consists of the natural transformations from $\F$ to $\G$. For morphisms $\alpha \colon \F \to \G$ and $\beta \colon \G \to \mathcal{H}$ in $\Fun(\C,\D)$, the composition $\beta \circ \alpha$ is defined to be the family $(\beta_C \circ \alpha_C)_{C \in \ob \C}$. For any functor $\F$, the identity $\id_\F \colon \F \Rightarrow \F$ is the family $(\id_{\F C})_{C \in \ob C}$.

\begin{remark}
For $\C$ a category and $I,O$ two objects of $\C$, the category of automata, $\Auto{\C,I,O}$, is a subcategory of $\Fun(\I,\C)$, whose objects are the functors that send $\In$ to $I$ and $\Out$ to $O$, and whose morphisms are the natural transformations $\alpha$ for which $\alpha_{\In} = \id_{I}$ and $\alpha_{\Out} = \id_O$. Concretely, such a natural transformation is uniquely given by its component $\alpha_{\St}$, and the naturality diagrams are then the ones given immediately after the definition of automaton morphism in the main text. Given a $(\C,I,O)$-language $\L$, we also write $\Auto{\L}$ for the full subcategory of $\Auto{\C,I,O}$ whose objects are those $(\C,I,O)$-automata computing $\L$.
\end{remark}

\subsection{Products, coproducts, powers, copowers, final and initial objects}
\label{app:cat:co-limits}
\subsubsection{Products, powers, final objects}
    Let $\C$ be a category and let $(A_i)_{i \in I}$ be a family of objects indexed by a set $I$. A \emph{product} of the objects $(A_i)_{i\in I}$ in $\C$ is an object $P$, together with a family of morphisms $(\pi_i \colon P \to A_i)_{i \in I}$, such that, for any object $B$ and $I$-indexed family of morphisms $(f_i \colon B \to A_i)_{i \in I}$, there exists a unique morphism $f \colon B \to P$ such that $\pi_i \circ f = f_i$ for every $i \in I$. In this case, we denote the object $P$ by $\prod_{i\in I} A_i$, and the morphism $f$ by $\langle f_i \rangle_{i \in I}$. In the special case where there is an object $A$ such that $A_i = A$ for all $i \in I$, we denote such an object $P$ by $\prod_{I} A$, and call it the \emph{$I$-fold power of the object $A$}.

To expand upon this definition a bit more, consider the case where $I = \{1,2\}$. The product of two objects $A_1$ and $A_2$ is then denoted $A_1 \times A_2$, and it is an object of $\C$ equipped with two morphisms $p_1 \colon A_1 \times A_2 \to A_1$ and $p_2 \colon A_1 \times A_2 \to A_2$, such that, for any pair of morphisms $f_1 \colon B \to A_1$, $f_2 \colon B \to A_2$, there exists a unique $f \colon B \to A_1 \times A_2$ such that $p_1 \circ f = f_1$ and $p_2 \circ f = f_2$, as in the following diagram:
\[
\begin{tikzcd}
& B \arrow[d,dashed,"f"] \arrow[ddl,bend right,swap,"f_1"] \arrow[ddr,bend left,"f_2"] & \\
& A_1 \times A_2 \arrow[ld, "\pi_1"] \arrow[rd, swap, "\pi_2"] & \\
A_1 & & A_2
\end{tikzcd}
\]
For example, in the category $\Set$, note that this property indeed holds for the Cartesian product of sets, and indeed could be used to define it uniquely, up to a bijection.

Now also consider the special case where $I = \emptyset$. In this case, a product of the empty family in $\C$ is just an object $P$ such that, for any object $B$, there is a unique morphism from $B$ to $P$. Such an object $P$ is called a \emph{final object} of $\C$, and we denote the unique morphism from $B$ to $P$ by $\flip_B$. In the category $\Set$, any singleton set is a final object.

In a general category, one may prove that, if $P$ and $P'$ are both products of the same family of objects, then $P \cong P'$. This is an instance of the general fact that \emph{objects defined by universal properties are unique up to isomorphism}. We can thus harmlessly (for the purposes of this paper) speak about `the' product of the family $(A_i)_{i \in I}$. However, note that, in an arbitrary category, the product of a family of objects may fail to exist; for instance, the category $\I$ above does not have a final object.

\begin{remark}
    Most category-theoretic literature writes $A^I$ for the power, instead of $\prod_{I} A$. However, this clashes with common notation for free modules: although it is true that, in the category $\Mod{R}$, when $Q$ is a \emph{finite} set, the free $R$-module $R^{Q}$ is $\prod_{Q} R$ in the category $\Mod{R}$, when $Q$ is infinite, this is no longer true in general. The free $R$-module on $Q$ only contains the \emph{finitely supported} functions from $Q$ to $R$, and is therefore still a submodule of the categorical power $\prod_{Q} R$, but the latter fails to be a free module in general, see, e.g.,~\cite{Sch2008}.
\end{remark}

Let $\C$ be a category with fixed objects $I$ and $O$ such that $\C$ has countable powers of $O$, and let $\L: \O \to \C$ be a $(\C,I,O)$-language. The final object of $\Auto{\L}$ exists and is given by the following formul\ae:
\begin{align*}
    \Afinal(\L)(\St) &= \prod_{\Sigma^*} O \\
    \Afinal(\L)(\triangleright) &= \langle \L(\triangleright w \triangleleft) \rangle_{w \in \Sigma^*} \\
    \Afinal(\L)(\sigma) &= \langle \pi_{\sigma w} \rangle_{w \in \Sigma^*} \\
    \Afinal(\L)(\triangleleft) &= \pi_{\varepsilon}
\end{align*}

For any other $(\C,I,O)$-automaton $\A$ computing $\L$, the underlying $\C$-arrow $\A(\St) \to \Afinal(\St) = \prod_{\Sigma^*} O$ of the unique morphism $\flip_\A: \A \to \Afinal(\L)$ is given by
\[ (\flip_\A)_\St = \langle\A(w \triangleleft)\rangle_{w\in\Sigma^*} \]

\begin{example}
    Given a ring $R$, the category $\Mod{R}$ has all products. The product of a family of modules $(M_i)_{i \in I}$ is the module whose elements are families $(m_i)_{i \in I}$ with $m_i \in M_i$, equipped with component-wise addition and scalar multiplication. Note that this differs from the usual direct sum, whose elements are families $(m_i)_{i \in I}$ with $m_i \in M_i$ but such that only a finite numbers of the $m_i$'s are non-zero.

    The final $(\Mod{R},R,R)$-automaton computing a $(\Mod{R},R,R)$-language thus exists. $\Afinal(\L)(\St)$ is the module $\powserw{R}$ of power-series, whose elements are functions $\Sigma^* \to R$. $\Afinal(\L)(\triangleright)$ is the linear transformation $R \to \powserw{R}$ that sends $1 \in R$ to the series $w \mapsto \L(\triangleright w \triangleleft)$. $\Afinal(\L)(\sigma)$ sends a series $s \in \powserw{R}$ to its derivative $\sigma^{-1}s: w \mapsto s(\sigma w)$. $\Afinal(\L)(\triangleleft)$ sends a series $s \in \powserw{R}$ to the scalar $s(\varepsilon)$.

    Given any other $(\Mod{R},R,R)$-automaton $\A$ computing $\L$, the unique morphism of $(\Mod{R},R,R)$-automata $\A \to \Afinal(\L)$ is given by the linear transformation $\A(\St) \to \powserw{R}$ that sends some $m \in \A(\St)$ to the series $w \mapsto \A(w\triangleleft)(m)$. In particular if $\A(\St)$ is free and has $Q$ for a basis, this linear transformation can be represented by the infinite $Q \times \Sigma^*$ matrix whose $(q,w)$ entry is $\A(w\triangleleft)(q)$.

    When $R$ is an integral domain, $\powserw{R}$ is a torsion-free module. More generally, a product of torsion-free modules is again torsion-free, so that $\TFMod{R}$ has all products as well, and these products are computed just as in $\Mod{R}$. In particular, the final $(\Mod{R},R,R)$- and $(\TFMod{R},R,R)$-automata computing a language coincide.
\end{example}

\subsubsection{Coproducts, copowers, initial objects}
Dual to products, powers, and final objects, we have coproducts, copowers, and initial objects, which we define now.
    Let $\C$ be a category and let $(A_i)_{i \in I}$ be a family of objects indexed by a set $I$. A \emph{coproduct} of the objects $(A_i)_{i\in I}$ in $\C$ is an object $C$, together with a family of morphisms $(j_i \colon A_i \to C)_{i \in I}$, such that, for any object $B$ and $I$-indexed family of morphisms $(f_i \colon A_i \to B)_{i \in I}$, there exists a unique morphism $f \colon C \to B$ such that $f \circ j_i = f_i$ for every $i \in I$. In this case, we denote the object $C$ by $\coprod_{i\in I} A_i$, and the morphism $f$ by $[f_i]_{i \in I}$. In the special case where there is an object $A$ such that $A_i = A$ for all $i \in I$, we denote such an object $C$ by $\coprod_{I} A$, and call it the \emph{$I$-fold copower of the object $A$}.

In $\Set$, the coproduct of a family of sets is its disjoint union. In $\Mod{R}$, the coproduct of a family of modules is their direct sum. For instance, for any set $Q$, the free $R$-module $R^Q$ is the $Q$-fold copower of the module $R$ in the category $\Mod{R}$, and can thus also be denoted $\coprod_{Q} R$. A coproduct of the empty family is called an \emph{initial object}, i.e., an object $C$ such that, for every object $B$, there is a unique morphism $!_B \colon C \to B$.

Let $\C$ be a category with fixed objects $I$ and $O$ such that $\C$ has countable copowers of $I$, and let $\L: \O \to \C$ be a $(\C,I,O)$-language. The initial object of $\Auto{\L}$ exists and is given by the following formul\ae:
\begin{align*}
    \Ainit(\L)(\St) &= \coprod_{\Sigma^*} I \\
    \Ainit(\L)(\triangleleft) &= \left[ \L(\triangleright w \triangleleft) \right]_{w \in \Sigma^*} \\
    \Ainit(\L)(\sigma) &= \left[ j_{w\sigma} \right]_{w \in \Sigma^*} \\
    \Ainit(\L)(\triangleright) &= j_{\varepsilon}
\end{align*}

For any other $(\C,I,O)$-automaton $\A$ computing $\L$, the underlying $\C$-arrow $\coprod_{\Sigma^*} I = \Ainit(\St) \to \A(\St)$ of the unique morphism $!_\A: \Ainit(\L) \to \A$ is given by
\[ (!_\A)_\St =  \left[\A(\triangleright w)\right]_{w \in \Sigma^*} \]

\begin{example}
    Given a ring $R$, the category $\Mod{R}$ has all coproducts. The coproduct of a family of modules $(M_i)_{i \in I}$ is the usual direct sum of these modules, i.e. the module whose elements are families $(m_i)_{i \in I}$ with $m_i \in M_i$ such that only a finite number of the $m_i$'s, equipped with component-wise addition and scalar multiplication.

    The initial $(\Mod{R},R,R)$-automaton computing a $(\Mod{R},R,R)$-language thus exists. $\Ainit(\L)(\St)$ is the free module $\frMod{R}{\Sigma^*}$, whose elements are finitely-supported functions $\Sigma^* \to R$. $\Ainit(\L)(\triangleright)$ is the linear transformation $R \to \frMod{R}{\Sigma^*}$ that sends $1 \in R$ to the function that sends $\varepsilon$ onto $1$ and every other word onto $0$. $\Ainit(\L)(\sigma)$ sends a function $f \in \frMod{R}{\Sigma^*}$ to the function that sends $w\in\Sigma^*$ to $f(w\sigma)$.     
 $\Ainit(\L)(\triangleleft)$ sends a function $f$ to the (finite) sum $\sum_{w \in \Sigma^*} f(w) \L(\triangleright w \triangleleft)$.

    Given any other $(\Mod{R},R,R)$-automaton $\A$ computing $\L$, the unique morphism of $(\Mod{R},R,R)$-automata $\Ainit(\L) \to \A$ is given by the linear transformation $R^{\Sigma^*} \to \A(\St)$ that sends a finitely-supported function $f: \Sigma^* \to R$ to the (finite) sum $\sum_{w \in \Sigma^*} f(w) \A(\triangleright w)$. If $\A(\St)$ is free and has $Q$ for a basis, this linear transformation can be represented by the infinite $\Sigma^* \times Q$ matrix whose $w$-indexed row is $\A(\triangleright w)$ (written as a row vector in the basis $Q$).

    A coproduct of torsion-free modules is again torsion-free, so that $\TFMod{R}$ has all coproducts as well, and these coproducts are computed just as in $\Mod{R}$. In particular, the initial $(\Mod{R},R,R)$- and $(\TFMod{R},R,R)$-automata computing a language coincide.
\end{example}

\subsection{Factorization systems}\label{app:cat:fac-systems}
    A \emph{factorization system} on a category $\C$ is a pair $(\E,\M)$ of subclasses of $\mor \C$, each closed under composition and containing all the isomorphisms, such that, for any morphism $f: X \to Y$, there exists a factorization $f = m \circ e$ with $e: X \twoheadrightarrow Z$ in $\E$ and $m: Z \rightarrowtail Y$ in $\M$ which is unique up to isomorphism: if $f = m' \circ e'$ with $e': X \twoheadrightarrow Z'$ and $m': Z' \rightarrowtail Y$ is another such factorization, there is a unique isomorphism $\chi: Z \cong Z'$ such that $\chi \circ e = e'$ and $m' \circ \chi = m$.

    Notice in particular how $\E$-morphisms are depicted with $\twoheadrightarrow$ and $\M$-morphisms are depicted with $\rightarrowtail$. It is natural to call the $\M$-morphism $m: Z \rightarrowtail Y$ in the factorization of a morphism $f: X \to Y$ its \emph{$(\E,\M)$-image}, and we do so here.

    \begin{example} The following are examples of factorization systems relevant to this paper:
    \begin{itemize}
        \item the category $\Set$ is equipped with the factorization system $(\Surj,\Inj)$ of surjections and injections: every function factors through its image, the latter being unique up to bijections;
        \item given any ring $R$, the category $\Mod{R}$ is equipped with the factorization system $(\Surj,\Inj)$ of surjective linear transformations and injective linear transformations: every linear transformation factors through its image, the latter being unique up to bijections;
        \item when $R$ is an integral domain, the factorization system $(\Surj,\Inj)$ on $\Mod{R}$ restricts to a factorization system on $\TFMod{R}$: this is because the image of a linear transformation between torsion-free modules, as a submodule of a torsion-free module, is again torsion-free;
        \item unless $R$ is a field, in which case every $R$-module is free, the factorization system $(\Surj,\Inj)$ on $\Mod{R}$ does not restrict to a factorization system on $\FreeMod{R}$: this is because the image of a linear transformation between free modules need not be a free module.
    \end{itemize}
    \end{example}

    We recall a few useful facts about factorization systems, also see, e.g.~\cite[Ch.~14]{adamekAbstractConcreteCategories2009}. Suppose $\C$ is equipped with a factorization system $(\E,\M)$. Then the following \emph{diagonal fill-in property} holds: for any morphisms $e \colon A \twoheadrightarrow B$ in $\E$, $m \colon C \rightarrowtail D$ in $\M$, and arbitrary morphisms $f \colon A \to C$, $g \colon B \to D$ in $\C$, if $m \circ f = g \circ e$, then there exists a unique morphism $d \colon B \to C$ such that both $d \circ e = f$ and $m \circ d = g$. This morphism is then called a \emph{diagonal fill-in} for the commutative square $m \circ f = g \circ e$; the following diagram shows why:
    \[
    \begin{tikzcd}
        A \arrow[r,"e",two heads] \arrow[d,swap,"f"] & B \arrow[d,"g"] \arrow[ld,dashed,swap,"d"] \\
        C \arrow[r,swap,"m", tail] & D
    \end{tikzcd}
    \]

    The intersection $\E \cap \M$ is exactly the class of all isomorphisms. If a composite $g \circ f$ is in $\E$ and $f$ is in $\E$ as well, then $g$ is also in $\E$. Dually, if a composite $g \circ f$ is in $\M$ and $g$ is in $\M$ as well, then $f$ is also in $\M$. %
    Together with the diagonal fill-in property, one may deduce from this the following fact: there is a (necessarily unique) $\M$-morphism between the $(\E,\M)$-image of a composite $g \circ f$ and that of $g$ that makes the following diagram commute:
    \[\begin{tikzcd}
	Y & {\im\, g} & Z \\
	X & {\im (g \circ f)}
	\arrow[two heads, from=1-1, to=1-2]
	\arrow["g", curve={height=-12pt}, from=1-1, to=1-3]
	\arrow[tail, from=1-2, to=1-3]
	\arrow["f", from=2-1, to=1-1]
	\arrow[two heads, from=2-1, to=2-2]
	\arrow[dashed, tail, from=2-2, to=1-2]
	\arrow[tail, from=2-2, to=1-3]
\end{tikzcd}\]
    This generalizes the basic fact that the image of a composite function $g \circ f$ is a subset of the image of the function $g$.

    A final, more involved example of factorization system is the following. 
    Let $\C$ and $\D$ be two categories, and let $(\E,\M)$ be a factorization system on $\D$. In $\Fun(\C,\D)$, let $\E_\Fun$ denote the class of those natural transformations $\alpha$ whose every component $\alpha_X$, with $X$ an object of $\C$, is in $\E$. Dually, let $\M_\Fun$ denote the class of those natural transformations $\alpha$ whose every component $\alpha_X$, with $X$ an object of $\C$, is in $\M$. 
    
    Then $(\E_\Fun,\M_\Fun)$, which we also just write $(\E,\M)$, forms a factorization system on $\Fun(\C, \D)$, for which natural transformations are factored component-wise. More precisely, given $\alpha: \F \Rightarrow \G$ and some $\C$-morphism $f: X \to Y$, denote by $\begin{tikzcd}
        \F X & \mathcal{H} X & \G X
        \arrow["\beta_X", two heads, from=1-1, to=1-2]
        \arrow["\gamma_X", tail, from=1-2, to=1-3]
    \end{tikzcd}$ the $(\E,\M)$ factorization of $\alpha_X$, and similarly along $Y$.
    Denote finally by $\mathcal{H} f$ the diagonal fill-in:
    \[\begin{tikzcd}
	{\F Y} & {\mathcal{H} Y} & {\G Y} \\
	{\F X} & {\mathcal{H} X} & {\G X}
	\arrow["{\beta_Y}", two heads, from=1-1, to=1-2]
	\arrow["{\gamma_Y}", tail, from=1-2, to=1-3]
	\arrow["{\F f}", from=2-1, to=1-1]
	\arrow["{\beta_X}"', two heads, from=2-1, to=2-2]
	\arrow[dashed, from=2-2, to=1-2]
	\arrow["{\gamma_X}"', tail, from=2-2, to=2-3]
	\arrow["{\G f}"', from=2-3, to=1-3]
    \end{tikzcd}\]  
    By uniqueness of the diagonal fill-in, this definition makes $\mathcal{H}$ into a functor and $\beta$ and $\gamma$ into natural transformations. The $(\E_\Fun,\M_\Fun)$-factorization of $\alpha$ is then given by $\alpha = \gamma \circ \beta$.
    
    Taking in particular $\C$ to be the freely generated category $I$ defined in \Cref{example:categories}, this factorization system on $\Fun(\I,\D)$ restricts to a factorization system on $\Auto{\D,I,O}$ where $\I$ and $\O$ are any two objects of $\D$.
    In concrete words: morphisms of $(\D,I,O)$-automata can be factored by factoring their underlying $\D$-morphisms, and the $(\E,\M)$-image of a morphism of $(\D,I,O)$-automata is itself canonically equipped with the structure of a $(\D,I,O)$-automaton. In a diagram, the image $\im \A$ of a morphism of automata $\A \to \A'$ is given as

\[\begin{tikzcd}[column sep=large]
	& {\A(\St)} & {\A(\St)} & {\A(\St)} & {\A(\St)} \\
	I & {(\im \A)(\St)} & {(\im \A)(\St)} & {(\im \A)(\St)} & {(\im \A)(\St)} & O \\
	& {\A'(\St)} & {\A'(\St)} & {\A'(\St)} & {\A'(\St)}
	\arrow[two heads, from=1-2, to=2-2]
	\arrow["{\A(\sigma)}", from=1-3, to=1-4]
	\arrow[two heads, from=1-3, to=2-3]
	\arrow[two heads, from=1-4, to=2-4]
	\arrow[two heads, from=1-5, to=2-5]
	\arrow["{\A(\triangleleft)}", curve={height=-12pt}, from=1-5, to=2-6]
	\arrow["{\A(\triangleright)}", curve={height=-12pt}, from=2-1, to=1-2]
	\arrow["{(\im \A)(\triangleright)}", dashed, from=2-1, to=2-2]
	\arrow["{\A'(\triangleright)}"', curve={height=12pt}, from=2-1, to=3-2]
	\arrow[tail, from=2-2, to=3-2]
	\arrow["{(\im \A)(\sigma)}", dashed, from=2-3, to=2-4]
	\arrow[tail, from=2-3, to=3-3]
	\arrow[tail, from=2-4, to=3-4]
	\arrow["{(\im \A)(\triangleleft)}", dashed, from=2-5, to=2-6]
	\arrow[tail, from=2-5, to=3-5]
	\arrow["{\A'(\sigma)}"', from=3-3, to=3-4]
	\arrow["{\A'(\triangleleft)}"', curve={height=12pt}, from=3-5, to=2-6]
\end{tikzcd}\]

If $\A$ is a $(\D,I,O)$-automaton computing $\L$, we denote in particular by $\Reach \A$ the image of the unique morphism $!_\A: \Ainit(\L) \to \A$, by $\Obs \A$ the image of the unique morphism $\flip_\A: \A \to \Afinal(\L)$ and by $\Min \L$ the image of the unique morphism $\flip_{\Ainit(\L)} = {!_{\Afinal(\L)}} : \Ainit(\L) \to \Afinal(\L)$, and the uniqueness of a factorization entails that $\Reach (\Obs \A) \cong \Obs (\Reach \A) \cong \Min \L$.

\begin{example}
    We describe $\Reach$, $\Obs$ and $\Min$ for $(\Mod{R},R,R)$-automata. Because $\TFMod{R}$ shares the coproducts, products and the factorization system of $\Mod{R}$, this description also applies to $(\TFMod{R},R,R)$-automata. Let $\A$ be a $(\Mod{R},R,R)$-automaton recognizing a $(\Mod{R},R,R)$-language $\L$.
    
    $(\Reach \A)(\St)$ is the submodule of \emph{reachable configurations} of $\A(\St)$, also called its forward module: its elements are those elements $m \in \A(\St)$ for which there is an $R$-linear combination $\sum_{i = 1}^n r_i w_i$ such that $m = \sum_{i = 1}^n r_i \A(\triangleright w_i)(1)$. The injective morphism of automata $\Reach \A \rightarrowtail \A$ is the embedding of this submodule into the bigger module $\A(\St)$.

    $(\Obs \A)(\St)$ is the quotient of $\A(\St)$ by the equivalence relation $m \sim m' \iff \forall w \in \Sigma^*, \A(w \triangleleft)(m) = \A(w \triangleleft)(m')$, also called its backward module -- it is a module because this equivalence relation is a congruence, i.e. is compatible with the module structure of $\A(\St)$. The surjective morphism of automata $\A \twoheadrightarrow \Obs \A$ is the linear transformation that sends some $m \in \A(\St)$ to its equivalence class.

    Finally, notice that the unique morphism $\Ainit(\L) = R^{\Sigma^*} \to \powserw{R} = \Afinal(\L)$ sends a word $w \in \Sigma^*$ -- an element of the basis of $R^{\Sigma^*}$ -- to the derivative $w^{-1} \L$, the series given by $w' \mapsto \L(\triangleright ww' \triangleleft)$. It can be represented by an infinite $\Sigma^* \times \Sigma^*$ matrix, whose $(w,w')$-indexed entry is $\L(\triangleright ww' \triangleleft)$. This matrix is called the \emph{Hankel matrix} of $\L$ in the literature. The image $(\Min \L)(\St)$ of this morphism is thus the submodule of $\powserw{R}$ generated by all the derivatives $\{ w^{-1} \L \mid w \in \Sigma^* \}$, the rows of the Hankel matrix.
    \end{example}

\subsection{Adjunctions}
\label{app:cat:adjunction}

We recall the definition of an \emph{adjunction} and introduce the notations we need. See, e.g.,~\cite[p.~78--81]{maclane1971} for more details.

Let $\C$ and $\D$ be locally small categories. 
An \emph{adjunction} between $\C$ and $\D$ is a pair of functors, $\F \colon \C \to \D$ and $\G \colon \D \to \C$, together with, for any objects $A$ of $\C$ and $B$ of $\D$, a bijection
\[ -^{\sharp} \colon \D(\F A,B) \to \C(A, \G B)\]
that is \emph{natural in both coordinates}, meaning that, for  any $f \in \D(\F A, B)$:
\begin{itemize}
    \item for any $\C$-morphism $c \colon A' \to A$ we have $(f \circ \F c)^\sharp = f^{\sharp} \circ c$, and 
    \item for any $\D$-morphism $d \colon B \to B'$, we have $(d \circ f)^\sharp = \G d \circ f^{\sharp}$. 
\end{itemize}
In this situtation, the functor  $\F$ is called \emph{left adjoint} to $\G$, and the functor $\G$ is called \emph{right adjoint} to $\F$.
\begin{example}
    The forgetful functor $U \colon \Mod{R} \to \Set$ given in \Cref{exa:functors}.\ref{it:forgetful} has a left adjoint: it is the functor $\frMod{R}{-}$ which sends any set $Q$ to the free $R$-module over $Q$, and  any function $f \colon Q \to Q'$ the unique $R$-linear map $\frMod{R}{f} \colon \frMod{R}{Q} \to \frMod{R}{Q'}$ that sends each basis element $q \in \frMod{R}{Q}$ to the basis element $f(q)$ of $\frMod{R}{Q'}$.

    In this case, for any set $Q$ and any module $M$, the function $-^{\sharp} \colon \Mod{R}(\frMod{R}{Q}, M) \to \Set(Q, UM)$ is given by restricting an $R$-linear map $f \colon \frMod{R}{Q} \to M$ to the basis $Q$, giving a function $Q \to UM$. The fact that $-^{\sharp}$ is a bijection is precisely the familiar fact that any $R$-linear map $\frMod{R}{Q} \to M$ is uniquely described by its action on the basis elements. The first naturality property then says that, for any function $c \colon Q' \to Q$, the restriction of $f \circ \frMod{R}{c}$ to the basis $Q'$ is the same function $Q' \to UM$ as the one obtained by first doing $c$, followed by the restriction of $f^{\sharp}$ to the basis $Q$. The second naturality property says that, for any $R$-linear map $d \colon M \to M'$, the restriction of $d \circ f$ to the basis $Q$ is the same function $Q \to UM'$ as the one obtained by first doing the restriction of $f$ to the basis $Q$, followed by the function underlying $d$.
\end{example}

We denote the inverse to $-^{\sharp}$ by $-_{\flat}$. Explicitly, for any objects $A$ of $\C$ and $B$ of $\D$, we have a function \[-_{\flat} \colon \C(A, \G B) \to \D(\F A,B)\] in such a way that $(f^{\sharp})_{\flat} = f$ and $(g_{\flat})^{\sharp} = g$ for any $f \in \D(\F A,B)$ and $g \in \C(A,\G B)$. It follows from the naturality of $-^{\sharp}$ that $-_{\flat}$ is also natural in both coordinates, which in this case means that, for any $g \in \C(A,\G B)$:
\begin{itemize}
    \item for any $\C$-morphism $c \colon A' \to A$ we have $(g \circ c)_{\flat} = g_\flat \circ \F c$, and 
    \item for any $\D$-morphism $d \colon B \to B'$, we have $(\G d \circ g)_{\flat} = d \circ g_\flat$. 
\end{itemize}
The \emph{unit} and \emph{counit} of an adjunction are important natural transformations, defined as follows. For any object $A$ of $\C$, let $\unit_A \coloneqq (\id_A)^{\sharp}$, which is a $\C$-morphism $A \to \G\F A$. For any object $B$ of $\D$, let $\counit_B \coloneqq (\id_B)_{\flat}$, which is a $\D$-morphism $\F\G B \to B$. It follows from the naturality properties of $-^{\sharp}$ and $-_{\flat}$ that $\unit$ is a natural transformation $\id_{\C} \Rightarrow \G \circ \F$, and $\eta$ is a natural transformation $\F \circ \G \Rightarrow \id_{\D}$. 

The following fact is standard, but we use it repeatedly in the categorical proofs in this paper, so we give a proof for the reader's convenience.%
\begin{lemma}\label{lem:trianglelaws}
    Let $\F \colon \C \leftrightarrows \D \colon \G$ be an adjunction, with $-^{\sharp}$, $-_{\flat}$, $\unit$, and $\counit$ as above. 
    \begin{enumerate}
        \item For any $\D$-morphism $f \colon \F A \to B$, we have $f^{\sharp} = \G f \circ \unit_A$. 
        \item For any $\C$-morphism $g \colon A \to \G B$, we have $g_{\flat} = \counit_B \circ \F g$.
    \end{enumerate}
\end{lemma}
\begin{proof}
    Using first the definition of $\unit_A$ and then the second naturality property of $-^{\sharp}$, we have
    \[ \G f \circ \unit_A = \G f \circ (\id_A)^{\sharp} = (f \circ id_A)^{\sharp} = f^{\sharp},\]
    establishing (1). 
    The proof of (2) is similar, using the first naturality property of $-_{\flat}$.
\end{proof}

\section{Proofs for \texorpdfstring{\Cref{sec:categorical-reduction}}{Section 5}}
\label{app:proofs-categorical-reductions}

\lemExtScalars*
\begin{proof} First let us recall that $\TFMod{R}$ and $\Vec{\KK}$ are both equipped with factorization systems consisting of (surjective linear maps, injective linear maps).
  \begin{enumerate}
    \item \Cref{assumption:adjunction}-\labelcref{assumption:adjunction:factorization-system-preservation} first requires that $\localize -$ preserves surjections. Assume $f\colon M\to N$ is surjective and $n/r\in\localize{N}$. There exists $m\in f^{-1}(n)$. Then $\localize{f}(m/r)=n/r$. More abstractly preservation of surjections also follows from $\F$ being a left adjoint. %

    The assumption also requires that if $f: M \to N$ in $\TFMod{R}$ is such that $\localize{f}$ is injective, then $f$ itself is injective. This is true
    because $M$ is torsion-free: if $m \in M$ is such that $f(m) = 0$, then
    $f(m/1) = 0$ hence $m/1 = 0$ and thus $s \cdot m = 0$ for some $s \in R$,
    implying that  $m$ itself is zero.
    
  \item  The right adjoint to $\localize -: \TFMod{R} \to \Vec{\KK}$ is the functor
    $-_R: \Vec{\KK} \to \TFMod{R}$, called \emph{restriction of scalars}, that
    views a $\KK$-vector space as only being an $R$-module and a $\KK$-linear
    transformation as only being $R$-linear. Given an $R$-module $M$ and a
    $\KK$-vector space $V$, notice that an $R$-linear transformation $f: M \to
    V_R$ indeed extends uniquely to a $\KK$-linear transformation $\localize
    M \to V$ by $f(m/r) = 1/r \cdot f(m)$.
  \item  The unique morphism $\unit_M: M \to (\localize M)_R$ that
    extends to the identity $\localize M \to \localize M$ is the
    embedding of $M$ into $\localize M$, sending some $m \in M$ to $m / 1$.
    This embedding is an injection as long as $M$ is torsion-free.
  \item  The identity $V_R \to V_R$ extends uniquely to the
    morphism $\counit_V: \localize V_R \to V$ which sends some $v/r \in
    \localize V_R$ to $1/r \cdot v$. $\counit_V$ is easily seen to be an isomorphism, and hence an injection.  \qedhere
  \end{enumerate}    
\end{proof}

\propMinimalityPreservation*

\begin{proof}
  By definition, $\Min \L$ is the unique (up to
  isomorphism) $(\C,I,O)$-automaton such that, in $\Auto{\L}$, the unique morphism $\Ainit(\L) \to
  \Min \L$ is in $\E_\C$ and the unique morphism $\Min \L \to \Afinal(\L)$ is
  in $\M_\C$. By
  \Cref{assumption:adjunction}-\labelcref{assumption:adjunction:factorization-system-preservation},
  we get, in $\Auto{\F \circ \L}$, an $\E_\D$-morphism $\F \circ \Ainit(\L)
  \to \F \circ \Min \L$ and an $\M_\D$-morphism $\F \circ \Min \L \to \F \circ
  \Afinal(\L)$. (Note that the second part of \Cref{assumption:adjunction}-\labelcref{assumption:adjunction:factorization-system-preservation} entails that $\F$ maps morphisms in $\M_\C$ to morphisms in $\M_\D$.)

  Because $\F$ is a left adjoint by
  \Cref{assumption:adjunction}-\labelcref{assumption:adjunction:right-adjoint},
  there is an isomorphism $\Ainit(\F \circ \L) \cong \F \circ \Ainit(\L)$:
  this can be seen either by the fact that $\Ainit(\F \circ \L)(\St) =
  \coprod_{\Sigma^*} \F I \cong \F \mathopen{}\left(\coprod_{\Sigma^*} I\right)\mathclose{}$ (left adjoints
  preserve coproducts), or more generally by the fact that $\Ainit$ is always
  a left Kan extension \cite[Lemma
  2.9]{petrisanAutomataMinimizationFunctorial2020}, and composing by left
  adjoints preserves left Kan extensions. Hence the unique morphism
  $\Ainit(\F \circ \L) \to \F \circ \Min \L$ is in $\E_\D$.

  It is not true that $\Afinal(\F \circ \L)$ and $\F \circ \Afinal(\L)$
  are isomorphic. Still, the unique morphism $\F \circ \Afinal(\L) \to
  \Afinal(\F \circ \L)$ is given on $\St$ by the canonical morphism \[ \langle \F \pi_w \rangle_{w \in \Sigma^*}: \F\mathopen{}\left( \prod_{\Sigma^*} O \right)\mathclose{} \to
    \prod_{\Sigma^*} \F O \]    
    which can be written
  as the composite
  \[ \F \prod_{\Sigma^*} O \xrightarrow{\F \prod_{\Sigma^*} \unit_O}
    \F \prod_{\Sigma^*} \G\F O \cong \F\G \prod_{\Sigma^*} \F O
    \xrightarrow{\counit_{\prod_{\Sigma^*} \F O}} \prod_{\Sigma^*} \F O \]

  (for the isomorphism in the middle, recall that $\G$, as a right adjoint,
  preserves products).
  
  The proof that this is indeed the morphism we are
    looking for is given by the universal property of the product and the
    following commuting diagram, where the vertical arrows are the projections.
\[\begin{tikzcd}[column sep=scriptsize]
	{\F \prod_{\Sigma^*} F} && {\F \prod_{\Sigma^*} \G \F F} && {\F\G \prod_{\Sigma^*} \F F} && {\prod_{\Sigma^*} \F F} \\
	{\F F} &&& {\F\G\F F} &&& {\F F}
	\arrow["{\F \prod_{\Sigma^*} \unit_F}", from=1-1, to=1-3]
	\arrow[from=1-1, to=2-1]
	\arrow[from=1-3, to=2-4]
	\arrow["\cong"{description}, draw=none, from=1-3, to=1-5]
	\arrow["{\counit_{\prod_{\Sigma^*} \F F}}", from=1-5, to=1-7]
	\arrow[from=1-5, to=2-4]
	\arrow[from=1-7, to=2-7]
	\arrow["{\F \unit_F}", from=2-1, to=2-4]
	\arrow[curve={height=24pt}, Rightarrow, no head, from=2-1, to=2-7]
	\arrow["{\counit_{\F F}}", from=2-4, to=2-7]
\end{tikzcd}\]

Each term in this decomposition is in $\M_\D$, because, from left to right:
$\unit_F$ is in $\M_\C$ by
\Cref{assumption:adjunction}-\labelcref{assumption:adjunction:unit}, $\M_\C$ is
stable under products \cite[Proposition
14.15]{adamekAbstractConcreteCategories2009} and $\F$ sends $\M_\C$-morphisms to
$\M_\D$-ones by
\Cref{assumption:adjunction}-\labelcref{assumption:adjunction:factorization-system-preservation};
$\M_\D$ contains all isomorphisms; $\counit_{\prod_{\Sigma}^* \F O}$ is in
$\M_\D$ by
\Cref{assumption:adjunction}-\labelcref{assumption:adjunction:counit}.

Hence the unique morphism $\F \circ \Afinal(\L) \to \Afinal(\F \circ \L)$ is in
$\M_\D$, and so is the unique morphism $\F \circ \Min \L \to \Afinal(\F \circ
\L)$ as the composite $\F \circ \Min \L \to \F \circ \Afinal(\L) \to \Afinal(\F
\circ \L)$.

It follows that $\F \circ \Min \L$ must be $\Min(\F \circ \L)$ up to
isomorphism.
\end{proof}

\dedekindAlmostStrongFatou*

\begin{proof}
    Let $n$ be the number of states of the minimal $\KK$-weighted automaton recognizing $L$. By \Cref{prop:minimality-preservation}, $\localize{(\Min L)} \cong \KK^n$ and $\Min L$ has rank $n$. Dedekind domains are Noetherian and integral, and thus in particular weak Fatou rings \cite[Chapter 7, Corollary 4.3]{berstelNoncommutativeRationalSeries2010}: this entails that $L$ is computed by a finite $R$-weighted automaton. Applying \Cref{prop:almost-minimal-dedekind-wa}, it follows that there is an $R$-weighted automaton with $n+1$ states computing $L$.
\end{proof}

To properly understand what happens under the hood of \Cref{alg:reduction},
\cite[Lemma 3.4]{petrisanAutomataMinimizationFunctorial2020} can be very useful.
Informally, it says that adjunctions lift to categories of automata. We first spell out how the adjunction gives a correspondence at the level of languages.

\begin{lemma}
\label{lem:language-correspondance}
 Let $\F: \C \to \D$ be left adjoint to $\G: \D \to \C$ and fix two objects $C$ in $\C$ and $D$ in $\D$. Then $(\C,C,\G D)$-languages are
  in one-to-one correspondence with $(\D,\F C,D)$-languages. 
\end{lemma}
\begin{proof}
Recall the natural isomorphism $-^\sharp\colon\D(\F C, D)\to\C(C,\G D)$ and its inverse$-_\flat\colon\C(C,\G D)\to\D(\F C, D)$. 

Let $\L$ be a $(\D,\F C,D)$-language. Then, given a word $w\in\Sigma$, we have a morphism $\L(\triangleright w \triangleleft)\colon \F C \to D$. We define $\L^\sharp$ to be the language defined by $\L^\sharp(\triangleright w \triangleleft) = (\L(\triangleright w \triangleleft))^\sharp$. Conversely, for a $(\C,C,\G D)$-language $\L$, we define $\L_\flat$ to be the $(\D,\F C,D)$-language defined by $\L_\flat(\triangleright w\triangleleft)=(\L(\triangleright w\triangleleft))_\flat$. 
\end{proof}

Through the remaining of this section, we make a  small abuse of notation and use  $\L^\sharp$ and $\L_\flat$ to describe languages obtained via the above correspondence.

For an example of this, consider $\localize{-}\colon \Mod{R} \to \Vec{\KK}$ and,
as seen in \Cref{lemma:extension-of-scalars-left-adjoint}, its right adjoint
$-_R: \Vec{\KK} \to \Mod{R}$. Then $\localize{R} = \KK$ and $\KK_R = \KK$,
and so $(\Mod{R},R,\KK)$-languages are in one-to-one correspondence with
$(\Vec{\KK},\KK,\KK)$-languages: in both cases, we just specify an element
$\L(\triangleright w \triangleleft) \in \KK$ for every word $w \in \Sigma^*$,
and this extends to an $R$-linear map from $R$ or a $\KK$-linear map from $\KK$.

\begin{proposition}[{\cite[Lemma 3.4]{petrisanAutomataMinimizationFunctorial2020}}] Let $\F: \C \to \D$ be left adjoint to $\G: \D \to \C$ and fix two objects $C$ in $\C$ and $D$ in $\D$.
  Let $\L$ be a $(\C,C,\G D)$-language, and let $\L_\flat$ be the corresponding $(\D, \F C, D)$-language. There is a functor $\F_\flat\colon \Auto{\L}
  \to \Auto{\L_\flat}$ with a right adjoint $\G^\sharp: \Auto{\L_\flat} \to \Auto{\L}$.
\end{proposition}

The action of $\F_\flat$ on a $(\C,C,\G D)$-automaton $\A$ and of $\G^\sharp$
on a $(\D,\F C, D)$-automaton $\A'$ are depicted below:
\begin{align*}
  \begin{tikzcd}[ampersand replacement=\&, column sep=large]
    C \& {\A(\St)} \& {\G D}
    \arrow["{\A(\triangleright)}", from=1-1, to=1-2]
    \arrow["{\A(a)}"', from=1-2, to=1-2, loop, in=125, out=55, distance=10mm]
    \arrow["{\A(\triangleleft)}", from=1-2, to=1-3]
  \end{tikzcd}
  && 
  \begin{tikzcd}[ampersand replacement=\&, column sep=large]
    {} \& {}
    \arrow["{\F_\flat}", maps to, from=1-1, to=1-2, line width=.75pt]
  \end{tikzcd}
  &&
  \begin{tikzcd}[ampersand replacement=\&, column sep=large]
    {\F C} \& {\F\A(\St)} \& D
    \arrow["{\F\A(\triangleright)}", from=1-1, to=1-2]
    \arrow["{\F\A(a)}"', from=1-2, to=1-2, loop, in=125, out=55, distance=10mm]
    \arrow["{(\A(\triangleleft))_\flat}", from=1-2, to=1-3]
  \end{tikzcd}&
  \\
  \begin{tikzcd}[ampersand replacement=\&, column sep=large]
    C \& {\G\A'(\St)} \& {\G D}
    \arrow["{(\A'(\triangleright))^\sharp}"', from=1-1, to=1-2]
    \arrow["{\G\A'(a)}", from=1-2, to=1-2, loop, in=235, out=305, distance=10mm]
    \arrow["{\G\A'(\triangleleft)}"', from=1-2, to=1-3]
  \end{tikzcd}
  && 
  \begin{tikzcd}[ampersand replacement=\&, column sep=large]
    {} \& {}
    \arrow["{\G^\sharp}", maps to, from=1-2, to=1-1, line width=.75pt]
  \end{tikzcd}
  &&
  \begin{tikzcd}[ampersand replacement=\&, column sep=large]
    {\F C} \& {\A'(\St)} \& D
    \arrow["{\A'(\triangleright)}"', from=1-1, to=1-2]
    \arrow["{\A'(a)}", from=1-2, to=1-2, loop, in=235, out=305, distance=10mm]
    \arrow["{\A'(\triangleleft)}"', from=1-2, to=1-3]
  \end{tikzcd}&
\end{align*}

The unit and counit of this adjunction respectively have for $\St$-components
$\A(\St) \to \G\F\A(\St) = (\G^\sharp\F_\flat\A)(\St)$ and
$(\F_\flat\G^\sharp\A')(\St) = \F\G\A'(\St) \to \A'(\St)$ the unit and counit of the
adjunction between $\F$ and $\G$.

Recall that $\Auto{\L}$ and $\Auto{\L_\flat}$ are subcategories of $\Auto{\C,C,\G D}$ and $\Auto{\D,\F C, D}$, respectively. The functors $\Fflat$ and $\G^\sharp$ easily extend to an adjunction between these larger categories. These functors of course preserve the correspondence of languages described in  \Cref{lem:language-correspondance}. As an instance of this, when $C=I$ and $D = \F O$ for some $\C$-objects $I$ and $O$, we obtain the following adjunction 
\[
\begin{tikzcd}[row sep=tiny, ampersand replacement=\&]
{\Auto{\D,\F I,\F O}} \arrow[rr, "\G^\sharp"{description}, bend left, out = 15, in = 165, looseness=0.3] 
\&
\&
{\Auto{\C,I,\G\F O}} \arrow[ll, "\Fflat"{description}, bend left, looseness=0.3, in = 165, out =15] %
\end{tikzcd}
\]

Note that $\Fflat$ and $\G^\sharp$, do not act on automata by mere functor composition, as in \Cref{lem:lifting-functors-to-automata}. Adjoint transposes $(-)^\sharp$ and $(-)_\flat$ are used in the  definitions of $\Fflat$ and $\G^\sharp$ to get the input and output object right -- thus justifying the notation  we adopted for these functors.  Nevertheless, the functor $\liftF$ as defined in \Cref{lem:lifting-functors-to-automata} does play a role in \Cref{alg:reduction}. Recall also from \Cref{sec:categorical-reduction} the  identity-on-morphisms functor $\unit_*\colon\Auto{\C,I,O}\to\Auto{\C,I,\G\F O}$ defined as follows:

\begin{align*}
  \begin{tikzcd}[ampersand replacement=\&, column sep=large]
    I \& {\A(\St)} \& {O}
    \arrow["{\A(\triangleright)}", from=1-1, to=1-2]
    \arrow["{\A(a)}"', from=1-2, to=1-2, loop, in=125, out=55, distance=10mm]
    \arrow["{\A(\triangleleft)}", from=1-2, to=1-3]
  \end{tikzcd}
  && 
  \begin{tikzcd}[ampersand replacement=\&, column sep=large]
    {} \& {}
    \arrow["{\unit_*}", maps to, from=1-1, to=1-2, line width=.75pt]
  \end{tikzcd}
  &&
  \begin{tikzcd}[ampersand replacement=\&, column sep=large]
    {I} \& {\A(\St)} \& O \& \G\F O
    \arrow["{\A(\triangleright)}", from=1-1, to=1-2]
    \arrow["{\A(a)}"', from=1-2, to=1-2, loop, in=125, out=55, distance=10mm]
    \arrow["{\A(\triangleleft)}", from=1-2, to=1-3]
    \arrow["\unit_O", from=1-3, to=1-4]
  \end{tikzcd}&
  \end{align*}

The categories of automata appearing in \Cref{alg:reduction} can be summarized as follows.

\begin{proposition}
    \label{prop:functors-in-reduction-alg}
    In the next diagram, $\F_\flat$ is left adjoint to $\G^\sharp$,  $\liftF = \Fflat\circ\unit_*$, but  $\unit_*\neq\G^\sharp\circ\liftF$.
    \[
    \begin{tikzcd}[row sep=tiny, ampersand replacement=\&]
{\Auto{\D,\F I,\F O}} \arrow[rr, "\G^\sharp"{description}, bend left, out = 15, in = 165, looseness=0.3] 
\&[-25pt]
\&[-25pt]
{\Auto{\C,I,\G\F O}} \arrow[ll, "\Fflat"{description}, bend left, looseness=0.3, in = 165, out =15] %
\\[20pt]
\& {\Auto{\C, I, O}} \arrow[lu, bend left, out=15, "\liftF"] \arrow[ru, bend right, out=-15, "\unit_*"'] 
\&                                   
\end{tikzcd}
\]
\end{proposition}
\begin{proof}
    We verify that $\liftF = \Fflat\circ\unit_*$. Given an $(\C,I,O)$-automaton $\A$, $\Fflat\circ\unit_*(\A)$ is the automaton described below
    \[
      \begin{tikzcd}[ampersand replacement=\&, column sep=large]
    \F I \& {\F \A(\St)} \&[10pt] {\F O}
    \arrow["{\F\A(\triangleright)}", from=1-1, to=1-2]
    \arrow["{\F\A(a)}"', from=1-2, to=1-2, loop, in=125, out=55, distance=10mm]
    \arrow["{(\unit_O\circ\A(\triangleleft))_\flat}", from=1-2, to=1-3]
  \end{tikzcd}
    \]

    An easy computation shows that $(\unit_O\circ\A(\triangleleft))_\flat = (\unit_O)_\flat\circ\F\A(\triangleleft) = \F\A(\triangleleft)$, thus proving that the automaton above is in fact $\liftF(\A)$.
\end{proof}

In \Cref{problem:reduction}, we thus start with an automaton $\A$ (in the
top-left corner) recognizing a $(\D,\F I, \F O)$-language $\L$, and we want to
compute some $(\C,I,O)$-automaton $\Min \L'$ (in the bottom corner) such
that $\F \circ \L' = \L$, i.e. such that $\F(\Min \L')$ is equivalent to $\A$.
The strategy of \Cref{alg:reduction} can be summed up as: first compute $\A' =
\Obs \A$, and then compute $\Reach\mathopen{}\left(\G^\sharp\A'\right)\mathclose{}$ in
$\Auto{\L^\sharp}$ (in the bottom-right corner). For efficiency purposes (\Cref{lemma:alg:reduction:second_while_loop_invariant}), this second computation is done in two steps, which correspond to the two \textbf{while} loops. If no counterexample word was produced in the process, the resulting automaton happens to come from the minimal automaton $\Min{\L'}$. This is proved below in \Cref{lem:reach-sharp-obs-is-minimal}.

In the $R$-weighted setting, we consider the localization functor $\F = \localize{-}\colon\TFMod{R}\to\Vec{K}$ and we instantiate the diagram of \Cref{prop:functors-in-reduction-alg} as follows.
\[
\begin{tikzcd}[row sep=tiny, ampersand replacement=\&]
{\cat{Auto}\Big(\Vec{\KK},\KK,\KK\Big)} \arrow[rr, "\G^\sharp"{description}, bend left, out = 30, in = 145, looseness=0.3] 
\&[-35pt]
\&[-35pt]
{\cat{Auto}\Big(\TFMod{R},R,\KK\Big)} \arrow[ll, "\Fflat"{description}, bend left, looseness=0.3, in = 145, out =30] %
\\[20pt]
\& {\Auto{{\TFMod{R}}, R, R}} \arrow[lu, bend left, out=15, "\liftF"] \arrow[ru, bend right, out=-15, "\unit_*"'] 
\&                                   
\end{tikzcd}
\]

The 
functor  $\liftF$  views an $R$-modular automaton as a $\KK$-weighted
one. The diagonal functor $\G^\sharp$ restricts a $\KK$-weighted automaton to an
extended kind of $R$-weighted automaton: its state-space is still a
vector-space, but seen as an $R$-module, and its output weights are still in
$\KK$, not in $R$. Finally, the  functor $\unit_*$ takes an
$R$-weighted automaton and sees it as having output weights in $\KK$.

For the rest of this section we work under the \Cref{assumption:adjunction} which we recall below.

\assumptionAdjunction*

We state some consequences of these assumptions, that will be needed in the following.

\begin{lemma}
\label{lem:easy-prop-from-assumption-adjunction}
Assuming \Cref{assumption:adjunction},  the following hold
\begin{enumerate}
    \item  $m\colon C\to \G D$ is in $\M_\C$ if and only if $m_\flat\colon \F C \to D$ is in $\M_\D$;
    \item  $m\colon \F C\to D$ is in $\M_\D$ in and only if $m^\sharp$ is in $\M_\C$;
    \item if $m\colon D\to D'$ is in $\M_D$, then $\G m\in\M_\C$.
\end{enumerate}
\end{lemma}

\begin{proof}
\begin{enumerate}
    \item Assume $m\colon C\to\G D$ is in $\M_C$. Recall that $m_\flat=\counit_D\circ \F m$. Since $\counit_D\in \M_\D$ (by \Cref{assumption:adjunction:counit}) and  $\F m\in \M_\D$ (by \Cref{assumption:adjunction}-\ref{assumption:adjunction:factorization-system-preservation}), so is their composition $m_\flat$.

    In other direction, assume $m_\flat\in \M_\D$. Using \cite[Proposition
  14.9.(2)]{adamekAbstractConcreteCategories2009}, we infer from the fact that $m_\flat = \counit_D\circ \F m \in \M_D$ and $\counit_D\in \M_\D$ that $\F m\in \M_\D$. By \Cref{assumption:adjunction}-\ref{assumption:adjunction:factorization-system-preservation} we conclude that $m\in\M_\C$.
  \item This follows trivially from the previous item, since $(m_\flat)^\sharp = m = (m^\sharp)_\flat$.
  \item Assume $m\colon D\to D'$ is in $\M_\D$. By the first item, it suffices to show that $(\G m)_\flat$ is in $M_D$. Using the naturality of $\counit$, we can show that $(\G m)_\flat = \counit_{D'}\circ \F\G m = m\circ\counit_D$. Since both $m$ and $\counit_D$ are in $\M_\D$, it follows that so is their composition $(\G m)_\flat$. 
\end{enumerate} 
\end{proof}

\begin{lemma}
    The functor $\Fflat$ also satisfies  \Cref{assumption:adjunction}.
\end{lemma}
\begin{proof}
 The factorizations system we consider on $\Auto{\C,C, \G D}$ and $\Auto{\D,\F C, D}$ are inherited from $\C$, respectively $\D$, in the sense that, for example an $\Auto{\C,C, \G D}$-morphism in $\E_{\Auto{\C,C, \G D}}$ iff its $\St$-component is in $\E_C$. Since $\Fflat$ is defined on morphisms just as $\F$, \Cref{assumption:adjunction}-\ref{assumption:adjunction:factorization-system-preservation} readily follows. For \Cref{assumption:adjunction}-\ref{assumption:adjunction:right-adjoint} recall the functor $\G^\sharp$. Since the units of this adjunction are defined on the $\St$-component as the units of the adjunction $\F\dashv\G$, the remaining two items are also easy to check. 
\end{proof}

Using the above we can show that the right adjoint $\G^\sharp$ preserves observable automata.

\begin{lemma}
\label{lem:Gsharp-pres-obs}
Under the \Cref{assumption:adjunction}, given an automaton $\A$ in $\Auto{\D,\F I, \F O}$ accepting a language $\L$ the $(\C, I, \G\F O)$-automaton  $\G^\sharp(\Obs{\A})$ is observable, that is, the unique morphism  from $\G^\sharp(\Obs{\A})\to\Afinal(\L^\sharp)$ has an underlying $\M_\C$-morphism.
\end{lemma}
\begin{proof}
    This is an immediate consequence of the following two facts: When restricted to the  subcategories $\Auto{\L}\to\Auto{\L^\sharp}$, $\G^\sharp$ is still a right adjoint and hence preserves final objects. That is $\G^\sharp(\Afinal(\L))\cong\Afinal(\L^\sharp)$. We now use the third item of \Cref{lem:easy-prop-from-assumption-adjunction} applied to the $\M_\D$-morphism $\Obs{\A}\rightarrowtail\Afinal{\L}$. It follows that  we have an $\M_\C$-morphism $\G^\sharp(\Obs{\A})\rightarrowtail\G^\sharp(\Afinal{\L})\cong\Afinal(\L^\sharp)$.
\end{proof}

As an immediate corollary  we obtain:

\begin{corollary}
\label{cor:minimal-automaton-in-C}
    Under \Cref{assumption:adjunction}, for any $(\D,\F I,\F O)$-automaton $\A$, the $(\C,I,\G\F O)$-automata $\Reach_\C(\G^\sharp(\Obs_\D(\A)))$ and $\Reach_\C(\G^\sharp(\Reach_\D(\Obs_\D(\A))))$ are minimal, and thus isomorphic. 
\end{corollary}

\begin{proof}
    For the second automaton, recall that $\Reach_\D(\Obs_\D(\A))$ is none other than $\Min \L$, where $\L$ is the language computed by $\A$. But $\Min \L \cong \Obs(\Min \L)$, hence \Cref{lem:Gsharp-pres-obs} applies.
\end{proof}

Similar to \Cref{prop:minimality-preservation}, we can show that the functor $\unit_*$ preserves minimality.

\begin{proposition}
\label{prop:unit-reserves-minimality}
Let $\Min{\L}$ be the minimal $(\C, I, O)$-automaton for a language $\L$. Then $\unit_*(\Min{\L})$ is a minimal automaton for the language $\unit_*(\L)$, defined by $\unit_*(\L)(\triangleright w\triangleleft)=\unit_O\circ\L(\triangleright w\triangleleft)$.
\end{proposition}

\begin{proof}
    Note that $\Min(\St)$ is obtained as the following factorization
    \begin{tikzcd}
        \coprod_{\Sigma^*} I & \Min{\L}(\St) & \prod_{\Sigma^*} O 
        \arrow[two heads, from=1-1, to=1-2]
        \arrow[tail, from=1-2, to=1-3]
    \end{tikzcd}
Note that, by virtue of  \cite[Proposition
14.15]{adamekAbstractConcreteCategories2009}, $\M_\C$ is
stable under products. Since $\unit_O\in\M_\C$, it follows that  the map $\prod_{\Sigma^*}\unit_O\colon \prod_{\Sigma^*} O\to \prod_{\Sigma^*} \G\F O$ is also in $\M_\C$. Therefore, we obtain a factorization 
\[
    \begin{tikzcd}
        \coprod_{\Sigma^*} I & (\Min{\L})(\St) & \prod_{\Sigma^*} O  & \prod_{\Sigma^*} \G\F O 
        \arrow[two heads, from=1-1, to=1-2]
        \arrow[tail, from=1-2, to=1-3]
        \arrow[tail, from=1-3, to=1-4]
    \end{tikzcd}
\]
proving that $(\Min{\L})(\St)$ is also  the state object of the minimal automaton for $\unit_*(\L)$. The rightmost map also underlies a morphism of automata, the unique one $\eta_*(\Afinal(\L)) \to \Afinal(\eta_* \L)$. It follows that $\eta_*(\Min \L)$ is the factorization of the unique morphism $\eta_*\Ainit(\L) \cong \Ainit(\eta_*\L) \to \Afinal(\eta_*\L)$, that is $\Min(\eta_* \L)$.
\end{proof}

\begin{lemma}
\label{lem:reach-sharp-obs-is-minimal}
  For any $(\D,\F I,\F O)$-automaton $\A$ recognizing a language $\L$ and any
  $(\C,I,O)$-language $\L'$, $\L = \F \circ \L'$ if and only if $\Reach\mathopen{}\left(
    \G^\sharp(\Obs \A) \right)\mathclose{}$ is in the image of $\unit_*$, and in that case $\Reach\mathopen{}\left( \G^\sharp (\Obs \A) \right)\mathclose{} \cong \eta_*(\Min \L')$.
\end{lemma}

\begin{proof}
  Assume $\Reach\mathopen{}\left(
    \G^\sharp(\Obs \A) \right)\mathclose{}$ is in the image of $\unit_*$. Then $\Reach\mathopen{}\left(
    \G^\sharp(\Obs \A)\mathclose{}\right)$ and also $\G^\sharp(\Obs \A)$ recognize a language of the form $\unit_*\L'$ for some $(\C, I, O)$-language $\L'$. This entails that $\Obs \A$ recognizes the language $(\unit_*\L')_\flat$. But by assumption $\Obs \A$ also recognizes $\L$, hence for any word $w\in\Sigma^*$ we have
    \begin{align*}
       \L(\triangleright w \triangleleft) &= (\unit_*\L')_\flat(\triangleright w \triangleleft) \\ 
       & =  \counit_{\F O}\circ \F((\unit_*\L')(\triangleright w \triangleleft))\\
       & = \counit_{\F O}\circ \F(\unit_O\circ \L'(\triangleright w \triangleleft)) \\
       & = \counit_{\F O}\circ \F(\unit_O)\circ \F(\L'(\triangleright w \triangleleft))\\
       & = \F(\L'(\triangleright w \triangleleft))
    \end{align*}
    Hence $\L = \F\circ \L'$.
    
 Conversely, assume $\L = \F \circ \L'$. By \Cref{cor:minimal-automaton-in-C}, $\Reach\mathopen{}\left( \G^\sharp (\Obs \A) \right)\mathclose{}$ is the minimal automaton for the language $\L^\sharp$.
 But $\L^\sharp= (\F \circ \L')^\sharp = ((\unit_*\L')_\flat)^\sharp=\unit_*\L'$. On the other hand, 
\Cref{prop:unit-reserves-minimality} establishes that the minimal automaton for $\unit_*\L'$ is $\unit_*(\Min{\L'})$. Hence, $\Reach \mathopen{}\left( \G^\sharp (\Obs \A) \right)\mathclose{} \cong \eta_*(\Min \L')$.
\end{proof}

The correctness of \Cref{alg:reduction}, claimed in
\Cref{thm:alg:reduction:correctness}, follows:

\begin{lemma} \label{lemma:alg:reduction:correct}
  \Cref{alg:reduction} either outputs some $w \in \Sigma^*$ such that
  $\L(\triangleright w \triangleleft) \notin \F[\C(I,O)]$, or there is some
  $(\C,I,O)$-language $\L'$ such that $\L = \F \circ \L$ in which case
  \Cref{alg:reduction} computes an $I$-generating family for and outputs $(\Min
  \L')(\St)$.
\end{lemma}

\begin{proof}
  Again, in the setting of \Cref{alg:reduction}, write $\A' = \Obs \A$. If
  \Cref{alg:reduction} stops and outputs some $w \in \Sigma^*$, then it is clear
  that $\L(\triangleright w \triangleleft) \notin \F[\C(I,O)]$ as $\A'$
  recognizes $\L$.

  Assume now that \Cref{alg:reduction} does not output any word and reaches
  \cref{alg:reduction:line:extract_basis}. Let $W \subseteq \Sigma^*$ be the set
  of generating words obtained at the end of the second \textbf{while} loop (\Cref{alg:reduction:line:end_2nd_while}). 
  
  Then by \cite[Theorem 3.6]{aristoteFunctorialApproachMinimizing2023b}, 
  \[ \mathopen{}\left(\G^\sharp \A'\right)\mathclose{}(\triangleright W) \cong \Reach \mathopen{}\left(\G^\sharp \A'\right)\mathclose{}(\St) \rightarrowtail \G^\sharp(\A')(\St) \]
  For a proof sketch of this, notice that at the end of this second \textbf{while} loop, for every $w \in W$ and $\sigma \in \Sigma$, $\mathopen{}\left(\G^\sharp \A'\right)\mathclose{}(\triangleright W) \cong \mathopen{}\left(\G^\sharp \A'\right)\mathclose{}(\triangleright (W \cup \{ w\sigma \}))$. By \cite[Lemma B.4]{aristoteFunctorialApproachMinimizing2023b}, we deduce that, unsurprisingly $\mathopen{}\left(\G^\sharp \A'\right)\mathclose{}(\triangleright W) \cong \mathopen{}\left(\G^\sharp \A'\right)\mathclose{}(\triangleright (W\cup W\Sigma))$. It then follows by \cite[Proposition B.3]{aristoteFunctorialApproachMinimizing2023b} that 
  \begin{align*}
      \mathopen{}\left(\G^\sharp \A'\right)\mathclose{}(\triangleright W) 
      &\cong \mathopen{}\left(\G^\sharp \A'\right)\mathclose{}\mathopen{}\left(\triangleright W\Sigma^{\le 1}\right)\mathclose{} 
      \\
      &\cong \mathopen{}\left(\G^\sharp \A'\right)\mathclose{}\mathopen{}\left(\triangleright W\Sigma^{\le 2}\right)\mathclose{} 
      \\
      &\cong \ldots 
      \\
      &\cong \mathopen{}\left(\G^\sharp \A'\right)\mathclose{}\mathopen{}\left(\triangleright W\Sigma^{\le k}\right)\mathclose{} 
      \\
      &\cong \ldots 
      \\ 
      &\cong \mathopen{}\left(\G^\sharp \A'\right)\mathclose{}\mathopen{}\left(\triangleright \Sigma^*\right)\mathclose{}
      \\
      &\cong \Reach \mathopen{}\left(\G^\sharp \A'\right)\mathclose{}(\St)
  \end{align*}
  The second-to-last isomorphism relies on the fact that $W$ is prefix-closed, i.e. contains $\varepsilon$ and is such that if $w\sigma \in W$, then $w \in W$ as well.

  Note here that the first \textbf{while} loop (\crefrange{alg:reduction:line:begin_1st_while}{alg:reduction:line:end_1st_while}) does not play a role in the correctness, and could be skipped. It is only there to ensure that \Cref{lemma:alg:reduction:second_while_loop_invariant} holds.

  \bigskip

Write $l'_w: I
  \to O$ to be such that $\L(\triangleright w \triangleleft) = \F l'_w$ for
  every $w \in W$.
  We now show that $\Reach \mathopen{}\left(\G^\sharp\A'\right)\mathclose{}$ is in the image of $\unit_*$. By \Cref{lem:reach-sharp-obs-is-minimal}, this will entail that \Cref{alg:reduction} indeed computes
  $(\Min \L')(\St) \cong \left(\Reach \mathopen{}\left(\G^\sharp\A'\right)\right)\mathclose{}(\St)$.
  
  Because $W$ is an $I$-generating family of $\G^\sharp \A'$, $\mathopen{}\left(\Reach \mathopen{}\left(\G^\sharp\A'\right)\mathclose{}\right)\mathclose{}(\St)$ is obtained as the factorization 
  \[ \coprod_{W} I \twoheadrightarrow \mathopen{}\left(\Reach \mathopen{}\left(\G^\sharp\A'\right)\mathclose{}\right)\mathclose{}(\St) \rightarrowtail \G^\sharp (\A')(\St) \]
  In particular, by \cite[Theorem 3.6]{aristoteFunctorialApproachMinimizing2023b}, $\mathopen{}\left(\Reach \mathopen{}\left(\G^\sharp\A'\right)\mathclose{}\right)\mathclose{}(\triangleleft)$ is obtained as the composite of the right cell in the diagram below:
  \[\begin{tikzcd}[row sep=large]
      {\coprod_W I} & {\mathopen{}\left(\Reach\mathopen{}\left(\G^\sharp\A'\right)\mathclose{}\right)\mathclose{}(\St)} & \G^\sharp (\A')(\St) \\
       & {\G \F O}
      \arrow[two heads, from=1-1, to=1-2]
      \arrow["{\left[\L(\triangleright w \triangleleft)^\sharp\right]_{w \in W}}"',
      curve={height=12pt}, from=1-1, to=2-2]
      \arrow["{\mathopen{}\left(\Reach\mathopen{}\left(\G^\sharp\A'\right)\mathclose{}\right)\mathclose{}(\triangleleft)}"{description}, from=1-2, to=2-2]
      \arrow[tail, from=1-2, to=1-3]
      \arrow["\G^\sharp (\A')(\triangleleft)", curve={height=-12pt}, from=1-3, to=2-2]
    \end{tikzcd}\]
    Because the whole diagram commutes, it follows that the left cell of this same diagram also commutes.

  By assumption, for all $w\in W$ we have $\L(\triangleright w \triangleleft)^\sharp = \G\F l'_w \circ
  \unit_I = \unit_O \circ l'_w$. This left cell can be rewritten as the following commuting square:
  \[\begin{tikzcd}
      {\coprod_W I} & {\mathopen{}\left(\Reach\mathopen{}\left(\G^\sharp\A'\right)\mathclose{}\right)\mathclose{}(\St)} \\
      O & {\G \F O}
      \arrow[two heads, from=1-1, to=1-2]
      \arrow["{[l'_w]_{w \in W}}"', from=1-1, to=2-1]
      \arrow[dashed, from=1-2, to=2-1]
      \arrow["{\mathopen{}\left(\Reach\mathopen{}\left(\G^\sharp\A'\right)\mathclose{}\right)\mathclose{}(\triangleleft)}", from=1-2, to=2-2]
      \arrow["{\unit_O}"', tail, from=2-1, to=2-2]
    \end{tikzcd}\] 
    
 Since $\unit_O$ is in
  $\M_\C$ by
  \Cref{assumption:adjunction}-\labelcref{assumption:adjunction:unit}, we can use the the diagonal fill-in property to deduce that $\left(
    \Reach \G^\sharp\A' \right)(\triangleleft)$ factors through some $\C$-morphism
  $\mathopen{}\left( \Reach\mathopen{}\left(\G^\sharp\A'\right)\mathclose{}\right)\mathclose{}(\St) \to O$: this means precisely that
  $\Reach \mathopen{}\left(\G^\sharp\A'\right)\mathclose{}$ is in the image of $\unit_*$.%
\end{proof}

\reductionCorrect*

\begin{proof}
    We proved the correctness of \Cref{alg:reduction} in \Cref{lemma:alg:reduction:correct}. We now argue that it is indeed the required reduction -- the proof is very similar to the explanation given in \Cref{sec:learning-and-reduction} for $\ZZ$-weighted automata.

    We reduce the problem of active learning $(\C,I,O)$-automata to that of learning $(\D,\F I, \F O)$-automata. Suppose thus given an oracle $\oracle{\C}$ able to answer value and equivalence queries for a $(\C, I, O)$-language $\L$. 

    One can first implement an oracle $\oracle{\D}$ able to answer value and equivalence queries for the $(\D, \F I, \F O)$-language $\F \circ \L$:
    \begin{itemize}
        \item for a value query for the word $w \in \Sigma^*$, ask $\oracle{\C}$ for $\L(\triangleright w \triangleleft)$ and apply $\F$;
        \item for an equivalence query for a $(\D,\F I, \F O)$-automaton $\A$, apply \Cref{alg:reduction} to $\A$: the output is either a counterexample word such that $\A(\triangleright w \triangleleft) \notin \F[\C(I,O)]$ and hence $\A(\triangleright w \triangleleft) \neq \L(\triangleright w \triangleleft)$; or a $(\C,I,O)$-automaton computing $\L$, which can then be given for an equivalence query to $\oracle{\C}$.
    \end{itemize}
    One may now learn a $(\D, \F I, \F O)$-automaton $\A$ recognizing $\F \circ \L$ using the aforementioned $\oracle{\D}$, and then use \Cref{alg:reduction} on $\A$ to construct a $(\C, I, O)$-automaton recognizing $\L$.
\end{proof}

\lemAlgReductionSecondWhileLoopInvariant*

\begin{proof}
  Recall that $m_i$ is obtained as the diagonal fill-in of the following diagram
  \[\begin{tikzcd}
      {\coprod_{W_{i+1}} I} & {\wordImage{\G^\sharp(\A')}{W_{i+1}}} & {\G^\sharp (A')(\St)} \\
      {\coprod_{W_i} I} & \wordImage{\G^\sharp(\A')}{W_{i}} & {\G^\sharp (\A')(\St)}
      \arrow[two heads, from=1-1, to=1-2]
      \arrow["{\mathopen{}\left[\G^\sharp \left(\A'\right)(\triangleright w)\right]\mathclose{}_{w \in W_{i+1}}}", curve={height=-12pt}, from=1-1, to=1-3]
     \arrow[tail, from=1-2, to=1-3]
      \arrow[from=2-1, to=1-1]
      \arrow[two heads, from=2-1, to=2-2]
      \arrow["{\mathopen{}\left[\G^\sharp (\A')(\triangleright w)\right]\mathclose{}_{w \in W_i}}"', curve={height=12pt}, from=2-1, to=2-3]
      \arrow["{m_i}", dashed, tail, from=2-2, to=1-2]
      \arrow[tail, from=2-2, to=2-3]
      \arrow[Rightarrow, no head, from=2-3, to=1-3]
    \end{tikzcd}\] Applying $\F$, which sends $\E_\C$-morphisms to $\E_\D$-ones
  and $\M_\C$-morphisms to $\M_\D$-ones, and post-composing by the
  $\M_\C$-morphism $\counit_{\A'(\St)}: \F\G\A'(\St) \rightarrowtail \A'(\St)$, we
  get that $\F m_i$ is the diagonal fill-in of the diagram
\[\begin{tikzcd}
	{\coprod_{W_{i+1}} \F I} & {\F\mathopen{}\left(\G^\sharp(\A')(\triangleright W_{i+1})\right)\mathclose{}} & {\A'(\St)} \\
	{\coprod_{W_i} \F I} & {\F\mathopen{}\left(\G^\sharp(\A')(\triangleright W_i)\right)\mathclose{}} & {\A'(\St)}
	\arrow[two heads, from=1-1, to=1-2]
	\arrow["{[\A'(\triangleright w)]_{w \in W_{i+1}}}", curve={height=-12pt}, from=1-1, to=1-3]
	\arrow[tail, from=1-2, to=1-3]
	\arrow[from=2-1, to=1-1]
	\arrow[two heads, from=2-1, to=2-2]
	\arrow["{[\A'(\triangleright w)]_{w \in W_i}}"', curve={height=12pt}, from=2-1, to=2-3]
	\arrow["{\F m_i}", dashed, tail, from=2-2, to=1-2]
	\arrow[tail, from=2-2, to=2-3]
	\arrow[Rightarrow, no head, from=2-3, to=1-3]
\end{tikzcd}\]
and in particular $\F\mathopen{}\left(\G^\sharp(\A')(\triangleright W_i)\right)\mathclose{} \cong \A'(\triangleright W_i)$.
Because the first \textbf{while} loop has stopped, $\A'(\triangleright W_0) \cong \A'(\triangleright (W_0 \cup \{ w\sigma \}))$ for every $w \in W$ and $\sigma \in \Sigma$, and so in particular $\F m_0$ is an isomorphism ($W_1$ is obtained from $W_0$ by adding a single word). By \cite[Lemma B.4]{aristoteFunctorialApproachMinimizing2023b}, because every $W_i$ contains $W_0$, we also get for every $i$ that $\A'(\triangleright W_i) \cong \A'(\triangleright (W_i \cup \{ w\sigma \}))$ for every $w \in W$ and $\sigma \in \Sigma$, and so similarly $\F m_i$ is also an isomorphism.
\end{proof}

\section{Algorithmic Framework for Number Rings}
\label{app:algoALGF}

We first recall the Hermite Normal Form of integer matrices, which is key in computing with ideals in number rings.

\subsection{Hermite Normal Form of Integer Matrices}
\label{app:HNF}
  
The  Hermite Normal Form (HNF for short) of integer matrices is analogous to the (column) Echelon Form  for rational matrices. Similarly, 
the HNF has  applications in solving linear systems of equations over integers. Furthermore, it plays a key role in algorithmic algebraic number theory, particularly in  ideals computation,  module theory, and  linear algebra over integers and number rings.

An $n\times m$ integer matrix   in the (column) HNF  can be decomposed as an $n\times r$ zero matrix on the left, concatenated  by an $n \times (m-r)$  integer matrix~$H$ on the right as $
     \begin{pmatrix}
    \bm{0}_{n \times r} \, |\,   H_{n\times (m-r)} 
\end{pmatrix}
$, 
where for $H=[h_{i,j}]$, there is a strictly increasing map 
$f:\{1,\ldots,m-r\} \to \{1,\ldots,n\}$ which assigns each column $j$ of $H$ to a row index $f(j)$, such that
\begin{itemize}
    \item $h_{i,j}=0$ for all $i>f(j)$, %
    \item $h_{f(j),j}\geq 1$,
    \item $0\leq h_{f(i),j} < h_{f(i),i}$ for all $i<j$. 
\end{itemize}

We note in passing that if the matrix $A$ has rank $n$ over~$\QQ$, then 
$H$ will be an $n\times n$ upper triangular matrix with all diagonal entries equal to~$1$; in other words,  $f$ will be the identity map. 
In the general setting, since $f$ is strictly increasing necessarily  $m-r \leq n$ holds.
By~\cite[Theorem 2.4.3]{cohen2013course}, for all $n\times m$ integer matrices~$A$, there exists a unique $n\times m$ matrix~$B$ in HNF such that $B=AU$, where $U$ is an invertible $m\times m$ matrix with $\det(U)=\pm 1$. 
The matrix $H$, formed by the non-zero columns of $B$ as described  above, is called the HNF of~$A$.
The HNF computation of  integer matrices can be done in polynomial time, in the dimension of the matrix and the  bit length of matrix entries written in binary~\cite{domichHermiteNormalForm1987}. 

\medskip

\noindent \textbf{Applications of HNFs.}
Let $M\subseteq \ZZ^n$ be a $\ZZ$-module, and  let $A$ be a $n \times m$ matrix whose columns generates $M$. The columns of the HNF  of $A$ form a unique $\ZZ$-basis of $M$ such that the  associated matrix (of the basis) is in HNF. Denote by $\det$  the matrix determinant. 

\begin{lemma}
\label{lem:HNF}
   Let $M\subseteq \ZZ^n$ be a full-rank $\ZZ$-module, and $A$ an integer  matrix  whose columns generate~$M$. Then $|\ZZ^n/M|=\det(A)$.
\end{lemma}
\begin{proof}
Denote by $H=[h_{i,j}]_{1\leq i,j\leq n}$  the HNF of~$A$. Since $H$ provides a $\ZZ$-basis for~$M$,  the image of $\ZZ^n$  under $H$, denoted by $\ZZ^nH$, is  $M$. The cosets of $\ZZ^nH$ in $\ZZ^n$ are $\bm{v}+\ZZ^nH$  where   the  $i$-th component of $\bm{v}\in \ZZ^n$ is in $\{0,\cdots,h_{i,i}-1\}$. Hence $|\ZZ^n/M|=\det(H)=\det(A)$.     
\end{proof}

The uniqueness of the HNF basis of modules allows to test whether  two input $\ZZ$-submodules of $\ZZ^n$ are equivalent. 
Given the linear system $A \bm{x}=\bm{b}$  of equations where $A$ and $\bm{b}$ have integer entries, clearly,
the variables $\bm{x}$ can be chosen integer if the $\ZZ$-modules generated by  columns of $A$ and by columns of $[A | \bm{b}]$ are equivalent. 
If so, the HNF of~$A$ combined with a simple backward substitution will give a simple way to compute the integer solutions~$\bm{x}$. 

\subsection{Computing in number rings}

The \emph{symbolic representation} of an algebraic number~$\alpha$  consists of its minimal polynomial~$m_{\alpha} \in \QQ[x]$ combined with a triple~$(a,b,R)\in \QQ^3$ such that~$\alpha$ is the unique root of~$m_{\alpha}$ that lies within the $R$-radius circle centered at $(a,b)$ in the complex plane~\cite[Section 4.2.1]{cohen2013course}.
The size of the symbolic representation of~$\alpha$ is the sum of the degree of its minimal polynomial and the bit length of all these numbers written in binary (where rational numbers are expressed as pairs of integers).

We analyze the complexity of our algorithms for a number field~$\KK$ with degree $d$ over~$\QQ$. 
Recall that an integral basis of $\OK$ is 
a $\ZZ$-basis $\{\omega_1,\cdots, \omega_d\}$ of $\OK$ as a $\ZZ$-module (that is,~$\OK=\ZZ \,\omega_1 \oplus \cdots \oplus \ZZ \, \omega_d$) and of $\KK$ as a $\QQ$-vector space (that is,~$\KK=\QQ\, \omega_1 \oplus \cdots \oplus \QQ \, \omega_d$).
In our algorithms, we  
assume that an integral basis $\Omega\coloneq\{\omega_1,\ldots,\omega_d\}$ of the ring of integer~$\OK$, with~$\omega_1=1$ and the $\omega_i$ represented symbolically, is provided a priori. 
For each $\omega_i$, we also assume that its \emph{regular representation}, the
$d\times d$-matrix $M_{i}$ corresponding to the multiplication map by~$\omega_i$ in the integral basis~$\Omega$, is given. 
There are  classical algorithms for the computation of integral bases for ring of integers from the primitive element of the field~\cite[Chapter 6]{cohen2013course}.

Let $\theta$ be a primitive element of $\KK$, that is, $\KK=\adjoinedff{\theta}$. Recall that   $[K:\QQ]=d$.
Then there are exactly $d$ distinct $\QQ$-linear embedding  $\sigma_i \, :\, \KK \to \mathbb{C}$ with  $i\in \{1,\ldots,d\}$.
Moreover, the
elements $\sigma_i(\theta)$ are precisely the $d$ distinct zeros  of the minimum polynomial $m_{\theta}$ of $\theta$. 
Given the  integral basis~$\Omega\coloneq\{\omega_1,\ldots,\omega_d\}$,  
the discriminant of $\KK$ is defined as
\[\Delta_{\KK}\coloneq\det \begin{pmatrix}
    \sigma_1(\omega_1) & \sigma_1(\omega_2) & \cdots & \sigma_1(\omega_d)\\
    \vdots & \vdots & \ddots & \vdots\\
    \sigma_d(\omega_1) & \sigma_d(\omega_2) & \cdots & \sigma_d(\omega_d)
\end{pmatrix}^2.
\]
The discriminant $\Delta_{\KK}$ is independent of the choice of~$\Omega$ and depends only on the number field~$\KK$. 
Define $C_{\KK}\coloneq d^{4} (\log d+ \log \Delta_{\KK})$ to be the \emph{complexity measure} of~$\KK$, where $\log$ is in based $2$. 

Since we use the main procedure in~\cite{biasseComputationHNFModule2017} for the computation of the HNF over $\OK$, we closely follow the algorithmic framework introduced therein. In particular, as required in~\cite{biasseComputationHNFModule2017}, we will need a primitive element~$\theta \in \OK$ to be given a priori, where its minimal polynomial $m_{\theta}=x^d+\sum_{i=0}^{d-1} a_i x^i$ satisfies that
\begin{itemize}
\item $\log(|\mathrm{disc}(m_{\theta})|)\leq C_{\KK}$, where
\[\mathrm{disc}(m_{\theta})\coloneq\prod_{1\leq i\leq j\leq d} |\sigma_i(\theta)-\sigma_j(\theta)|^2\,,\] 
    \item the bit length of $a_i$ is bounded by $C_{\KK}$. 
\end{itemize}
By a variant of the primitive element theorem~\cite{sonn1967theorem}, there exist $c_1, \ldots,c_d\in \{0,1\}$ such that 
$\sum_{i=1}^d c_i \omega_i$, with $\omega_i\in \Omega$, is a primitive element of $\KK$. It is shown in~\cite{biasseComputationHNFModule2017} that among such primitive elements there exists one that satisfies the required condition.

Given the integral basis~$\Omega$, an algebraic number $\alpha$ can be written uniquely as 
$\alpha\coloneq\sum_{1\leq i\leq d} c_i \omega_i$;   we call the vector of coefficients   
$[c_1 , \, \cdots \, , c_d] \in \QQ^d$ the \emph{vector representation} of $\alpha$.
The size of the vector representation of $\alpha$, denoted by $S(\alpha)$, is the sum of 
the bit length of the $c_i$'s. 
Clearly, if $\alpha \in \OK$ its vector  representation is in~$\ZZ^{d}$. 
In our algorithms we  work with vector representations of algebraic numbers.
The complexity measure $C_{\KK}$ of 
is defined such that 
\begin{equation}
\label{eq:norm}
\log(|\norm{\alpha}|) \leq  d (\log(C_{\KK})+S(\alpha))
\end{equation}
for all $\alpha\in \OK$; see~\cite[Inequalities~(1-3) on Page 9]{biasseComputationHNFModule2017} for more details. 

Given the vector representation of algebraic numbers in $\KK$ with respect to~$\Omega$, we adapt ~\cite[Proposition 4]{biasseComputationHNFModule2017} to show that 
for all~$\alpha,\beta\in K$ and $m\in \ZZ$, the following holds
\begin{enumerate}
    \item $S(m\alpha)=d \log(m)+S(\alpha)$,
    \item $S(\alpha\beta)=S(\alpha)+S(\beta)+C$
    \item $S(\frac{1}{\alpha})=d S(\alpha)+C$,
    \item $S(\alpha+\beta) \leq \max(S(\alpha),S(\beta))+d$, if  $\alpha+\beta \in \OK$,
\end{enumerate}
where $C_{\KK}$ is the complexity measure of~$\KK$. We note in passing that the constant $C_{\KK}$ accounts for the effect of~$\Omega$ in the computation. These operations on algebraic numbers can be performed in polynomial time in the complexity measure~$C_{\KK}$.

Following~\cite{biasseComputationHNFModule2017} the fractional ideals~$\mathfrak{a}$ of $\KK$, as  free $\ZZ$-modules of rank~$d$, are represented by a $\ZZ$-basis matrix $N_{\mathfrak{a}} \in \QQ^{d\times d}$. 
For algorithmic purposes, the matrix $N_{\mathfrak{a}}$ is assumed to be in HNF. 
The size of $\mathfrak{a}$ is  the sum of the bit length of all coefficients in~$N_{\mathfrak{a}}$ (plus the overhead due to symbolic representation of the integral basis). 
The sum of two fractional ideals $\mathfrak{a}$ and $\mathfrak{b}$ is defined as $\mathfrak{a} + \mathfrak{b} \coloneq\{ a+b \, |\,  a \in\mathfrak{a} , b \in \mathfrak{b} \}$. It  can be computed as the sum of $\ZZ$-modules, whose HNF representation is obtained by the computation of the HNF of the matrix $[N_{\mathfrak{a}} \, |\, N_{\mathfrak{b}}]$. For integral ideals, the ideal $\mathfrak{a} + \mathfrak{b}$ can be seen as the greatest common divisor of $\mathfrak{a}$  and $\mathfrak{b}$.
The product $\mathfrak{a} \mathfrak{b}\coloneq\{ ab | a \in\mathfrak{a} , b \in \mathfrak{b} \}$ can similarly be computed through the HNF representation of the $\ZZ$-modules of $\mathfrak{a}$ and $\mathfrak{b}$. 
The above ideal operations in $\OK$  can be performed in polynomial time in the complexity measure~$C_{\KK}$.

The norm of an integral ideal~$\mathfrak{a}$, denoted by $\norm{\mathfrak{a}}$, is defined as $|\OK/\mathfrak{a}|$; it can be computed in polynomial time as the absolute value of $\det(N_{\mathfrak{a}})$, where $\det$ denotes the determinant of the input matrix. Since the norm is multiplicative, the norm of $\mathfrak{a}^{-1}$ is $\norm{\mathfrak{a}}^{-1}$. 
Since all fractional ideals can be written as $\mathfrak{a}\mathfrak{b}^{-1}$, where $\mathfrak{a}$ and $\mathfrak{b}$ are integral,
their norm can be computed in polynomial time as well.

\section{Proofs for \texorpdfstring{\Cref{sec:computational}}{Section 6}}
\label{app:modules-numberring}

\begin{algorithm*}[t]
    \caption{Computing generators of the forward module or a counterexample.}
    \label{alg:compute-OK-generators}
    \begin{algorithmic}[1]
        \REQUIRE a $\KK$-weighted automaton $\A=(Q,\A(\triangleright),(\A(\sigma))_{\sigma\in\Sigma},\A(\triangleleft))$ with $m$ states
        \STATE $W = \{ \varepsilon \}$
        \COMMENT{Finding words that increase the rank}
        \WHILE{there is $(w,\sigma) \in  W_B \times \Sigma$ such that $\A(\triangleright w\sigma )  \notin \langle \A(\triangleright u ) \mid u \in W_B \rangle_{\KK}$} \label{alg:compute-OK-generators:line:begin-1st-while}
        \STATE $W \coloneq W \cup \{ w\sigma \}$
        \COMMENT{alternatively,  check whether $\A( \triangleright w \sigma  ) \notin \OK^m$}
        \IF{$\A(\triangleright w \sigma \triangleleft ) \notin \OK$}
        \RETURN $w \sigma$
        \ENDIF
        \ENDWHILE \label{alg:compute-OK-generators:line:end-1st-while}
        \COMMENT{Finding words that augment the module}
        \WHILE{there is $(w,\sigma) \in W_B \times \Sigma$ such that $\A(\triangleright w \sigma)  \notin \langle \A(\triangleright u) \mid u \in W_B \rangle_{\OK}$} \label{alg:compute-OK-generators:line:begin-2nd-while}
        \STATE $W \coloneq W \cup \{ w\sigma \}$
        \COMMENT{alternatively, check whether $\A( \triangleright w\sigma ) \notin \OK^m$}
        \IF{$\A(\triangleright w\sigma \triangleleft ) \notin \OK$}
        \RETURN $w \sigma$
        \ENDIF
        \ENDWHILE \label{alg:compute-OK-generators:line:end-2nd-while}
        \RETURN $W$
    \end{algorithmic}
\end{algorithm*}

\lemisoJJp*

\begin{proof}[Sketch of the proof of \Cref{lem-isoJJp}]
Since for all fractional ideals $\mathfrak{d}$ there exists an integer $r\in \ZZ$ such that $r\mathfrak{d}$ is integral, without loss of generality, we can assume that $\mathfrak{a}$ and $\mathfrak{b}$ are integral.
Recall that two ideals   are coprime if their sum is the whole ring.  
By~\cite[Corollary 1.2.11]{cohen2012advanced}
there exists a non-zero $\alpha \in \KK$ such that $\alpha \mathfrak{a}$ is integral and coprime to $\mathfrak{b}$.
Again, without loss of generality, we can assume that $\mathfrak{a}$ and $\mathfrak{b}$ are coprime, meaning that $\mathfrak{a} + \mathfrak{b}=\OK$. 
By~\cite[Proposition 1.3.12]{cohen2012advanced} we can find elements 
\[a\in \mathfrak{a}  \qquad b\in \mathfrak{b}  \qquad   c\in \mathfrak{b}^{-1}  \qquad d \in \mathfrak{a}^{-1}\]  such that 
$ad-bc=1$. 
The proof follows by observing that 
\[ \begin{pmatrix}
\OK & (\mathfrak{a}  \mathfrak{b})^{-1}
\end{pmatrix} = \begin{pmatrix}
\mathfrak{a}^{-1}  & \mathfrak{b}^{-1}
\end{pmatrix} \begin{pmatrix}
a& c\\
b& d
\end{pmatrix}\,. \]
\end{proof}

\lempseudobasisgen*
\begin{proof}
To achieve this complexity bound,  we use the ideal factor refinement algorithms, introduced for integers in~\cite{bach1993factor}  and generalized to number fields in~\cite{ge1994recognizing}.
The factor refinement of integers computes, for inputs $a_1, \ldots, a_k \in \ZZ$,
 a set of pairwise-coprime factors $m_1,\ldots , m_{\ell}\in  \ZZ$  of the $a_i$’s such that each $a_i$ can be written as a product
 of these factors, that is, $a_i=\prod_{j=1}^{\ell} m_j^{e_j}$ with the $e_j \in \NN$. Denoting by $a =\mathrm{lcm}(a_1, \ldots , a_k)$, the factor refinement algorithm runs in time $\mathcal{O}(\log^2(a))$; see \cite{bach1993factor} and~\cite[Lemma 3.1]{blomer1998probabilistic}. The ideal factor refinement~\cite[Algorithm 5.6]{ge1994recognizing} is a generalization, which, given 
  input ideals $\mathfrak{I}_1, \ldots, \mathfrak{I}_k \subseteq \OK$, 
  computes 
  pairwise-coprime ideals $\mathfrak{m}_1,\ldots , \mathfrak{m_{\ell}} \subseteq \OK$
 such that each $\mathfrak{I}_i$ can be written as a product of these ideal factors, that is, $\mathfrak{I}_i=\prod_{j=1}^{\ell} \mathfrak{m}_j^{e_j}$ with the $e_j \in \NN$.  This algorithm runs in polynomial time in the size of the input ideals in their HNF representation, and in the complexity measure of~$\KK$~\cite[Proposition 5.7]{ge1994recognizing}. We refer the reader to~\cite[Section 5]{ge1994recognizing} and \cite{buchmann1999factor} for more details\footnote{
Algorithm 5.6 in~\cite{ge1994recognizing} is stated for orders of $\OK$: it computes a factor refinement of input ideals of an order $\mathcal{O}$ into a larger order $\mathcal{O'} \subseteq \OK$.  
 }.

\begin{algorithm}[t]
    \caption{Computing an (almost) minimal generating set from a pseudo-basis}
    \label{alg-psudobasis-generator}
    \begin{algorithmic}[1]
        \REQUIRE A pseudo-basis   $\{(\mathfrak{a}_i,\bm{v_i})| 1\leq i\leq \ell\}$ defining an $\OK$-module  $M$ 

        \COMMENT{Iteratively apply a constructive version of \Cref{lem-isoJJp} relying on ideals factor refinement and the Chinese Remainder Theorem}
        \STATE \label{alg-psudobasis-generator:line-compute-generating-set} compute a pseudo-generating set for $M$ in the form  
        \[ \{(\OK,\bm{y_i})| 1\leq i\leq \ell-1\} \, \cup \,  \left\{\left(\prod_{i=1}^{\ell} \,\mathfrak{a}_i, \bm{z}\right)\right\} \]
        for some $\bm{z} \in \OK^n$. 

        \COMMENT{By ideal factor refinement}
        \STATE \label{alg-psudobasis-generator:line-compute-two-generators} find two elements $x_1$ and $x_2$ such that $\prod_{i=1}^{\ell} \mathfrak{a}_i=x_1 \OK + x_2 \OK$. 
        \RETURN the generating set~$\{\bm{y_i} \mid 1\leq i\leq \ell+1\}$  where 
       \[\bm{y_{\ell}} \coloneq x_1 \bm{z} \qquad \text{ and } \qquad \bm{y_{\ell+1}} \coloneq x_2 \bm{z}\]
    \end{algorithmic}
\end{algorithm}

The steps required to carry out the task at hand -- computing a generating set for $M$ of cardinality at most $n+1$ -- are shown in~\Cref{alg-psudobasis-generator}. 
\underline{First}, on \Cref{alg-psudobasis-generator:line-compute-generating-set}, to compute a pseudo-generating set in the form
\[ \{(\OK,\bm{y_i})| 1\leq i\leq n-1\} \, \cup \,  \left\{\left(\prod_{i=1}^n \,\mathfrak{a}_i, \bm{z}\right)\right\} \, ,\] we inductively apply the linear transformation constructed in the proof of~\Cref{lem-isoJJp}. That is, given fractional ideals~$\mathfrak{a}$ and~$\mathfrak{b}$ we find elements 
    \begin{equation}
    \label{eq:abcd}
        a\in \mathfrak{a} \qquad \qquad b\in \mathfrak{b} \qquad \qquad   c\in \mathfrak{b}^{-1} \qquad \qquad d \in \mathfrak{a}^{-1}
    \end{equation}  such that 
$ad-bc=1$.
    Then  $\mathfrak{a} \bm{x}+ \mathfrak{b} \bm{y}=
\OK \bm{x}'+ \mathfrak{ab} \bm{y}'$ where
\[ \begin{pmatrix}
\bm{x}' & \bm{y}'
\end{pmatrix} = \begin{pmatrix}
\bm{x}  & \bm{y}
\end{pmatrix} \begin{pmatrix}
a& c\\
b& d
\end{pmatrix}\,,\]
see~\cite[Corollary 1.3.6]{cohen2012advanced}
for more details.
    The crucial argument in the complexity analysis of this transformation is to show that the elements $a,b,c$ and $d$ satisfying the required conditions can be found in polynomial time. 
Write~$d$ for a common denominator of the generators of $\mathfrak{a}$, and define $\mathfrak{a'}\coloneq d\mathfrak{a}$. Observe that $\mathfrak{a'}$ is integral.
    Repeat this  for $\mathfrak{b}$ and define the integral ideal $\mathfrak{b'}$, analogously.
Next, we adapt~\cite[Lemma 5.2.2]{stein2012algebraic} to find $a \in \mathfrak{a'}$ such that $a \mathfrak{a'}^{-1}$ is integral and coprime to~$\mathfrak{b'}$, but through ideal factor refinement rather than the prime ideal factorization used in the proof therein.  More precisely, by~\cite[Algorithm 5.6]{ge1994recognizing}, we compute pairwise-coprime ideals $\mathfrak{m}_1,\ldots , \mathfrak{m_{\ell}} \subseteq \OK$ such that $\mathfrak{a'}$ and $\mathfrak{b'}$ can be written as products of these ideals.   
 We assume without loss of generality that all factors of $\mathfrak{b'}$ appear among~$\mathfrak{m}_1,\ldots, \mathfrak{m}_r$ for some $r\leq \ell$, that is, 
      \[\mathfrak{a'}\coloneq\prod_{1\leq i \leq \ell} \mathfrak{m}_i^{e_i} \qquad \text{ and } \qquad \qquad \mathfrak{b'}\coloneq\prod_{1\leq i \leq r} \mathfrak{m}_i^{g_i} \,,\]
         with the $e_i,g_i \in \NN$. 
For all $i\in \{1,\cdots,r\}$, let $a_i$ be an element such that $a_i\in \mathfrak{m}_i^{e_i} \setminus \mathfrak{m}_i^{e_i+1}$.
By the generalized Chinese Remainder Theorem~\cite[Theorem 5.1.4]{stein2012algebraic}, there exists an element $a\in \OK$ such that 
\begin{align*}
    a & \equiv a_1 \qquad \pmod{\mathfrak{m}_1^{e_1+1}}\\
    & \vdots\\
    a & \equiv a_r \qquad \pmod{\mathfrak{m}_r^{e_r+1}}\\
    a &\equiv 0 \qquad \pmod{\prod_{r\leq i\leq \ell}\mathfrak{m}_i^{e_i}} 
\end{align*}
Following an identical reasoning to~\cite[Lemma 5.2.2]{stein2012algebraic}, $a\mathfrak{a'}^{-1}$ is integral and coprime to~$\mathfrak{b'}$. 
The procedure of finding $a$ can be done within polynomial time through the HNF representation of ideals~$\mathfrak{m}_1^{e_1}$ and $\mathfrak{m}_1^{e_1+1}$. Moreover,  given the coprime ideals $a \mathfrak{a'}^{-1}$ and ~$\mathfrak{b'}$,  we can find $e\in a\mathfrak{ a'}$ and $b\in \mathfrak{b'}$ such that $e+b=1$. Clearly,
$c=-1$ and $d=e /a$ satisfy \eqref{eq:abcd} and $ad-bc=1$. 
     
\underline{Second}, on \Cref{alg-psudobasis-generator:line-compute-two-generators}, we find two elements $x_1$ and $x_2$ such that $\prod_{i=1}^n \mathfrak{a}_i=x_1 \OK + x_2 \OK$. For this task, a randomized polynomial-time procedure in~\cite[Alg 1.3.15]{cohen2012advanced} is given.
To perform this task in polynomial time, we follow~\cite[Proposition  5.2.3]{stein2012algebraic} and again rely on ideal factor refinement. 
Given an an element~$x_1$ of $\prod_{i=1}^n \mathfrak{a}_i$, 
we find $x_2 \in I$ such that $x_2(\prod_{i=1}^n \mathfrak{a}_i)^{-1}$ is integral and coprime to $x_1 \OK$ (we follow the detailed procedure explained above for \Cref{alg-psudobasis-generator:line-compute-generating-set}). 
Clearly, the equality $\prod_{i=1}^n \mathfrak{a}_i = x_1 \OK + x_2 \OK$ holds. 

\end{proof}

\lemlengthdchain*
\begin{proof}
By the properties of the pseudo-HNF of~$A$ and of ideals~$\mathfrak{a}_i$, the ideal
$\mathfrak{d}$ is integral. 
Below, we argue that $|\OK^n/M|=|\OK / \mathfrak{d}|$. Having shown this, the definition of ideal norm gives that 
$|\OK^n/M|=\norm{\mathfrak{d}}$. Since $\ZZ^n$ is commutative, all its subgroups are normal, and we may in particular consider the group quotients $\OK^n/M_i$. By the third isomorphism theorem,
 each $\OK^n/M_i$ is a subgroup of $\OK^n/M_{i+1}$. By Lagrange’s Theorem, the order of a subgroup of a finite group is always a
divisor of the order of the group, so that $|\OK^n/M|$ strictly divides $|OK^n/M_2|$, which strictly divides in turn $|\OK^n/M_3|$, etc. all the way up to $|\OK^n/M_k|$ (the divisibility is strict because the inclusions are strict). It follows that $|\OK^n/M| \ge 2^k$, and the claimed bound follows.

It remains to prove that $|\OK^n/M|=|\OK / \mathfrak{d}|$. 
Using the pseudo-HNF, there exists $U$ and fractional ideals~$\mathfrak{c}_1, \ldots, \mathfrak{c}_n$ such that $AU=H$ and 
$M=\mathfrak{c}_1 \bm{h}_1 \oplus \ldots \oplus \mathfrak{c}_n \bm{h}_n$  where the $\bm{h_i}$ is the $i$-th column of $H$.
 Furthermore,  $\prod_{i=1}^n \mathfrak{c}_i=\mathfrak{d}$ (recall ~\cite[Definition 1.4.8 and Theorem 1.4.9]{cohen2012advanced}). 
Recall from the proof of~\Cref{lem-isoJJp} that, for all fractional ideals $\mathfrak{a}$ and $\mathfrak{b}$, 
there are elements 
\[a\in \mathfrak{a}  \qquad b\in \mathfrak{b} \qquad    c\in \mathfrak{b}^{-1}  \qquad d \in \mathfrak{a}^{-1}\]  such that 
$ad-bc=1$. Then  $\mathfrak{a} \bm{x}+ \mathfrak{b} \bm{y}=
\OK \bm{x}'+ \mathfrak{ab} \bm{y}'$ where
\[ \begin{pmatrix}
\bm{x}' & \bm{y}'
\end{pmatrix} = \begin{pmatrix}
\bm{x}  & \bm{y}
\end{pmatrix} \begin{pmatrix}
a& c\\
b& d
\end{pmatrix}\,,\]
see~\cite[Corollary 1.3.6]{cohen2012advanced}
for more details. In particular, extending this inductively, there are vectors 
$\{\bm{w_i} |  1\leq i\leq n\}$ such that the
module $M$ can be written as $\OK \bm{w_1} \oplus \cdots \oplus \OK \bm{w_{n-1}} \oplus \mathfrak{d} \bm{w_{n}}$ and hence $|\OK^n/M|=|\OK / \mathfrak{d}|$ indeed.  
\end{proof}

\complexAlg*

\begin{proof}
Let $\A=(Q,\A(\triangleright),(\A(\sigma))_{\sigma\in\Sigma},\A(\triangleleft))$ be the input automaton to these algorithms. 
We analyze the complexity of the algorithms under the assumption that the
entries of $\A(\triangleright)$, $\A(\triangleleft)$ and $\A(\sigma)$ for $\sigma\in \Sigma$ are given in their vector representations with respect to the fixed integral basis $\Omega=\{w_1=1, \ldots, w_d\}$ of $\OK$. 
Write all fractions in these vector representations over a common dominator. Consider the max of the numerators and the denominator of the resulting fractions, and denote by $B_\A$ its bit length, which is polynomially bounded in the size (of the vector representation) of $\A$.

Recall that we can assume that the forward-module of the automaton $\A$ is full-rank inside $\A(\St)$, because minimization is a polynomial-time procedure. This is important because it then allows us to restrict to the use of full-rank pseudo-HNFs: the non-full-rank pseudo-HNFs are not proven to be computable in polynomial-time in \cite{biasseComputationHNFModule2017}, they are only claimed so.

\bigskip

\noindent \textbf{Complexity Analysis of  \Cref{alg:compute-OK-generators}:} Recall that we write $m$ for the numbers of states of $\A$. We consider the entries of $\A(\triangleright)$ and the $\A(w)$'s, where $w$ ranges over all words in $\Sigma^*$, as integers in~$\OK$ divided by a common rational integer in $\ZZ$. By a simple induction on the length of the words, for $\A(\triangleright w) = \frac{1}{q}[a_1 \cdots a_n]$, 
\begin{itemize}
    \item  the bit length of  the common  denominator $q\in \ZZ$ is  bounded by $B_\A(|w|+1)$, and 
    \item  $S(a_i)\leq (|w|+1)(B_\A+C_{\KK}+md)$ for all $i\in \{1,\ldots,n\}$,
\end{itemize}
 where $C_{\KK}$ is the complexity measure of $\KK$.

Let $W$ be the set of words found in the first \textbf{while} loop on \Crefrange{alg:compute-OK-generators:line:begin-1st-while}{alg:compute-OK-generators:line:end-1st-while}, so that $\{\A(\triangleright w) | w\in W\}$ is a basis for the forward $\KK$-vector space of $\A$. Clearly, the dimension of the forward $\KK$-vector space is at most~$m$ and thus 
the first \textbf{while} loop iterates at most~$m$ times. Moreover, 
 every word in $W$ has length at most~$m-1$. 
 At the end of the loop (if it wasn't interrupted) the entries of $\A(\triangleright w)$ for $w\in W$ are in $\OK$.
 As computed above, $S((\A(\triangleright w)_i)$ is bounded by 
 $m(B_\A+C_{\KK}+md)$.

Let $W=\{w_1=\varepsilon,\ldots, w_\ell \}$ for some $\ell \leq m$. Define $M\coloneq\OK \, \A(\triangleright w_1) \oplus \ldots \oplus \OK \, \A(\triangleright w_\ell)$.
Let $A$ be the $\ell \times m$ matrix whose $i$-th column is $\A(\triangleright w_i)$.
Let $g$ be the sum of all $\ell \times \ell$ minor ideals of $A$. 
By~\Cref{lem:lengthdchain}, all strictly increasing chains of $\OK$-modules $M_1\subset M_2 \subset \cdots  \subset M_k$
    with $M=M_1$ 
have length $k$ bounded by the number of prime divisors of $\norm{g}$.
The entries $a_{i,j}$ of $A$ are in $\OK$, and $S(a_{i,j})$ is bounded by 
 $m(B_\A+C_{\KK}+md)$, 
as computed above. 
Given an $\ell \times \ell$ matrix with  such entries the determinant $D$ can be computed as a summation of $\ell$
terms, each being a product of $\ell$ distinct matrix entries. To avoid bit-explosion, we perform the $\ell!$ summation in a binary-tree approach:
first, summations are done in pairs; then, the sums of these pairs are recursively summed to form the parent nodes. This process continues until a single summation (the determinant) is computed.
Then \[S(D) \leq d \ell \log \ell +\ell^2(B_\A+C_{\KK}+\ell d)+ \ell C \leq d\ell^2(B_\A+2C_{\KK}+2\ell d)\,.\]
The greatest common   divisor $g$ of all $\ell \times \ell$ minors of $A$  also satisfies $S(g) \leq d\ell^2(B_\A+2C_{\KK}+2\ell d)$.
By inequality~\eqref{eq:norm},
\[\log(|\norm{g}|) \leq  d (\log(C_{\KK})+S(g))\,\]
which is polynomial in the complexity measure~$C_{\KK}$ of the number field, 
the number of states $m$ of the input automaton, and the bit length~$B_\A$ obtained from the vector representation of  entries  in $\A(\triangleright)$, $\A(\triangleleft)$ and the $\A(\sigma)$  over a common dominator.  
Hence, the second \textbf{while} loop  at Line 11 iterates polynomial many times in 
our input size.

The manipulation of the algebraic numbers throughout the algorithms is in polynomial time in $C_{\KK}$, $B_\A$ and $m$.
It remains to  analyze the complexity of the test whether 
$\A(\triangleright w \sigma) \notin \langle \A(\triangleright u) \mid u \in W\rangle_{\OK}$ at Line 11. 
For this, we test whether two modules
\[M\coloneq\bigoplus_{u\in W} \OK \, \A(\triangleright u) \qquad \text{and} \qquad 
N\coloneq M + \OK \, \A(\triangleright w\sigma) \]
are equivalent. 
This test can be performed through a pseudo-HNF computation. Indeed, if $M \subsetneq N$ then~$|\OK^m/M|$ is strictly larger than~$|\OK^m/N|$. Now, on the one hand, the forward module of $\A'$ is full-rank (inside $\A'(\St) \cong \OK^m)$, and, on the other hand, the rank of the module spanned by the vectors corresponding to words in $W$ at the end of the first \textbf{while} loop is full-rank (inside the forward module of $\A'$). It follows that the two modules $M$ and $N$ above are full-rank, and that we can use a full-rank pseudo-HNF to check this: this is thus a polynomial-time operation, as proved in the main Theorem of~\cite{biasseComputationHNFModule2017}. 
This concludes the complexity analysis of \Cref{alg:compute-OK-generators}.

\bigskip

\noindent \textbf{Complexity Analysis of \Cref{alg!make-automaton-integral}:}
Let $n$ be the number of states of the input~$\A$. 
The manipulation of the algebraic numbers throughout the algorithms is in polynomial time in $C_{\KK}$, $B_\A$ and $n$.
The computation of the pseudo-basis on \Cref{alg!make-automaton-integral!line!pseudo-basis} can be done through a full-rank pseudo-HNF computation (because the pseudo-basis is that of the full-rank forward module of $\A'$), which was placed in polynomial time in the main Theorem of~\cite{biasseComputationHNFModule2017}. By~\Cref{lem:pseudobasisgen} a generating family of cardinality of at most $n+1$ can be constructed from the pseudo-basis in polynomial time. The claimed complexity bound for \Cref{alg!make-automaton-integral} follows.

\end{proof}

\decidingMinimalityPIPHard*

Recall that for an integral domain $R$ and two $R$-modules $M \subseteq N$, $M$'s saturation into $N$, written $\sat_N M$, is the intersection $\localize M \cap N$ (within $\localize N$). $M$ is \emph{saturated} into $N$ when $\sat_N M = M$, or, in other words, when for every $n \in N$ such that $\lambda n \in M$ for some non-zero $\lambda \in R$, then $n \in M$ as well.

\begin{proof}
    Given an ideal $\mathfrak{a} \subseteq \OK$, we construct in polynomial time an $\OK$-WA that is state-minimal if and only if $\mathfrak{a}$ is not principal.

    By \Cref{lem:pseudobasisgen}, compute first, in polynomial time a size-$2$ generating set $\{x,y\}$ of $\mathfrak{a}$, so that $\mathfrak{a} = x \OK + y \OK$. Compute moreover, again in polynomial time, $\mathfrak{a}^{-1}$ \cite[Proposition 13]{biasseComputationHNFModule2017} and the isomorphism $\mathfrak{a} \oplus \mathfrak{a}^{-1} \cong \OK \oplus \OK$ (as shown to be possible in the proof of \Cref{lem:pseudobasisgen}). Let $\begin{pmatrix}
        x_1 \\ x_2
    \end{pmatrix}$ and $\begin{pmatrix}
        y_1 \\ y_2
    \end{pmatrix}$ be the respective images of $x$ and $y$ under this isomorphism. Construct finally the $\OK$-WA $\A$ on the alphabet $\Sigma = \{a,b\}$ depicted in \Cref{fig:ok-wa}.
    \begin{figure}[h]
        \centering
        \includegraphics[width=0.25\linewidth]{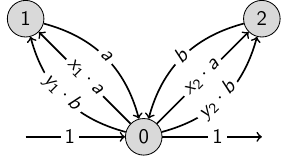}
        \caption{An $\OK$-WA encoding the ideal $\mathfrak{a}$.}
        \label{fig:ok-wa}
    \end{figure}

    Let $L\in\powser{\OK}{a,b}$ be the language recognized by $\A$. Writing $L_1$ and $L_2$ for the languages recognized by $\A$ starting from states $1$ and $2$, we have \begin{align*}
        a^{-1} L &= x_1 L_1 + x_2 L_2
        &b^{-1} L = y_1 L_1 + y_2 L_2
    \end{align*}
    Moreover $L_1$ is zero on $b\{a,b\}^*$ while $L_2$ is zero on $a\{a,b\}^*$, hence $L_1$ and $L_2$ are linearly independent and \[ \OK \cdot (x_1 L_1 + x_2 L_2) + \OK \cdot (y_1 L_1 + y_2 L_2) \cong \mathfrak{a} \]
    Note that $L$ is zero on words of odd length, so that $\OK \cdot L \cap (\OK \cdot a^{-1} L + \OK \cdot b^{-1} L) = \{0\}$. Furthermore, $w^{-1}L \in \OK \cdot L$ for any $w \in \{aa,ab,ba,bb\}$. It follows that the module spanned by the rows of the Hankel matrix of $L$ is the direct sum 
    \[ \OK \cdot L \oplus [\OK \cdot (x_1 L_1 + x_2 L_2) + \OK \cdot (y_1 L_1 + y_2 L_2)] \] 
    which is in turn isomorphic to $\OK \oplus \mathfrak{a}$. Therefore, $\Min L$ has rank $2$. 

    We now show that this module, $\Min L$, is saturated into $\powser{\OK}{a,b}$. We have that $\Min L \subseteq \OK \cdot L \oplus \OK \cdot L_1 \oplus \OK \cdot L_2$.  Now $\OK \cdot L \oplus \OK \cdot L_1 \oplus \OK \cdot L_2$ is saturated into $\powser{\OK}{a,b}$: if $\lambda f = \alpha L + \beta L_1 + \gamma L_2$ where $\lambda, \alpha, \beta, \gamma \in \OK$ and $\lambda$ is non-zero, then it is easy to see that $f = f(\varepsilon) L + f(a)L_1 + f(b) L_2$. We therefore only need to show that $\Min L$ is saturated into $\OK \cdot L \oplus \OK \cdot L_1 \oplus \OK \cdot L_2$. But under the isomorphism $\mathfrak{a} \oplus \mathfrak{a}^{-1} \cong \OK \cdot L_1 \oplus \OK \cdot L_2$, this is equivalent to saying that $\OK \oplus \mathfrak{a}$ is saturated into $\OK \oplus \mathfrak{a} \oplus \mathfrak{a}^{-1}$, which is obviously true since $\OK \oplus \mathfrak{a}$ and $\mathfrak{a}^{-1}$ are in direct sum there. $\Min L$ is thus indeed saturated into $\powser{\OK}{a,b}$.

    We now have all the ingredients to conclude. Any $2$-state automaton $\B$ computing $L$ must be such that $\Min L \subseteq \B(\St) \subseteq \sat_{\powserw{\OK}}(\Min L)=\Min L$ by \Cref{lemma:sandwich-minimal-automata} below, and hence $\OK^2 = \B(\St) \cong \Min L \cong \OK \oplus \mathfrak{a}$. This holds only if $\mathfrak{a}$ is principal, and hence such a $\B$ may thus only exists when $\mathfrak{a}$ is principal. Conversely if $\mathfrak{a}$ is principal the minimal $\OK$-modular automaton recognizing $L$, having state-space $\OK \oplus \mathfrak{a} \cong \OK^2$, is in fact a $2$-state $\OK$-WA computing $L$.
\end{proof}

\begin{lemma}
    \label{lemma:sandwich-minimal-automata}
    Let $R$ be an integral domain with field of fractions $\KK$, and let $L: \Sigma^* \to R$ be a function whose minimal $\KK$-weighted automaton has $n$ states. If $\B$ is an rank-$n$ $R$-modular automaton computing $L$ (this includes in particular $n$-state $R$-weighted automata), then $\Min L \subseteq \B(\St) \subseteq \sat_{\powserw{\OK}}(\Min L)$.
\end{lemma}

\begin{proof}
    $\localize \B$ is (by definition of the rank) an $n$-state $\KK$-weighted automaton computing $L$, and is thus in fact the minimal $\KK$-weighted automaton computing $L$. The image of the morphism of automata $\B \to \Afinal(L)$ by the functor $\localize -$ must thus be an injection (since composed with the canonical $\KK$-linear transformation $\localize(\powserw{\OK}) \to \powserw{\KK}$, it yields the injective morphism $\localize \B \to \Afinal(\localize L)$) and hence so must be the morphism $\B \to \Afinal(L)$ itself since $\localize -$ reflects injections. It follows that \[ \Min L \subseteq \B(\St) \subseteq \powserw{\OK} \]
    Because $\Min L$ and $\B(\St)$ have the same rank, $\localize(\Min L) \cong \localize(\B(\St))$, and so 
    \[ \B(\St) \subseteq \sat_{\powserw{\OK}}(\B(\St)) = \sat_{\powserw{\OK}}(\Min L) \qedhere \]
\end{proof}

\end{document}